\documentclass[11pt]{article} 

\usepackage{amsfonts, amssymb, latexsym, amsmath, amsthm, xcolor, mathtools}

\renewcommand{\theequation}{\arabic{section}.\arabic{equation}}

\def\R{\mathbb R}
\def\N{\mathbb N}

\def\supp{\mathrm{supp}\,}

\def\p0{\langle v \rangle}

\begin{document}

\sloppy

\newtheorem{theorem}{Theorem}[section]
\newtheorem{definition}[theorem]{Definition}
\newtheorem{proposition}[theorem]{Proposition}
\newtheorem{cor}[theorem]{Corollary}
\newtheorem{lemma}[theorem]{Lemma}
\newtheorem*{comment}{Comment}
\theoremstyle{remark}
\newtheorem{remark}[theorem]{Remark}

\renewcommand{\theequation}{\arabic{section}.\arabic{equation}}

\title{Oppenheimer-Snyder type collapse for a collisionless gas}

\author{H{\aa}kan Andr\'{e}asson\\
        Mathematical Sciences\\
        Chalmers University of Technology\\
        G\"{o}teborg University\\
        S-41296 G\"oteborg, Sweden\\
        email: hand@chalmers.se\\
        \ \\
        Gerhard Rein\\
        Fakult\"at f\"ur Mathematik, Physik und Informatik\\
        Universit\"at Bayreuth\\
        D-95440 Bayreuth, Germany\\
        email: gerhard.rein@uni-bayreuth.de}

\maketitle

\begin{abstract}
  In 1939, {\sc Oppenheimer} and {\sc Snyder} showed that
  the continued gravitational collapse of a self-gravitating matter
  distribution can result in the formation of a black hole, cf.~ \cite{OS}.
  In this paper, which has greatly influenced the evolution of ideas
  around the concept of a black hole, matter was modeled as dust, a fluid
  with pressure equal to zero. We prove that when the corresponding initial
  data are suitably approximated by data for a collisionless gas
  as modeled by the Vlasov equation, then a trapped surface forms
  before the corresponding solution to the Einstein-Vlasov system
  can develop a singularity and again a black hole arises. As opposed to the
  dust case the pressure does not vanish for such solutions.
  As a necessary
  starting point for the analysis, which is carried out in
  Painlev\'{e}-Gullstrand coordinates, we prove a local existence
  and uniqueness theorem for regular solutions together with a corresponding
  extension criterion. The latter result will also become useful when
  one perturbs dust solutions containing naked singularities in the Vlasov
  framework.
\end{abstract}
\section{Introduction}
\setcounter{equation}{0}
It is well known that in the context of General Relativity the gravitational
collapse of some matter distribution can result in the formation of a spacetime
singularity where the structure of spacetime breaks down. Historically the first
example where this was shown is the seminal paper \cite{OS} by
{\sc J.~R.~Oppenheimer} and {\sc H.~Snyder}.
The matter model employed in their analysis was a perfect, compressible
fluid with pressure identically equal to zero, a model often referred to
as dust; the initial data consist of a homogeneous
ball of such dust. Since such matter cannot
build up any force acting against being compressed
ad infinitum by its own gravity, 
it is not surprising that the dust ball collapses.
The ground-braking insight was what this entails for the metric
and causal structure of spacetime: A spacetime singularity
forms. The comforting aspect of the analysis is that this
singularity is hidden behind
an event horizon and cannot be observed from far away.
The Oppenheimer-Snyder solution
is therefore an example where the weak cosmic censorship hypothesis, which
was later formulated by {\sc Penrose} \cite{pen}, holds.

But the dust matter model can be criticized as being somewhat peculiar
and unrealistic, since usually matter will build up a force going counter
to the compression. In addition, it is known that for dust and a class of
inhomogeneous data,
gravitational collapse can also result in naked singularities
which are not hidden behind an event horizon and are forbidden by the
weak cosmic censorship hypothesis \cite{chrnd}.

In the present paper we 
investigate a situation analogous to the Oppenheimer-Snyder
set-up, but with a collisionless gas as matter
model instead of dust, i.e., we try to prove on analogue
to the Oppenheimer-Snyder collapse for the Einstein-Vlasov system;
the system is formulated in Section~\ref{section-ev}.
Besides the fact that a self-gravitating collisionless gas is often used in
astrophysics to model objects like galaxies or globular clusters,
cf.~\cite{BT},
there are several more specific reasons which motivate this analysis.
Firstly, dust can be viewed as a
singular special case of Vlasov and initial dust data can in
a precise sense be approximated
by Vlasov data. Secondly, Vlasov matter has, when viewed macroscopically,
non-trivial pressure which may in general prevent collapse;
for example, there exist plenty of steady states
with Vlasov matter, but there can exist none with dust;
for Vlasov matter the pressure in general is anisotropic and also
the current is non-trivial.
Thirdly, so far no naked singularities
have been shown to exist with Vlasov matter.
Finally, dust produces singularities also in the Newtonian
regime, but with Vlasov solutions exist globally in the Newtonian
case. In \cite{RT} it was analyzed what happens if Vlasov data
are pushed towards dust data in the Newtonian case. That paper
was supposed to serve as a blueprint for the present analysis,
the main result of which is the following theorem, which we state
here in a somewhat vague fashion; the precise version follows below,
cf.~Theorem~\ref{maine}.

\begin{theorem}\label{maini}
  For regular initial data which approximate Oppenheimer-Snyder
  data in a suitable way 
  the corresponding solution to the Einstein-Vlasov system approximates the
  Oppenheimer-Snyder dust solution arbitrarily well. In particular, it forms
  a trapped surface before
  it can form a spacetime singularity so that just as
  in the Oppenheimer-Snyder case
  the singularity is hidden behind an event horizon.
\end{theorem}

We should add that the Hawking-Penrose singularity theorem \cite{pen}
implies that the spacetime obtained in Theorem~\ref{maini} must be geodesically
incomplete, but this does not necessarily imply that a singularity
occurs in the sense that curvature-related scalar quantities blow up,
as they do for the Oppenheimer-Snyder solution. Whether this is the case
is an interesting topic for further research.

Theorem~\ref{maini} may
at first glance look like a result on continuous dependence on initial data.
However, there exists no mathematical framework in which both the Einstein-dust
system and the Einstein-Vlasov system are well-posed. As a matter of fact,
the result we prove gives a much more specific relation between the
solutions of the two systems: In a core region the Vlasov solution
is equal to a homogeneous solution of the Einstein-Vlasov system,
and this homogeneous solution can in turn be compared to the dust solution.
This more specific relation is then accessible to a rigorous proof.
The whole analysis is restricted to a spherically symmetric, asymptotically
flat situation and employes Painlev\'{e}-Gullstrand coordinates;
the Einstein-Vlasov system in these coordinates is formulated in
Section~\ref{section_ev_pg}.
A necessary prerequisite of the analysis is a local existence
and uniqueness result for the initial value for the
Einstein-Vlasov system in Painlev\'{e}-Gullstrand coordinates, 
together with a continuation criterion which tells us exactly which
components of the solution must be controlled in order to extend it.
This result will also become important in the analysis of
Vlasov-perturbed dust solutions with naked singularities, cf.~\cite{AR_naked}. 
The local existence and uniqueness result is formulated and proven in
Sections~\ref{section_locex_form} and \ref{section_locex_proof}.
We note that comoving coordinates, which are often used in
the dust case, do not seem to be suitable for proving a local
existence result for the Einstein-Vlasov system due to problems at the
center.
In Section~\ref{sec_hom} we review homogeneous solutions in the dust
case, construct such solutions in the Vlasov case, and establish various
relations between these solutions. Section~\ref{section_OScollapse} is
then devoted to the proof of the precise version of Theorem~\ref{maini},
cf.~Theorem~\ref{maine}.

Before we proceed some further references to the literature seem in order.
In 1965, {\sc Penrose} \cite{pen} showed that once a spacetime develops
a trapped surface then it must become singular in the sense of geodesical
incompleteness. Since trapped surfaces are stable under small perturbations
of the spacetime this result extends the Oppenheimer-Snyder
singularity formation to
more general spacetimes, at the price of a rather weak definition
of ``singularity''. In addition, this result brings up the question
of which data result in the formation of trapped surfaces;
for a long time, the zero pressure
Oppenheimer-Snyder solution was the only
known example.

The present investigation is by no means the first attempt to extend the
Oppenheimer-Snyder result to a matter model with non-zero pressure.
For the fluid case this was done in \cite{ST94}. In the case of a collapsing
fluid ball a major difficulty, already in the Oppenheimer-Snyder
case, arises at the interface between matter and vacuum which is treated as
a shock wave in \cite{ST94} and across which two different coordinate systems
need to be matched in \cite{OS}. No such difficulties arise in our approach.
We use a single coordinate system to cover the whole
spacetime---Painlev\'{e}-Gullstrand coordinates---and rely on a general
local existence
result in which the metric is as regular at the boundary of the matter as
anywhere else; the boundary takes care of itself.

In 1984, {\sc Christodoulou} \cite{chrnd} showed that
with dust as matter model solutions can 
develop naked singularities which violate weak cosmic censorship.
Later he showed that naked singularities occur also
for a self-gravitating scalar field,
but that for this ``matter'' model weak cosmic censorship does hold
in a suitable sense, cf.~\cite{Chr86,Chr87,Chr91,Chr93,Chr94,Chr99a,Chr99}.

Gravitational collapse in the context of the Einstein-Vlasov system
has previously been investigated in
\cite{AAR1,AAR2,An2,AKR10,AKR11,AR06,AR10,AR10b,OChop,RRS},
both numerically and analytically. But for the analytic results the data were
specifically tailored to lead to collapse and in particular were far from
Oppenheimer-Snyder type data or ``real'' astrophysical situations;
the initial data constructed in \cite{An2} have the
important additional property that they have a complete regular past.
The Vlasov matter model is distinguished from other models mentioned
above by several features. There are classes of initial
data---small data, data with ``outgoing'' particles---which lead
to global, geodesically complete solutions of the Einstein-Vlasov system,
cf.~\cite{AKR08,FaJoSm,LiTa19,RR}. There also exists a plethora of static and stationary solutions of this system, cf.~\cite{An0,An3,AFT,AKR14,RaRe,Rein94, RR00}. Stability or instability
of these steady states has been investigated in
\cite{GueStrRe,HaLinRe,HaRe2013,HaRe2014}, and perturbations of stable
steady states can result in time-periodic oscillations,
cf.~\cite{AR06,GueStrRe}. More background on the Einstein-Vlasov system
and also its Newtonian counterpart can be found in \cite{AGS,An1,GB,rein07}.
Altogether, a self-gravitating, collisionless gas can exhibit a large
variety of dynamical behaviors, both in the Newtonian and in the relativistic
regime, and with the present investigation we intend to add a new
one to this list which has its roots in the classical paper \cite{OS} 
by {\sc Oppenheimer} and {\sc Snyder}.

\section{The Einstein-Vlasov system} \label{section-ev}
\setcounter{equation}{0}
\subsection{The Einstein-Vlasov system in general coordinates}
The Einstein-Vlasov system describes a collisionless, self-gravitating
ensemble of particles in the context of general relativity.
Let $M$ denote a smooth four-dimensional spacetime manifold and
$g_{\mu \nu}$ a Lorentz metric on $M$ with signature $(- + + +)$.
In local coordinates $x^\mu$ the corresponding line element takes the form
\[
  ds^2 = g_{\mu \nu} d x^\mu d x^\nu;
\]
Greek indices run from $0$ to $3$,
and we use the Einstein summation convention.
The particle ensemble has a density
$f=f(x^\mu, p_\nu)\geq 0$ defined on the cotangent bundle $TM^*$,
where $p_\nu = g_{\nu \gamma} p^\gamma$ and $p^\gamma$ denote the canonical momenta
corresponding to the local spacetime coordinates $x^\mu$.
All the particles are assumed to have the same rest mass,
which we normalize to $1$,
i.e., $f$ is supported on the mass shell
\[
PM^* \coloneqq
\left\{ (x^\mu, p_\nu) \in TM^* \mid  g^{\mu \nu} p_\mu p_\nu = -1, \;
  p^\mu \ \mbox{future pointing}\right\}.
\]
We assume that on $PM^*$ the component $p^0$ can be expressed by the coordinates
$(x^\mu,p_j)$ where Latin indices run from $1$ to $3$. In what follows we write
$x^\mu = (t,x^j)$ and view $t$ as a time-like coordinate.
The basic unknowns in the Einstein-Vlasov system (besides the spacetime $M$)
are the Lorentz metric $g_{\mu \nu}$
and the particle density $f$. They obey the Einstein equations
\begin{align}\label{EE}
  G_{\mu \nu} = 8 \pi T_{\mu \nu} 
\end{align}
coupled to the Vlasov equation
\begin{align}\label{Vlaseq-general}
  \partial_t f + \frac{p^i}{p^0} \partial_{x^i} f -
  \frac{1}{2p^0} \frac{\partial g^{\nu \delta}}{\partial x^i}
  \, p_\nu p_\delta \, \partial_{p_i}f = 0 .
\end{align}
Here $G_{\mu \nu}$ is the Einstein tensor induced by the Lorentz metric
$g_{\mu \nu}$. The two equations are coupled by defining
the energy momentum tensor as
\begin{align}\label{energymomentumtensor}
  T_{\mu \nu} =  |g|^{-\frac{1}{2}} \int p_\mu p_\nu f \,  \frac{dp_1 dp_2 dp_3}{p^0},
\end{align}		
where $g$ is the determinant of the metric $g_{\mu \nu}$;
we use units so that the speed of light
and the gravitational constant equal unity.
The geodesic equations
\[
\frac{dx^\mu}{d\tau} = p^\mu = g^{\mu \nu} p_\nu, \quad
\frac{dp_\mu}{d\tau} = - \frac{1}{2} \frac{\partial g^{\nu \delta}}{\partial x^\mu}
\, p_\nu p_\delta,
\]
which describe the worldlines of test particles on the spacetime manifold
$(M, g_{\mu \nu})$, correspond to the characteristic equations of the Vlasov
equation \eqref{Vlaseq-general} for $f$ supported on the mass-shell $PM^*$,
the latter being invariant under the geodesic flow.

We wish to study the Einstein-Vlasov system
\eqref{EE}, \eqref{Vlaseq-general}, \eqref{energymomentumtensor}
under the assumptions of spherical symmetry and asymptotic flatness.
\subsection{The system in Painlev\'{e}-Gullstrand coordinates}
\label{section_ev_pg}
We consider the Einstein-Vlasov system with a metric of the form
\[
ds^2 = - dt^2 + a^2 (dr + \beta dt)^2 +
r^2 \left( d\theta^2 + \sin^2\theta d\phi^2\right) 
\]
where $t\in \R$, $r\geq 0$, $\theta\in [0,\pi]$, $\phi\in [0,2 \pi]$,
and $a=a(t,r),\ \beta=\beta(t,r)$. The radial coordinate $r$ is the
area radius, i.e., a sphere given by fixed values of $t$ and $r$
and parameterized by the polar angles $\theta$ and $\phi$ has
surface area equal to $4\pi r^2$. The meaning of the time coordinate
$t$ can be seen as follows. If we consider an observer moving
towards the future along a radial geodesic given by
\[
\dot r = \frac{dr}{dt} = - \beta(t,r), 
\]
then along this curve, i.e., for such an observer, $ds^2 = - dt^2$,
so that the time coordinate $t$ coincides with the proper time measured
by this observer; time is measured by radially falling geodesic clocks,
cf.~\cite{GC}.
Asymptotic flatness and regularity at the center require that
\begin{align} \label{bc}
\beta(t,0)=\beta(t,\infty) =0,\ a(t,0) = a(t,\infty)=1.
\end{align}
Taking suitable combinations
of the Einstein equations and defining
\begin{align} \label{Adef}
  A = \frac{1}{a^2} - \beta^2
\end{align}
results in the following set of field equations:
\begin{align}
  \partial_t\big(r(1-A)\big)
  &=
  8\pi r^2\left(\beta(\rho + p) - (a\beta^2 + a^{-1})j\right),\label{ee1}\\
  \partial_r\big(r(1-A)\big)
  &=
  8\pi r^2(\rho-a\beta j),\label{ee2}\\
  \partial_t\beta-\beta\,\partial_r\beta
  &=
  \frac{1}{2r}(1-A)+4\pi rp,\label{ee3}\\
  \partial_t a-\beta\,\partial_r a
  &=
  -4\pi r a^2 j.\label{ee4}
\end{align}
The definition of the matter terms $\rho, p, j$ in terms of $f$
follows shortly, and we note that
\begin{align*}
  T_{00}
  &= \rho - 2 a \beta j + a^2 \beta^2 p,\\
  T_{01}
  &= a^2 \beta p - a j,\\
  T_{11}
  &= a^2 p,\\
  T_{22}
  &= r^2 p_T;
\end{align*}
the tangential pressure $p_T$, which does not appear in the field equations
\eqref{ee1}--\eqref{ee4}, will play its role later on.
To formulate the Vlasov equation and the matter terms we
use Cartesian coordinates
\[
(x^1,x^2,x^3) = r\,(\sin\theta\cos\phi,\sin\theta\sin\phi,\cos\theta),
\]
and non-canonical momentum variables defined by the relations
\[
v_c = p_c + \left(\frac{1}{a} -1\right) \frac{x^j p_j}{r} \frac{x_c}{r},\quad
p_c = v_c + (a-1) \frac{x\cdot v}{r} \frac{x_c}{r}.
\]
Here $p^c$ are the canonical momenta corresponding to the spatial coordinates
defined above,
$\cdot$ denotes the Euclidean scalar product in $\R^3$,
and indices of vectors
$x, v \in \R^3$ are raised and lowered with the Kronecker $\delta$;
in addition $|v|$ denotes the Euclidean norm.
Then
\[
p^0 = \sqrt{1+|v|^2} =: \p0.
\]
When using these coordinates the function $f$ must explicitly
be required to be spherically symmetric,
i.e., $f(t,A x,A v)= f(t,x,v)$ for $A\in \mathrm{SO}(3)$.
The Vlasov equation reads
\begin{align}\label{Vl-cart}
  \partial_t f
  &+
  \left(\frac{v}{a \p0}-\beta\frac{x}{r}\right)\cdot \partial_x f\notag \\
  & {}+ \left[\left(\beta\frac{\partial_r a}{a} -\frac{\partial_t a}{a}
    +\partial_r\beta\right)\frac{x\cdot v}{r} \frac{x}{r}
    +\frac{\beta}{r}\left(v - \frac{x\cdot v}{r} \frac{x}{r}\right)\right]
  \cdot\partial_v f=0.
\end{align}
The matter terms are
\begin{align}
\rho(t,r) 
&=
\rho(t,x) = \int \p0 f(t,x,v)\,dv ,\label{rhoc}\\
p(t,r) 
&=
p(t,x) = \int \left(\frac{x\cdot v}{r}\right)^2
 f(t,x,v)\frac{dv}{\p0}, \label{pc}\\
j(t,r) 
&=
j(t,x) = \int \frac{x\cdot v}{r} f(t,x,v)\, dv, \label{jc}\\
p_T(t,r) 
&=
p_T(t,x) = \frac{1}{2} \int \left|{\frac{x\times v}{r}}\right|^2
f(t,x,v) \frac{dv}{\p0}. \label{p_Tc}
\end{align}
The system \eqref{ee1}--\eqref{jc} together with the boundary conditions
\eqref{bc} and the definition \eqref{Adef} is the Einstein-Vlasov system
in Painlev\'{e}-Gullstrand coordinates.
We need to prescribe the following initial data:
\begin{align}\label{indata}
f(0) =\mathring{f},\ a(0)=\mathring{a},\
\beta(0)=\mathring{\beta};
\end{align}
here $f(0)$ stands for $f$
restricted to the initial time slice $t=0$ etc.
Besides regularity assumptions which we specify below
the data $\mathring{a}$ and $\mathring{\beta}$ must
satisfy the boundary conditions
\eqref{bc} and the constraint equation~\eqref{ee2}.

In passing we note the following fact which will become important later on.
For a characteristic curve $s\mapsto (x(s),v(s))$ of the Vlasov equation
\eqref{Vl-cart}, $v(0)=0$  implies that $v(s)=0$ for all times.
A characteristic with this property corresponds to one of the radially
falling geodesic clocks, and in general, $v/a \p0$ is the velocity relative
to the latter.
\section{The local existence result}
\label{section_locex_form} 
\setcounter{equation}{0}
\subsection{A suitable reduced system}
In order to prove a local existence result we
need to select a suitable subsystem, as the system stated above is
over-determined.
We replace the combination of derivatives of $a$
in \eqref{Vl-cart} via \eqref{ee4} and keep only
\eqref{ee3} and \eqref{ee4} together with the corresponding
matter terms to determine $a$ and $\beta$.
We use the integrated version of \eqref{ee2} to define $A$
and arrive at the following reduced system:
\begin{align}\label{Vl-red}
  \partial_t f
  &+
  \left(\frac{v}{a \p0}-\beta\frac{x}{r}\right)\cdot \partial_x f\notag \\
  & {}+ \left[4\pi r \,a\, j\, \frac{x\cdot v}{r} \frac{x}{r}
    +\left(\partial_r\beta-\frac{\beta}{r}\right)\frac{x\cdot v}{r} \frac{x}{r}
    +\frac{\beta}{r} v\right]
  \cdot\partial_v f=0.
\end{align}
\begin{align}
  \partial_t\beta-\beta\,\partial_r\beta
  &=
  \frac{1}{2r}(1-A)+4\pi rp,\label{ee-red3}\\
  \partial_t a -\beta\, \partial_r a
  &=
  - 4\pi r a^2 j,\label{ee-red4}
\end{align}
where $\rho, p, j$ are given in terms of $f$ by \eqref{rhoc}--\eqref{jc},
and
\begin{align} \label{Adef-red}
  A(t,r) = 1-\frac{8\pi}{r}\int_0^r\left(\rho - a \beta j\right)(t,s)\,s^2 ds.
\end{align}
It is important to make sure that $a$ remains positive.
This fact together with the form of \eqref{ee-red4} motivates an
auxiliary quantity defined by
\begin{align} \label{lambdadef}
  a = e^\lambda,
\end{align}
in terms of which \eqref{ee-red4} can be rewritten as
\begin{align} \label{ee-red4-lambda}
  \left(\partial_t - \beta\, \partial_r\right)\, e^{-\lambda} = 4 \pi r\, j.
\end{align}
\subsection{Initial data}\label{data}
The initial data prescribed in \eqref{indata} must satisfy the following
assumptions:
\[
\mathring{f} \in C^1_c (\R^6),\ 
\mathring{a}\in C^1([0,\infty[),\
\mathring{\beta}\in C^2([0,\infty[),
\]
where the subscript $c$ denotes compact support, $\mathring{f}$
is spherically symmetric,
\[
\mathring{f} \geq 0,\
\mathring{a}(0) = 1 = \mathring{a}(\infty),\  
\mathring{a}'(0) = 0,\
\mathring{\beta}(0) = \mathring{\beta}''(0) = 0 = \mathring{\beta}(\infty),
\]
\[
\mathring{a}(r) > 0\ \mbox{for}\ r\geq 0,
\]
and
\begin{align} \label{constr}
\mathring{A}(r)=\frac{1}{\mathring{a}(r)^2} -\mathring{\beta}(r)^2
\ \mbox{for}\ r\geq 0,
\end{align}
where $\mathring{A}$ is as in \eqref{Adef-red} with
$\rho=\mathring{\rho}$ and $j=\mathring{\jmath}$ induced by $\mathring{f}$.
\begin{remark}
  Such data can for example be obtained by taking $\mathring{a}=1$,
  $\mathring{f}$ even in $v$ so that $\mathring{\jmath}=0$,
  and then solving \eqref{constr} for $\mathring{\beta}$,
  i.e.,
  \begin{align} \label{beta0def}
  \mathring{\beta} (r)
  = \left(\frac{8\pi}{r}\int_0^r \mathring{\rho}(s)\,s^2 ds\right)^{1/2}.
  \end{align}
  Using l'Hospital's rule one can show that indeed
  \[
  \mathring{\beta}\in C^2([0,\infty[),
  \ \mathring{\beta}(0) = \mathring{\beta}''(0) =0.
  \]
  However, we must not a-priori restrict our local existence result
  to such data, because when we extend the solution to some maximal
  existence interval, $f_{|t=t_0}, a_{|t=t_0}, \beta_{|t=t_0}$
  must for any $t_0>0$ again be admissible data.
\end{remark}
\subsection{The solution concept and the local existence result}
We first specify the regularity properties which the various parts of a
solution must satisfy.
\begin{definition}\label{regular}
  Let $I\subset \R$ be an interval.
  \begin{itemize}
    \item[(a)]
      $f\colon I\times \R^6 \rightarrow [0,\infty[$ is {\em regular}, if
      $f\in C^1(I\times \R^6)$, $f(t)$ is spherically symmetric for
      $t\in I$, and $\supp f(t)$ is compact, locally uniformly in $t \in I$. 	
    \item[(b)]
      $\rho $ (or $p$, $p_T$)$\colon  I \times \R^3 \rightarrow [0,\infty[$
      is {\em regular}, if
      $\rho \in C^1 (I\times \R^3)$, $\rho(t)$
      is spherically symmetric for $t\in I$,
      and $\supp \rho(t)$ is compact, locally uniformly in $t\in I$. 		
    \item[(c)]
      $j\colon I\times \R^3 \rightarrow \R$ is {\em regular}, if
      $ j \in C(I\times \R^3) \cap C^1 (I\times \R^3 \setminus \{0\} )$,
      $j(t)$ is spherically symmetric for $t \in I$, $ \supp j(t)$ is compact,
      locally uniformly in $t \in I$, $j \in C^1(I\times[0,\infty[)$
      as a function in $(t,r)$, and $j(t,0)=0$. 
    \item[(d)]
      $a\colon I\times [0,\infty[ \; \rightarrow \; ]0,\infty[$
      is {\em regular},
      if $a \in C^1(I\times [0,\infty[)$,
      and the boundary conditions $\partial_r a(t,0)=0,\ a(t,0)=1$ hold.
    \item[(e)]
      $\beta \colon I\times [0,\infty[ \; \rightarrow \; \R$
      is {\em regular}, if
      $\beta  \in C^1(I \times [0,\infty[)$,
      $\partial_r \beta  \in C^1(I \times [0,\infty[)$,
      and the boundary conditions
      $\beta(t,0) = \partial_r^2 \beta(t,0) =0$ hold.
    \item[(f)]
      {\em The fields are locally bounded}, if
      $a^{-1}$, $r a j$, $\beta$,
      $\partial_r\beta$ 
      are bounded on $J \times [0,\infty[$ for any compact interval
      $J\subset I$.
  \end{itemize}
\end{definition}
Note that the boundary conditions at infinity which were part of \eqref{bc}
are not included in the definition above; we will take care of them after
having established a solution in the sense of the definition.
We aim to prove the following result; by $\| \cdot \|$ we denote the
sup norm with respect to $r=|x|\in [0,\infty[$ or $(x,v)\in \R^3\times\R^3$,
as the case may be.
\begin{theorem}\label{locex}
  Each set of initial data $(\mathring{f}, \mathring{a}, \mathring{\beta})$
  as specified in Section~\ref{data} launches a unique solution $(f,a,\beta)$
  of the reduced system \eqref{Vl-red}--\eqref{Adef-red} on some time interval
  $[0,T[$, the components of
  which are regular in the sense of Definition~\ref{regular}.
  The boundary conditions \eqref{bc} hold.    
  If $T>0$ is chosen maximal and
  \begin{align}\label{contcrit}
    &
    \sup\bigl\{|v| \mid (x,v)\in \supp f(t),\ 0\leq t < T\bigr\}\notag\\
    &\qquad{}
    + \sup\bigl\{\|a(t)\| + \|\partial_r \beta(t)\| \mid 0\leq t < T\bigr\}
    < \infty,
  \end{align}
  then $T=\infty$.
\end{theorem}
This theorem will be proven by setting up an iteration scheme which generates
a sequence of approximate solutions converging to the desired solution
of the reduced system.
\subsection{Recovering the full system}
It is important that a sufficiently regular solution of the subsystem
is indeed a solution to the full set of equations.
\begin{proposition}\label{fullsystem}
  Assume that $(f,a,\beta)$ is a solution of the subsystem as obtained in
  Theorem~\ref{locex}. Then the equations \eqref{Adef}, \eqref{ee1}
  and the Vlasov equation in the form \eqref{Vl-cart} hold as well.
  The full set of the Einstein equations \eqref{EE} is satisfied.
\end{proposition}
In order to prove this result, the following auxiliary results will be useful.
\begin{lemma}\label{dtrhodtj}
  Let $f$ be a regular solution of the Vlasov equation \eqref{Vl-red}
  corresponding to a given, regular set of fields
  $(\tilde a,\tilde \jmath,\tilde \beta)$. Then
  \begin{align*}
    \partial_t \rho
    &=
    \tilde \beta \partial_r \rho
    -\frac{1}{\tilde a} \partial_r j - \frac{2}{r \tilde a} j
    +
    \frac{2 \tilde \beta}{r}(\rho + p_T)
    +
    \left(\partial_r \tilde \beta +4\pi r \tilde a \tilde\jmath\right)
    (\rho + p), \\
    \partial_t j
    &=
    \tilde \beta \partial_r j + \frac{2\tilde \beta}{r} j
    -\frac{1}{\tilde a}
    \left(\partial_r p + \frac{2}{r} p - \frac{2}{r} p_T\right)
    +
    2 \left(\partial_r \tilde \beta +4\pi r \tilde a \tilde\jmath\right) j .
  \end{align*}
\end{lemma}
\begin{proof}
  First we differentiate the formulas \eqref{rhoc} and \eqref{jc} with respect
  to $t$. Then we substitute $\partial_t f$ via the Vlasov equation
  \eqref{Vl-red}. The terms containing $\partial_v f$ can be integrated
  by parts. The terms containing $\partial_x f$ can be related to
  (radial derivatives of) various source terms,
  more precisely,
  \begin{align*}
    \int v \cdot \partial_x f dv
    &=
    \partial_r j + \frac{2}{r} j,\\ 
    \int \p0 \frac{x}{r} \cdot \partial_x f dv
    &=
    \partial_r \rho,\\
    \int \frac{x\cdot v}{r} \frac{x}{r}\cdot \partial_x f dv
    &=
    \partial_r j,\\
    \int \frac{x\cdot v}{r} \frac{v}{\p0}\cdot \partial_x f dv
    &=
    \partial_r p + \frac{2}{r} p - \frac{2}{r} p_T,
  \end{align*}
   cf.~\cite[Lemma 6.6, Lemma 6.8]{GB}. Collecting terms yields the result.
\end{proof}
  
Notice that in the proof only the Vlasov equation is used,
the field equations play no role. This will become important
in the next section, where we will apply this
result to the iterates in an iteration scheme.
An analogous comment applies to the second auxiliary result.
\begin{lemma}\label{lemmabetaeq}
  Let $T>0$ and let $\tilde \beta$ be
  regular on $[0,T[\times [0,\infty[$ in the sense of
  Definition~\ref{regular}~(e), and bounded on
  $[0,T']\times [0,\infty[$ for any $T'<T$.
  \begin{itemize}        
  \item[(a)]
    For any $r\geq 0$ and $t\in [0,T[$ there exists an unique $C^1$ solution
    $[0,T[\,\ni s\mapsto R(s,t,r)$ of the equation
    \[        
    \dot r = -\tilde \beta(s,r)
    \]
    with $R(t,t,r)=r$. Moreover, $R \in C^1([0,T[^2\times [0,\infty[)$ and
    $\partial_r R \in C^1([0,T[^2\times [0,\infty[)$ with
    $R(s,t,0) = 0 = \partial_r^2 R(s,t,0)$.
  \item[(b)]
    For $\mathring{b}\in C^1([0,\infty[)$
    and $q\in C^1([0,T[\times [0,\infty[)$,
    \[
    b(t,r) = \mathring{b}(R(0,t,r)) + \int_0^t q(s,R(s,t,r))\, ds
    \]
    defines the unique solution $b\in C^1([0,T[\times[0,\infty[)$ of
    the initial value problem
    \[        
    \partial_t b-\tilde \beta(t,r)\, \partial_r b = q(t,r),\
    b(0,r) =\mathring{b}(r).
    \]
  \end{itemize}
\end{lemma}
\begin{proof}
  In order to prove part (a) we note that the assumptions on $\tilde \beta$
  allow us to extend this function
  via $\tilde \beta(t,r) \coloneqq - \tilde \beta(t,-r)$ to
  $\tilde \beta\in C^1([0,T[\times \R)$ with $\tilde \beta(t,0)=0$.
  For this extended characteristic equation the existence 
  of $R\in C^1([0,T[^2\times \R)$ is standard; notice that by assumption,
  $\tilde \beta$ is bounded on compact subintervals of $[0,T[$ so that
  the characteristics do exist
  on $[0,T[$. Since $\tilde \beta(t,0)=0$ it follows that $R(s,t,0)=0$,
  and no characteristic
  can cross the line $r=0$. Hence we can consistently restrict $R$ to
  $[0,T[^2\times [0,\infty[$.        
  Since
  \[
  \frac{d}{ds} \partial_r R(s,t,r)
  = - \partial_r \tilde \beta(s,R(s,t,r)) \partial_r R(s,t,r),\
  \partial_r R(t,t,r) = 1
  \]
  it follows that
  \begin{align}\label{drRrep}
    \partial_r R(s,t,r) =
    \exp\left(\int_s^t \partial_r \tilde \beta(\tau,R(\tau,t,r))\,d\tau\right),
  \end{align}
  and since $\partial_r \tilde \beta$ is continuously
  differentiable, 
  $\partial_r R \in C^1([0,T[^2\times \R)$.
  Moreover, $\partial_r^2 R(s,t,0)=0$, since
  $\partial_r^2 \tilde \beta(t,0) =0$.  
  Part (b) is standard theory of first order PDEs.
\end{proof}
We are now ready to prove Proposition~\ref{fullsystem}.

\noindent
{\em Proof of Proposition~\ref{fullsystem}}.
Assume that we are given a regular solution to the subsystem
\eqref{Vl-red}--\eqref{Adef-red}. Then by \eqref{ee-red4}, the original
Vlasov equation \eqref{Vl-cart} holds as well. The key points are to see that
$A$, which in the subsystem is given by \eqref{Adef-red}, indeed satisfies
the relation \eqref{Adef} and the field equation \eqref{ee1};
the field equation \eqref{ee2} follows directly from \eqref{Adef-red}.

We first compute the left hand side of \eqref{ee1},
using the definition \eqref{Adef-red}
and Lemma~\ref{dtrhodtj}. This results in terms containing $\partial_r \rho$,
$\partial_r p$, or $\partial_r j$ under the radial integral.
If we integrate these terms by parts, we find that
\begin{align}\label{dtA}
  \partial_t (r(1-A))
  &=
  8\pi r^2 \left(\beta (\rho + p) - (a \beta^2 +1/a) j\right)\notag\\
  &
  {} + 4\pi
  \int_0^r s^2 a j \,\partial_r \left(\frac{1}{a^2} -\beta^2 -A\right)(t,s)\, ds;
\end{align}
this is indeed \eqref{ee1}, provided $A$ satisfies the relation \eqref{Adef}.

Let us define
\[
d \coloneqq A - \frac{1}{a^2} + \beta^2.
\]
Using \eqref{dtA}, an integration by parts with respect to the radial variable,
and \eqref{Adef-red} it follows that
\[
\partial_t d - \beta \, \partial_r d = 4 \pi r a j d
- \frac{4\pi}{r}
\int_0^r \partial_r(a j \sigma^2)(t,\sigma) d(t,\sigma) \, d\sigma.
\]
Now the assumption on the data guarantees that $d(0,r)=0$, and we can use
Lemma~\ref{lemmabetaeq} to conclude that
\[
d(t,r)= 4\pi\int_0^t \left[(r a j) - \frac{1}{r}
\int_0^{r} \partial_r(a j \sigma^2)(t,\sigma) d(t,\sigma)\, d\sigma\right]
(s,R(s,t,r))\, ds.
\]
In the present context the function $d$ is bounded with respect to $r$,
locally uniformly in $t$. The previous equation thus implies that on
any compact time interval $[0,T']\subset [0,T[$,
\[
\|d(t)\| \leq C \int_0^t \|d(s)\|\, ds
\]
which implies that $d=0$. Hence \eqref{Adef} and by
\eqref{dtA} also \eqref{ee1} hold.

To conclude this proof we must recall that \eqref{ee1}--\eqref{ee4}
together with \eqref{Adef}, which now all hold, are not the Einstein
equations $E_{\mu\nu}\coloneqq G_{\mu\nu}- 8\pi T_{\mu\nu} = 0$, but suitable
combinations of some of them. Indeed, \eqref{ee3} is, up to an overall
factor $2 r a^2$, the equation $E_{11}=0$. Next we observe that
$a E_{01} - a\beta E_{11} = 0$ is equivalent to the equation \eqref{ee4},
but since we already know that $E_{11}=0$
it follows that $E_{01}=0$ as well.
Equation \eqref{ee2} arises as the combination
$E_{00} - \beta E_{01} =0$ from which we conclude that $E_{00}=0$.
It remains to show that $E_{22}=0$; the also nontrivial equation
$E_{33}=0$ is just a multiple of this, and all
the other components of the field equations vanish by symmetry.
The equation $E_{22}=0$ explicitly reads as follows:
\begin{align}\label{EE22}
\frac{r}{a^{3}}
\Bigl[r\, a^{3} \partial_r\partial_t\beta  - r\, a^{2} \partial^2_{t} a 
- r \,a^{2} \beta^{2} \partial^2_{r} a  
-  r \,a^{3}\beta \partial^2_{r} \beta
+ 2 r\, a^{2}\beta  \partial_r \partial_t a 
&\nonumber \\
+ 2 r\,   a^{2} \partial_t a \partial_r \beta  + r\,a^{2} \partial_t\beta
 \partial_r a 
- r\, a^{3}\left( \partial_r \beta\right)^{2}-  a^{2} \beta^{2} \partial_r a 
&\nonumber\\ 
- 2\, a^{3}\beta \partial_r \beta
- \partial_r a + a^{2}\beta  \partial_t a    
-3 r \, a^{2}\beta  \partial_r a \partial_r \beta 
+ a^{3}\partial_t\beta\Bigr]
&=
8\pi r^2 p_T.
\end{align}
We differentiate \eqref{ee3} with respect to $r$,
and \eqref{ee4} with respect to
$t$ or with respect to $r$. If we substitute the resulting expressions for
$\partial_{r}\partial_t\beta$, $\partial^2_{t} a$, and $\partial_r \partial_t a$
into the left hand side of \eqref{EE22}, observe Lemma~\ref{dtrhodtj},
and collect terms we see that \eqref{EE22} holds. \hspace*{\fill} $\Box$

\begin{remark}\label{areg}
  Strictly speaking, the regularity of $a$ which is specified in
  Definition~\ref{regular}
  is not sufficient for the above derivation of \eqref{EE22}.
  However, this is only
  a technical point, and after having proven Theorem~\ref{locex}
  we will see that
  $a$ is indeed as regular as $\beta$, provided
  $\mathring{a}\in C^2([0,\infty[)$,
  cf.\ Corollary~\ref{corareg} below.
  For the proof of Theorem~\ref{locex} it is useful not to
  require more regularity
  than needed in order to make sense of the subsystem.
  One reason for discussing
  \eqref{EE22} already at this point is the fact that it
  motivates a certain maneuver
  which leads to control of first order derivatives of $f$
  in Section~\ref{magic} below.
\end{remark}
\section{Proof of the local existence result}
\label{section_locex_proof}
\setcounter{equation}{0}
In this section we prove Theorem~\ref{locex}.
The proof is split into a number of steps
which we deal with in the various subsections which follow.
\subsection{The Vlasov equation and the characteristic system}
We will deal with the modified Vlasov equation \eqref{Vl-red} via the method of
characteristics. When constructing the iterative scheme below it is important
that we study \eqref{Vl-red} for a prescribed set of fields
$\tilde a,\tilde \beta, \tilde\jmath$.
For $s\in [0,T[$, $x,v\in \R^3$ we define
\begin{align}
  &
  F_1(s,x,v)\coloneqq
  \frac{v}{\tilde a(s,r) \p0}-\tilde\beta(s,r)\frac{x}{r}, \label{F1-def}\\
  &
  F_2(s,x,v)\coloneqq
  4\pi r \,(\tilde a\, \tilde\jmath)(s,r)\, \frac{x\cdot v}{r} \frac{x}{r}
  +\left(\partial_r\tilde \beta
  -\frac{\tilde \beta}{r}\right)(s,r)\frac{x\cdot v}{r} \frac{x}{r}
    +\frac{\tilde \beta(s,r)}{r} v. \label{F2-def}
\end{align}
In what follows we often abbreviate $z=(x,v)\in \R^3\times\R^3$.
The characteristic system of the modified Vlasov equation \eqref{Vl-red}
can now be written in the form
\begin{align}\label{modcharsys}
  \dot z = F(s,z).
\end{align}
\begin{lemma} \label{modcharsyssolve}
  Assume that $\tilde a,\tilde \beta, \tilde\jmath$
  are regular and locally bounded
  on the time interval $[0,T[$ as specified
  in Definition~\ref{regular}.
  \begin{itemize}
  \item[(a)]
    For any $t\in [0,T[$ and $z\in \R^3\times\R^3$ there exists a unique
    solution $[0,T[\,\ni s \mapsto Z(s,t,z)=(X,V)(s,t,x,v)$ of
    \eqref{modcharsys} with
    $Z(t,t,z)=z$. Moreover,
    $Z\in C^1 ([0,T[\times [0,T[\times\R^3\times\R^3)$.
  \item[(b)]
    By $f(t,z)\coloneqq\mathring{f}(Z(0,t,z))$ for $t\in [0,T[$ and
    $z\in \R^3\times\R^3$
    the unique, regular solution to the modified Vlasov equation \eqref{Vl-red}
    with $f(0)=\mathring{f}$ is defined.
  \end{itemize}
\end{lemma}
\begin{proof}
  We claim that $F\in C^1([0,T[\times\R^3\times\R^3)$. The crucial point here is
  the regularity at the spatial origin. The behavior of $a,\ j,\ \beta$
  at $r=0$ together with l'Hospital's rule imply that
  \[
  \lim_{r\to 0}\left(\partial_r\beta-\frac{\beta}{r}\right)(s,r)
  = 0 = \lim_{r\to 0}\left(\frac{r \partial_r\beta-\beta}{r^2}\right)(s,r).
  \]
  Hence the limits
  \[
  \lim_{r\to 0} \partial_{x_i}\left(\beta\frac{x_j}{r}\right)
  = \partial_r\beta(t,0)\delta_{ij},
  \]
  \[
  \lim_{r\to 0} \partial_{x_i}\left(r a j\frac{x\cdot v}{r}\frac{x_j}{r}\right)
  = 0 =
  \lim_{r\to 0} \partial_{x_i}
  \left(\left(\partial_r\beta-\frac{\beta}{r}\right)
  \frac{x\cdot v}{r}\frac{x_j}{r}\right)
  \]
  exist, and $F$ is $C^1$, also at the center.
  Since for any $0<T'<T$ the function $F$ is bounded with respect to
  $s\in [0,T']$ and $x\in \R^3$ and
  grows at most linearly in $|v|$, any solution of the
  characteristic system exists
  for $s\in [0,T[$. The remaining assertions are standard.
\end{proof}
\subsection{The iteration scheme---the set-up}
\subsubsection{The $0\,$th iterates}
We define the $0\,$th iterate of the fields on $[0,\infty[\times [0,\infty[$
as
\[
a_0(t,r) \coloneqq 1,\ j_0(t,r) \coloneqq 0,\ \beta_0(t,r) \coloneqq 0.
\]
These iterates need not fit the given initial data, but it is
important that they satisfy \eqref{ee-red4}.
For technical reasons we also define
$f_0 \coloneqq 0$.
\subsubsection{Induction assumption on the $(n-1)\,$st iterates}
We assume that $a_{n-1}, j_{n-1}, \beta_{n-1}$ are defined on
$[0,T_{n-1}[\times [0,\infty[$ for some $T_{n-1}>0$,
and that they are regular and locally bounded
in the sense of Definition~\ref{regular}. In addition, we
assume that
\begin{align}\label{ee4-red-n-1}
  \partial_t a_{n-1} -\beta_{n-1}\, \partial_r a_{n-1} = - 4\pi r a_{n-1}^2 j_{n-1};
\end{align}
notice that these assumptions hold for the $0\,$th iterates, i.e., for $n=1$.
\subsubsection{$f_n$ and its characteristics}
We replace
$\tilde a,\ \tilde\jmath,\ \tilde \beta$ in \eqref{F1-def}, \eqref{F2-def}
by the above $(n-1)\,$st
iterates and denote the resulting function by $F_{n-1}$. Then we can apply
Lemma~\ref{modcharsyssolve} to obtain corresponding iterates for the
characteristics $Z_n=(X_n,V_n)$ as the solution of
$\dot z = F_{n-1}(s,z)$ and $f_n(t,z)\coloneqq\mathring{f}(Z_n(0,t,z))$
as the $n\,$th iterate
for $f$, which is defined on $[0,T_{n-1}[\times\R^3\times\R^3$.
As desired, $f_n$ is regular in the sense of Definition~\ref{regular}~(a).
\subsubsection{The $n\,$th iterates for the source terms}
We define $\rho_n, p_n, j_n, p_{T,n}$
by \eqref{rhoc}--\eqref{p_Tc} with $f_n$ instead of $f$. These functions
are regular in the sense of Definition~\ref{regular}~(b), (c);
for the details of this assertion we refer to \cite[Lemma 6.8]{GB}.
\subsubsection{The $n\,$th iterates for the metric quantities---$A_n$}
We define
\begin{align}\label{Andef}
  A_n(t,r) = 1 - \frac{8\pi}{r}
  \int_0^r \left(\rho_n - a_{n-1}\beta_{n-1} j_n\right)(t,s)\, s^2 ds.
\end{align}
It is easily seen that the function
$(1-A_n)/(2r)$ is continuously differentiable on $[0,T_{n-1}[\times [0,\infty[$
and twice continuously differentiable with respect to $r$ on this set
with  $[(1-A_n)/(2r)]_{|r=0} = 0 =\partial_r^2 [(1-A_n)/(2r)]_{|r=0}$.
\subsubsection{The $n\,$th iterates for the metric quantities---$\beta_n$}
\label{beta_nsection}
This is the crucial step.
We define $\beta_n$ as
the solution of the initial value problem
\begin{align}\label{betaeq}
 \partial_t\beta-\beta_{n-1}\, \partial_r\beta =
 \frac{1}{2r}(1-A_n)+4\pi r p_n,\ 
\beta(0,r) =\mathring{\beta}(r).
\end{align}
By the induction assumption on the previous iterates
we can apply Lemma~\ref{lemmabetaeq}~(a) with
$b=\beta_{n-1}$.
This implies that for any $t\in [0,T_{n-1}[$ and $r\geq 0$
there exists a unique solution
$[0,T_{n-1}[\,\ni s\mapsto R_n(s,t,r)$ of the characteristic equation
\[
\dot r = - \beta_{n-1}(s,r)
\]
with $R_n(t,t,r)=r$, and $R_n$ has the regularity properties stated
in Lemma~\ref{lemmabetaeq}~(a).
Following part (b) of that lemma we define the next iterate for $\beta$ as
\begin{align}\label{betandef}
  \beta_n(t,r)
  &\coloneqq
  \mathring{\beta}(R_n(0,t,r)) +
  \int_0^t\left(\frac{1}{2r}(1-A_n)\right)(s,R_n(s,t,r))\, ds
  \nonumber\\
  & \qquad {}
  + 4\pi \int_0^t(r p_n)(s,R_n(s,t,r))\, ds\nonumber\\
  &=:
  \beta_{n,1}(t,r) + \beta_{n,2}(t,r) + \beta_{n,3}(t,r).
\end{align}
The difficulty is to recover for $\beta_n$ the regularity properties
which we assumed for $\beta_{n-1}$. To do so we first observe that analogously
to \eqref{drRrep},
\begin{align}\label{drRnrep}
\partial_r R_n(s,t,r) =
\exp\left(\int_s^t \partial_r\beta_{n-1}(\tau,R_n(\tau,t,r))\,d\tau\right).
\end{align}
Clearly,
$\beta_{n,1}\in C^1([0,T_{n-1}[ \times [0,\infty[)$ with 
$\partial_r \beta_{n,1}\in C^1([0,T_{n-1}[ \times [0,\infty[)$.
The same is true for $\beta_{n,2}$ by the regularity properties
of  $(1-A_n)/(2r)$ stated above.
The problematic term is $\beta_{n,3}$. 
Obviously,
\[
\partial_r \beta_{n,3}(t,r) = \int_0^t \partial_r (r p_n)(s,R_n(s,t,r))
\partial_r R_n(s,t,r)\, ds.
\]
From \cite[p.~98]{GB} we see that
\begin{align}\label{pr}
\partial_r (r p_n)(s,r) = 2 p_{T,n}(s,r) - p_n(s,r)
+ r \int \frac{x\cdot v}{r} \frac{v}{\p0} \cdot \partial_x f_n(s,x,v)\, dv.
\end{align}
Only the last term, which we denote by $\sigma_n(s,x)$,
needs to be investigated further.
To do so we need to exploit the fact that
$f_n$ solves the ``previous'' Vlasov equation.
We define a
Cartesian version of the characteristic $R_n(s,t,r)$, namely
$s\mapsto X_n(s,t,x)$ is the solution to
\begin{align}\label{X_nbeta}
\dot x = - \beta_{n-1}(s,r)\frac{x}{r}
\end{align}
with $X_n(t,t,x) = x$;
note that the right hand side of this differential equation
is in $C^1([0,T_{n-1}[\times \R^3)$ and vanishes at the center $x=0$. In
particular, $X_n(s,t,0)=0$, and no other characteristic can cross the center.
It is easy to see that $|X_n(s,t,x)|= R_n(s,t,r)$ where $r=|x|$.
Moreover,
\[
\frac{X_n(s,t,x)}{|X_n(s,t,x)|} = \frac{x}{r}
\]
for $x\neq 0$; notice that $\sigma_n(s,0)=0$.
It is interesting to note that $s\mapsto (X_n(s,t,x),0)$
is actually a characteristic
of the Vlasov equation, namely $(X_n(s,t,x),0)=(X_n(s,t,x,0),V_n(s,t,x,0))$.
Let
\[
D_t f_n(t,x,v) \coloneqq \partial_t f_n(t,x,v) -
\beta_{n-1}(t,r) \frac{x}{r}\cdot \partial_x f_n (t,x,v).
\]
In what follows we often abbreviate
$X_n(s) = X_n(s,t,x)$ and $R_n(s)=R_n(s,t,r)$.
Then the chain rule and the Vlasov equation for $f_n$ imply that
\begin{align*}
\frac{d}{ds} f_n(s,X_n(s),v)
&=
(D_t f_n)(s,X_n(s),v) =
-\left(\frac{v}{a_{n-1} \p0}\cdot \partial_x f_n\right)(s,X_n(s),v)\\
&\quad
{} - \left(F_{n-1,2} \cdot \partial_v f_n\right)(s,X_n(s),v);
\end{align*}
$F_{n-1} = (F_{n-1,1},F_{n-1,2})$ is defined analogously to
\eqref{F1-def}, \eqref{F2-def}. This implies that
\begin{align*}
\left(\frac{x\cdot v}{r} \frac{v}{\p0} \cdot
\partial_x f_n\right)(s,X_n(s),v)
&=
- \left(a_{n-1} \frac{x\cdot v}{r} D_t f_n\right)(s,X_n(s),v)\\
&\quad
{}- \left(a_{n-1} \frac{x\cdot v}{r} F_{n-1,2}
\cdot \partial_v f_n\right)(s,X_n(s),v).
\end{align*}
Integration with respect to $v$ and an integration by parts imply that
\begin{align*}
  &
  \left(\int \frac{x\cdot v}{r} \frac{v}{\p0} \cdot \partial_x f_n dv\right)
  (s,R_n(s))\\
  &=
  - a_{n-1}(s,R_n(s)) \frac{d}{ds} j_n(s,R_n(s))\\
  &
  {}+ 2 a_{n-1}(s,R_n(s))
  \left(4\pi \,a_{n-1}\, j_{n-1}\, j_n
  +\left(\partial_r\beta_{n-1} + \frac{\beta_{n-1}}{r}\right)
  j_n\right) (s,R_n(s)).
\end{align*}
We substitute this into the formula for $\partial_r \beta_{n,3}$ to obtain
\begin{align}\label{drbeta_n3}
  &
  \frac{1}{4\pi}\partial_r \beta_{n,3}(t,r) =
  \int_0^t (2 p_{T,n}-p_n)(s,R_n(s,t,r))\, \partial_r R_n(s,t,r)\, ds \notag\\
  &
  {}- \int_0^t (r a_{n-1})(s,R_n(s,t,r))  \partial_r R_n(s,t,r) \frac{d}{ds}
  j_n(s,R_n(s,t,r))\, ds\notag \\
  &
  {}+ 2 \int_0^t \left(r a_{n-1} j_n
  \left(4\pi \,a_{n-1}\, j_{n-1} +
  \left(\partial_r\beta_{n-1}+\frac{\beta_{n-1}}{r}\right)
  \right)\right) (s,R_n(s,t,r))\,\notag\\
  &\qquad\qquad\qquad\qquad\qquad\qquad\qquad\qquad\qquad\qquad
  \qquad\qquad\qquad
  \partial_r R_n(s,t,r)\,ds \notag\\
  &=:
  \beta'_{n,31}(t,r) + \beta'_{n,32}(t,r) + \beta'_{n,33}(t,r) .
\end{align}
The term $\beta'_{n,31}$ is continuously differentiable, where we have to
observe that $\partial_r^2 R_n$ exists and is continuous.
The term $\beta'_{n,33}$ is continuously differentiable by the properties of the
previous iterates. In the term $\beta'_{n,32}$ we integrate by parts
and observe that \eqref{ee4-red-n-1} holds:
\begin{align}\label{drbeta_n3crit}
  \beta'_{n,32}(t,r)
  &=
  - r a_{n-1}(t,r) j_n(t,r) +(r a_{n-1} j_n)(0,R_n(0,t,r))
  \partial_r R_n(0,t,r)\notag\\
  &
  +\int_0^t j_n(s,R_n(s,t,r))\,
  \Bigl[\partial_s R_n(s,t,r)\partial_r R_n(s,t,r) a_{n-1}(s,R_n(s,t,r)) \notag\\
  &
  \qquad + R_n(s,t,r)
  \partial_s\partial_r R_n(s,t,r) a_{n-1}(s,R_n(s,t,r)) \notag\\
  &
  \qquad - R_n(s,t,r) \partial_r R_n(s,t,r)
  \, 4 \pi (r a^2_{n-1} j_{n-1})(s,R_n(s,t,r))
  \Bigr]\, ds.
\end{align}
The first two terms are continuously differentiable.
The same is true for the first and last term under the integral.
The crucial term is $\partial_s\partial_r R_n(s,t,r)$.
First we note that
\[
\partial_s\partial_r^2 R_n(s,t,r) =
- \partial_r\bigl(\partial_r \beta_{n-1}(s,R_n(s,t,r))\partial_r R_n(s,t,r)\bigr)
\]
which exists and is continuous.
Now it holds that
\[
\partial_t R_n(s,t,r) = \partial_r R_n(s,t,r)\, \beta_{n-1}(t,r)
\]
which is continuously differentiable with respect to $s$ and $t$,
and this implies that also $\partial_s\partial_t\partial_r R_n(s,t,r)$
exists and is continuous.
Hence the term $\beta'_{n,32}$ is continuously differentiable with
respect to $t$ and $r$.
It follows that $\beta_n \in C^1([0,T_{n-1}[\times [0,\infty[)$ with
$\partial_r \beta_n \in C^1([0,T_{n-1}[\times [0,\infty[)$ as desired.
Also, $\beta_n(t,0) = 0 = \partial_r^2 \beta_n(t,0)$.
The assertion for $\beta_n(t,0)$ follows from the definition
\eqref{betandef}. For the assertion of its second order derivative
in $r$ we check all the terms which appear in that derivative and
observe that each of these contains at least one factor of the form
$R_n,\ \partial_s R_n,\ \partial_r^2 R_n,\ \partial_r^2[(1-A_n)/(2r)],\ j_n,
\ \partial_r p_n,\ \partial_r p_{T,n}$,
each of which vanishes at $r=0$.
Altogether, we have shown that $\beta_n$ is regular in the sense of
Definition~\ref{regular}~(e).
\subsubsection{The $n\,$th iterates for the metric quantities---$a_n$}
\label{sectionandef}
We define the next iterate $a_n$ by defining
$\lambda_n$ as the solution to the initial value
problem
\begin{align}\label{ee-red4-lambdan}
 \left(\partial_t -\beta_{n}\, \partial_r\right) e^{-\lambda} = 4\pi r j_n,\ 
e^{-\lambda(0,r)} = e^{-\mathring{\lambda}(r)}=1/\mathring{a}(r),
\end{align}
cf.\ \eqref{lambdadef} and \eqref{ee-red4-lambda}.
Note that we use the new iterate $\beta_n$ at this stage, but otherwise
the structure of this problem is the same as \eqref{betaeq}.
Since $\beta_n$ has the required regularity we can proceed as before
and use Lemma~\ref{lemmabetaeq} to define
$[0,T_{n-1}[\,\ni s\mapsto R_{n+1}(s,t,r)$
as the solution of the characteristic equation
\[
\dot r = - \beta_{n}(s,r)
\]
with $R_{n+1}(t,t,r)=r$,
and
\begin{align}
  e^{-\lambda_n(t,r)}
  &\coloneqq
  e^{-\mathring{\lambda}(R_{n+1}(0,t,r))}
  + 4\pi \int_0^t(r j_n)(s,R_{n+1}(s,t,r))\, ds
  \label{lambda_ndef}
\end{align}
on some maximal interval $[0,T_{n}[\subset [0,T_{n-1}[$
on which the right hand side is positive.
We let $a_n = e^{\lambda_n}$ and need to check that $a_n$ is 
regular, but this is fine from
\eqref{lambda_ndef}, since we only need
$a_n \in C^1([0,T_{n}[ \times [0,\infty[)$, i.e., the complicated maneuver to
get the second order derivatives of $\beta_n$ need not be repeated
for $a_n$. Since $R_{n+1}(s,t,0)=0$, it follows that $e^{-\lambda_n(t,0)}=1$
and hence $a_n(t,0)=1$, and $\partial_r a_n(t,0)=0$.
Finally,
\begin{align}\label{a_neqn}
\partial_t a_n - \beta_n \partial_r a_n = -4\pi r a_n^2 j_n,
\end{align}
which is exactly \eqref{ee4-red-n-1} required 
in the induction assumption,
but now for the next iterate. 
    
One iteration cycle is now complete. It needs to be checked that
$a_n, j_n, \beta_n$ are locally bounded on $[0,T_n[$, but this will be done
as part of the bounds to be established next. Since the local boundedness
of $\beta_n$ was used to define $a_n$ above,
it is important that that boundedness
can be proven, before $a_n$ is even defined, which will indeed be the case.
\subsection{Bounds on the iterates---loop 1}\label{loop1}
We fix two constants $P_0\geq 1, Q_0\geq 1$ such that
\[
\supp \mathring{f} \subset B_{Q_0}\times B_{P_0};
\]
$B_R$ denotes the closed ball in $\R^3$ with radius $R$ and center $0$.
In what follows, $C>0$ denotes a constant which may depend on
\[
P_0, Q_0, \|\mathring{f}\|, \|\mathring{a}\|, \|1/\mathring{a}\|,
\|\mathring{\beta}\|, \|\partial_r\mathring{ \beta}\|;
\]
$C$ does not depend on $t$ or $n$, and in what
follows it may change from line to line.
For $n\in \N$ and $t\in [0,T_{n-1}[$ we define
\begin{align}
  P_n(t)
  &\coloneqq
  \sup\bigl\{|V_n(s,0,z)| \mid z\in \supp \mathring{f},\
  \ 0 \leq s\leq t\bigr\},
  \label{Pn-def}\\
  Q_n(t)
  &\coloneqq
  \sup\bigl\{|X_n(s,0,z)| \mid z\in \supp \mathring{f},
  \ \ 0 \leq s\leq t\bigr\}.
  \label{qn-def}
\end{align}
We aim for a Gronwall-loop (``loop 1'') by which we can bound these
quantities and various others which naturally come up in the loop.
On the interval $[0,T_{n-1}[$,
\[
f_n(t,x,v) = 0\ \mbox{if}\ |x|\geq Q_n(t)\ \mbox{or}\ |v| \geq P_n(t),
\]
\[
\|\rho_n(t)\|,\ \|p_n(t)\|,\  \|p_{T,n}(t)\|,\  \|j_n(t)\| \leq
C P_n(t)^4,
\]
and
\[
\rho_n(t,x)=p_n(t,x)=p_{T,n}(t,x)=j_n(t,x)=0\ \mbox{if}\ |x|\geq Q_n(t).
\]
We use the characteristic equations
which define $X_n(s)$ and $V_n(s)$ to estimate $Q_n$ and $P_n$:
\begin{align}
  Q_n(t)
  &\leq C + \int_0^t \left(\|e^{-\lambda_{n-1}(s)}\| + \|\beta_{n-1}(s)\|\right) ds,
  \label{qnest}\\
  P_n(t)
  &\leq C + C \int_0^t
  \bigl(Q_{n-1}(s) \|a_{n-1}(s)\|\, P_{n-1}(s)^4
  + \|\partial_r\beta_{n-1}(s)\|\bigr) P_n(s)\, ds .
  \label{pnest}
\end{align}
Next \eqref{Andef} and \eqref{betandef} imply that
\begin{align}
  \left\|\frac{1}{r} ( 1-A_n)(t)\right\|
  &\leq
  C Q_n(t) P_n(t)^4 \left(1+\|a_{n-1}(t)\| \|\beta_{n-1}(t)\|\right),
  \label{Anest}\\
  \|\beta_n(t)\|
  &\leq \|\mathring{\beta}\| +
  C \int_0^t Q_n(s) P_n(s)^4 \left(1+\|a_{n-1}(s)\| \|\beta_{n-1}(s)\|\right) ds.
  \label{betanest}
\end{align}
Before we continue with our Gronwall argument we notice that the estimates
\eqref{qnest}--\eqref{betanest},
which we could have carried out before constructing
$a_n$ in Section~\ref{sectionandef}, show that $\beta_n$
satisfies the boundedness condition
which was used in constructing $a_n$, cf.~the discussion
at the end of Section~\ref{sectionandef}.

We continue with the Gronwall loop by estimating $\partial_r \beta_n$. 
From the discussion for the radial derivatives of the right had side
of \eqref{betandef} it follows that
\begin{align}\label{drbetanest}
 &\|\partial_r \beta_n(t)\| \leq
  C \biggl[
    \exp\left(\int_0^t\|\partial_r \beta_{n-1}(\tau)\| d\tau\right)
    + Q_n(t) P_n(t)^4 \|a_{n-1}(t)\|\notag \\
    &\quad {}+ \int_0^t P_n(s)^4 \biggl(1+\|a_{n-1}(s)\| \|\beta_{n-1}(s)\|\notag\\
    &\qquad + Q_n(s) \|a_{n-1}(s)\|
    \left(\|\partial_r\beta_{n-1}(s)\|+Q_n(s)\|a_{n-1}(s)\| P_{n-1}(s)^4\right)
    \biggr) ds\notag\\
    &\quad {}+ \int_0^t
    \exp\left(\int_s^t\|\partial_r \beta_{n-1}(\tau)\| d\tau\right) P_n(s)^4
    \biggl(1+\|a_{n-1}(s)\| \|\beta_{n-1}(s)\|\notag\\
    &\qquad + Q_n(s) \|a_{n-1}(s)\|
    \left(Q_{n-1}(s)\|a_{n-1}(s)\| P_{n-1}(s)^4
    +\|\partial_r\beta_{n-1}(s)\|\right)\biggr) ds;
\end{align}
notice that the term $\partial_r R_n(s,t,r)$, which arises when differentiating
\eqref{betandef}, is given by \eqref{drRnrep}.

Finally we bound $1/a_n = e^{-\lambda_n}$ and $a_n$. As to the former,
\eqref{lambda_ndef} implies that
\begin{align}\label{lambda_nest}
  \|e^{-\lambda_n(t)}\| \leq C \left(1+\int_0^t Q_n(s) P_n(s)^4 ds\right).
\end{align}
On the other hand, we can apply Lemma~\ref{lemmabetaeq}~(b) to the equation
\eqref{a_neqn} and obtain the estimate
\begin{align}\label{a_nest}
  \|a_n(t)\| \leq C \left(1+\int_0^t Q_n(s) P_n(s)^4 \|a_n(s)\|^2 ds\right).
\end{align}
The estimates \eqref{qnest}--\eqref{a_nest} form a closed Gronwall-type loop
which allows us to establish uniform control of the involved quantities.
We define
\[
  [0,T[\,\ni t\mapsto
      (z_Q(t), z_P(t), z_{\beta}(t), z_{\partial_r \beta}(t), z_{1/a}(t), z_a(t))
\]
as the maximal solution of the set of integral equations
\begin{align*}
  z_Q(t)
  &= C + \int_0^t \left(z_{1/a}(s) + z_\beta (s)\right) ds,\\
  z_P(t)
  &= C + C \int_0^t
  \left(z_Q(s) z_a(s) z_P(s)^4
  + z_{\partial_r\beta}(s)\right) z_P(s)\, ds,\\  
  z_\beta (t)
  &= C +
  C \int_0^t z_Q(s) z_P(s)^4 \left(1+z_a(s) z_\beta (s)\right) ds,\\
  z_{\partial_r \beta}(t)
  &=
  C \biggl[
    \exp\left(\int_0^t z_{\partial_r \beta} (\tau) d\tau\right)
    + z_Q(t) z_P(t)^4 z_a (t) \notag \\
    &\quad\qquad {}+ \int_0^t z_P(s)^4 \biggl(1+z_a(s) z_\beta (s)\\
    &\qquad\qquad + z_Q(s) z_a(s)
    \left(z_{\partial_r\beta}(s)+z_Q(s)z_a(s) z_P(s)^4\right)\biggr) ds\\
    &\quad\qquad {}+ \int_0^t
    \exp\left(\int_s^t z_{\partial_r \beta}(\tau) d\tau\right) z_P(s)^4
    \biggl(1+z_a(s) z_\beta (s)\\
    &\qquad\qquad
    + z_Q(s) z_a(s)
    \left(z_Q(s) z_a(s) z_P(s)^4
    +z_{\partial_r\beta}(s)\right)\biggr) ds \biggr],\\
  z_{1/a}(t)
  &= C \left(1+\int_0^t z_Q(s) z_P(s)^4 ds\right),\\
  z_a(t)
  &= C \left(1+\int_0^t z_Q(s) z_P(s)^4 z_a(s)^2 ds\right).
\end{align*}
It follows by induction that for all $n\in \N$, $T_n>T$, and for $t\in[0,T[$,
\[
Q_n(t)\leq Q_P(t),\ \ldots,\ \|a_n(t)\|\leq z_{a}(t).
\]
\subsection{Bounds on the iterates---loop 2}\label{loop2}
In order to see that the quantities which were
bounded in the previous subsection
do actually converge we need to bound $\partial_z F_n$. If we examine which new
terms this requires in addition to those already bounded in the previous loop,
we find that we in addition need to bound
\begin{align}\label{loop2aim}
  \partial_r a_n,\ \partial_r j_n,\ \partial_r^2 \beta_n .
\end{align}
If we recall \eqref{lambda_ndef}, it follows that
\begin{align*}
  e^{-\lambda_n(t,r)} \partial_r\lambda_n(t,r)
  &= -e^{-\mathring{\lambda}(R_{n+1}(0,t,r))}\mathring{\lambda}'(R_{n+1}(0,t,r))
  \partial_r R_{n+1}(0,t,r)\\
  &\quad
  {}- \int_0^t \partial_r(r j_n)(s,R_{n+1}(s,t,r)) \partial_r R_{n+1}(s,t,r)\, ds.
\end{align*}
If we recall that $a_n=e^{\lambda_n}$ and \eqref{drRnrep}, we see that
\begin{align}\label{dra_nest}
  \|\partial_r a_n(t)\| \leq C(t)\left(1+\int_0^t\|\partial_r j_n(s)\| ds\right),
\end{align}
where the $C(t)$ denotes a positive, continuous function defined on the interval
$[0,T[$ which depends on the $z$-functions introduced in the previous
subsection, and on $\|\mathring{a}'\|$.

In order to estimate $\partial_r^2 \beta_n$
we recall the split in \eqref{betandef}.
First we observe that by differentiating \eqref{drRnrep} once more we obtain the
estimate
\begin{align}\label{drrR_n}
  \left|\partial_r^2 R_n (s,t,r)\right|
  \leq C(t) \int_s^t \|\partial_r^2\beta_{n-1}(\tau)\|\, d\tau .
\end{align}
This immediately implies that
\begin{align}\label{drrbeta_n1}
  \|\partial_r^2 \beta_{n,1}(t)\|
  \leq C(t) \left(1+\int_0^t \|\partial_r^2\beta_{n-1}(\tau)\|\, d\tau\right).
\end{align}
In order to differentiate $\beta_{n,2}$ twice, we need to do this for
$(1-A_n)/(2r)$ first, recalling \eqref{Andef}:
\begin{align*}
  \partial_r\left(\frac{1-A_n(t,r)}{2r}\right)
  &=4\pi(\rho_n - a_{n-1} \beta_{n-1} j_n)(t,r)\\
  &\quad {} -
  \frac{8\pi}{r^3}\int_0^r s^2(\rho_n - a_{n-1} \beta_{n-1} j_n)(t,s)\,ds,
  \\
  \partial_r^2\left(\frac{1-A_n(t,r)}{2r}\right)
  &=4\pi \partial_r (\rho_n - a_{n-1} \beta_{n-1} j_n)(t,r)\\
  &\quad {} - \frac{8\pi}{r}
  \int_0^r \partial_r(\rho_n - a_{n-1} \beta_{n-1} j_n)(t,s)\,ds\\
  &\quad {} + \frac{24\pi}{r^4}\int_0^r s^2
  \int_0^s\partial_r(\rho_n - a_{n-1} \beta_{n-1} j_n)(t,\tau)\,d\tau\,ds.
\end{align*}
Using these identities it follows that
\begin{align}\label{drrbeta_n2}
 \|\partial_r^2 \beta_{n,2}(t)\|
 &\leq
 \int_0^t C(s)\Bigl( 1 + \|\partial_r \rho_n(s)\|
 + \|\partial_r j_n(s)\|  \notag \\
 &\qquad \qquad\qquad {} + \|\partial_r j_{n-1}(s)\|
 + \|\partial_r^2 \beta_{n-1}(s)\|\Bigr) ds.
\end{align}
In order to bound $\partial_r^2 \beta_{n,3}(t)$
we check all the terms which arose
when we showed that $\beta_{n,3}(t)$ is twice differentiable
in Section~\ref{beta_nsection}.
Collecting all the terms results in the estimate
\begin{align}\label{drrbeta_n3}
 \|\partial_r^2 \beta_{n,3}(t)\|
 &\leq
 \int_0^t C(s)\Bigl( 1 + \|\partial_r p_n(s)\| + \|\partial_r p_{T,n}(s)\|
 + \|\partial_r j_n(s)\|  \notag \\
 &\qquad \qquad\qquad {} + \|\partial_r j_{n-1}(s)\|
 + \|\partial_r^2 \beta_{n-1}(s)\|\Bigr) ds\\
 &\quad {}+ C(t)\Bigl( 1 + \|\partial_r j_n(0)\| + \|\partial_r j_n(t)\|\Bigr).
\end{align}
The estimates \eqref{drrbeta_n1},\eqref{drrbeta_n2},\eqref{drrbeta_n3} together
imply that
\begin{align}\label{drrbeta_n}
 \|\partial_r^2 \beta_n(t)\|
 &\leq
 C(t) \int_0^t \Bigl(\|\partial_r \rho_n(s)\| + \|\partial_r p_n(s)\|
 + \|\partial_r p_{T,n}(s)\|
 + \|\partial_r j_n(s)\|  \notag \\
 &\qquad \qquad\qquad {} + \|\partial_r j_{n-1}(s)\|
 + \|\partial_r^2 \beta_{n-1}(s)\|\Bigr) ds \notag\\
 &\quad {}+ C(t)\Bigl( 1 + \|\partial_r j_n(t)\|\Bigr).
\end{align}
As in \cite[Lemma 6.8]{GB} it follows that
\begin{align}\label{drrho_nest}
 \|\partial_r \rho_n(t)\|, \|\partial_r p_n(t)\|, 
 \|\partial_r p_{T,n}(t)\|, \|\partial_r j_n(t)\| \leq C(t) \|\partial_x f_n(t)\|,
\end{align}
and clearly,
\begin{align}\label{dzf_nest}
  |\partial_z f_n(t,z)| \leq \|D \mathring{f}\| |\partial_z Z_n(0,t,z)|;
\end{align}
we note that only characteristics which start in the support of $\mathring{f}$
need to be considered here.
We can now differentiate the characteristic system
$\dot z = F_{n-1}(s,z)$ satisfied by $Z_n(s,t,z)$ with
respect to $z$ so that with the abbreviation
$Z_n(s) = Z_n(s,t,z) = (X_n,V_n)(s,t,x,v)$,
\[
\partial_z \dot Z_n(s) =\partial_z F_{n-1}(s,Z_n(s)) \cdot \partial_z Z_n(s),
\]
and hence after applying a Gronwall argument,
\begin{align}\label{dZ_nest}
  |\partial_z Z_n(s,t,z)|
  &\leq \exp\left(\int_s^t
  \sup\{|\partial_z F_{n-1}(\tau,x,v)| \mid
  |v| \leq z_P(t)\}\, d\tau\right)\notag \\
  &\leq \exp\left(\int_s^t C(\tau)\left(1+\|\partial_r j_{n-1}(\tau)\|
  +\|\partial_r^2 \beta_{n-1}(\tau)\|\right) d\tau\right),
\end{align}
where the function $z_P$ was introduced in \ref{loop1} and $|V_n(s)|\leq z_P(t)$
for $0\leq s\leq t < T$. 
If we define
\[
S_n(t)\coloneqq \max_{0\leq k\leq n}\left(\|\partial_z f_k(t)\| +
\|\partial_r^2 \beta_k(t)\|\right),
\]
then the the estimates \eqref{drrbeta_n}, \eqref{drrho_nest}, \eqref{dzf_nest},
\eqref{dZ_nest} can be combined to yield
\begin{align}\label{S_nest}
  S_n(t) \leq C(t) \exp\left(C(t) \int_0^t S_n(s) ds\right)
\end{align}
on $[0,T[$, where $C(t)$ denotes a continuous, increasing function which depends
on the $z$-functions introduced in Section~\ref{loop1},
and on certain derivatives of the
initial data. Let 
$\xi \colon [0,T_1[\,\ni t\mapsto \xi(t)$ denote the maximal solution
to the integral equation
\begin{align}\label{xi}
  \xi(t) = C(t) \exp\left(C(t) \int_0^t \xi(s) ds\right);
\end{align}   
of course $T_1 \leq T$. Then for all $n\in \N$ and $t\in [0,T_1[$,
$S_n(t) \leq \xi(t)$, and hence we obtain uniform bounds for
the quantities stated in \eqref{loop2aim}.
\subsection{Convergence---loop 1}\label{loop1conv}
We fix some interval $[0,\delta]\subset [0,T_1[$ so that all the uniform
bounds established in the previous two section hold on $[0,\delta]$;
in what follows $C$ denotes a constant which is independent of $t\in [0,\delta]$
and $n\in \N$. Since
\begin{align*}
  |R_{n+1}(s,t,r) - R_n(s,t,r)|
  &\leq \int_s^t\!\!
  \left|\beta_n(\tau,R_{n+1}(\tau,t,r))
  -\beta_{n-1}(\tau,R_n(\tau,t,r)\right|\, d\tau\\
  &\leq
  \int_s^t \|\beta_n(\tau)-\beta_{n-1}(\tau)\| d\tau \\
  &\quad {}+ C \int_s^t
  \left|R_{n+1}(\tau,t,r)-R_n(\tau,t,r)\right|\, d\tau,
\end{align*}
so that by Gronwall's lemma,
\begin{align}\label{R_ndiff}
  |R_{n+1}(s,t,r) - R_n(s,t,r)|
  \leq C \int_0^t
  \|\beta_n(\tau)-\beta_{n-1}(\tau)\| d\tau,\ 0\leq s\leq t\leq \delta.
\end{align}
Combining this estimate with the formula \eqref{betandef} and the uniform bounds
for the various source terms which appear on the right hand side of that
equation it follows that for $n\geq 2$,
\begin{align}\label{beta_ndiff}
  \|\beta_n(t)-\beta_{n-1}(t)\|
  &\leq
  C \int_0^t\Bigl(\|f_n(s)-f_{n-1}(s)\| + \|a_{n-1}(s)-a_{n-2}(s)\| \notag\\
  &\qquad\qquad {}+ \|\beta_{n-1}(s)-\beta_{n-2}(s)\|\Bigr)\, ds .
\end{align}
Combining \eqref{R_ndiff} with \eqref{lambda_ndef} it follows that
\begin{align}\label{a_ndiff}
\|a_n(t)-a_{n-1}(t)\|
  \leq
  C \int_0^t\Bigl(\|f_n(s)-f_{n-1}(s)\| +
  \|\beta_{n}(s)-\beta_{n-1}(s)\|\Bigr)\, ds,
\end{align}
and combining this in turn with \eqref{beta_ndiff} yields
\begin{align}\label{beta_ndiff2}
  \|\beta_n(t)-\beta_{n-1}(t)\|
  &\leq
  C \int_0^t\Bigl(\|f_n(s)-f_{n-1}(s)\| + \|f_{n-1}(s)-f_{n-2}(s)\| \notag\\
  &\qquad\qquad {}+ \|\beta_{n-1}(s)-\beta_{n-2}(s)\|\Bigr)\, ds .
\end{align}
We must estimate differences like $f_n-f_{n-1}$. To this end we
choose $U>0$ such that $|V_n(s,t,x,v)| \leq U$ for all
$0\leq s\leq t\leq\delta$,
$n\in \N$, and $|v|\leq z_P(t)$. The estimates from the previous subsection
imply that
\[
|\partial_z F_n(s,x,v)|\leq C,\ 0\leq s\leq \delta,\ n\in \N,\ |v|\leq U.
\]
By the characteristic system for the $n\,$th and $(n-1)$st iterate
and a Gronwall argument it follows that
\begin{align*}
  |Z_n(s) - Z_{n-1}(s)|
  &\leq
  C \int_0^t\Bigl(\|a_{n-1}(s)-a_{n-2}(s)\| + \|\beta_{n-1}(s)-\beta_{n-2}(s)\|\\
  &\quad {}+ \|f_{n-1}(s)-f_{n-2}(s)\|
  + \|\partial_r\beta_{n-1}(s)-\partial_r\beta_{n-2}(s)\|\Bigr) ds,
\end{align*}
which together with \eqref{a_ndiff} implies that
\begin{align}\label{f_ndiff}
  \|f_n(t) - f_{n-1}(t)\|
  &\leq
  C \int_0^t\Bigl(\|f_{n-1}(s)-f_{n-2}(s)\| + \|\beta_{n-1}(s)-\beta_{n-2}(s)\|
  \notag\\
  &\qquad\qquad {}
  + \|\partial_r\beta_{n-1}(s)-\partial_r\beta_{n-2}(s)\|\Bigr) ds,
\end{align}
and in order to close the Gronwall loop
we are left with estimating $\partial_r\beta_{n}-\partial_r\beta_{n-1}$.
To do so we recall the split in \eqref{betandef} and the fact that
computing $\partial_r \beta_{n,1}$ and $\partial_r \beta_{n,2}$ causes no
problems. We first observe that by \eqref{drRrep},
\begin{align*}
  &|\partial_r R_n-\partial_r R_{n-1}|(s,t,r)\\
  &\leq
  C \int_0^t \left(\|\partial_r\beta_{n-1}(\tau)-\partial_r\beta_{n-2}(\tau)\|
  + |R_{n}-R_{n-1}|(\tau,t,r) \right)d\tau\\
  &\leq
  C \int_0^t \left(\|\partial_r\beta_{n-1}(\tau)-\partial_r\beta_{n-2}(\tau)\|
  + \|\beta_{n-1}(\tau)-\beta_{n-2}(\tau)\|\right)d\tau .
\end{align*}
Using this it is easy to see that
\begin{align}\label{b_n12diff}
  &\|\partial_r \beta_{n,1}(t) - \partial_r \beta_{n-1,1}(t)\|
  + \|\partial_r \beta_{n,2}(t) - \partial_r \beta_{n-1,2}(t)\|\notag \\
  &\leq
  C \int_0^t\Bigl(\|f_{n}(s)-f_{n-1}(s)\| +\|f_{n-1}(s)-f_{n-2}(s)\| +
  \|\beta_{n-1}(s)-\beta_{n-2}(s)\|
  \notag\\
  &\qquad\qquad {}
  + \|\partial_r\beta_{n-1}(s)-\partial_r\beta_{n-2}(s)\|\Bigr) ds.
\end{align}
When estimating the difference
$\partial_r \beta_{n,3} - \partial_r \beta_{n-1,3}$
the same terms will come up again
from the terms $\beta'_{n,31}$ and $\beta'_{n,33}$ in \eqref{drbeta_n3},
and the only qualitatively new term is
\[
\partial_s\partial_r R_n(s,t,r) - \partial_s\partial_r R_{n-1}(s,t,r),
\]
the one arising
after the integration by parts in $\beta'_{n,32}$,
cf.\ \eqref{drbeta_n3crit}. But
differentiating \eqref{drRrep} with respect to $s$ we see that the latter
difference term can be estimated by
\[
\|\partial_r \beta_{n}(t) - \partial_r \beta_{n-1}(t)\|+
|\partial_r R_n - \partial_r R_{n-1}|(s,t,r)+
|R_n - R_{n-1}|(s,t,r).
\]
Altogether, 
\begin{align}\label{drb_ndiff}
  &\|\partial_r \beta_{n}(t) - \partial_r \beta_{n-1}(t)\| \notag\\
  &\leq
  C \int_0^t\Bigl(\|f_{n}(s)-f_{n-1}(s)\| +\|f_{n-1}(s)-f_{n-2}(s)\|
  \notag\\
  &\qquad\qquad {}
  + \|\beta_{n-1}(s)-\beta_{n-2}(s)\|+
  \|\partial_r\beta_{n-1}(s)-\partial_r\beta_{n-2}(s)\|\Bigr) ds.
\end{align}
We define
\[
d_n(t)\coloneqq
\|f_{n}(t) - f_{n-1}(t)\| + \|\beta_{n}(t) - \beta_{n-1}(t)\|
+\|\partial_r \beta_{n}(t) - \partial_r \beta_{n-1}(t)\|
\]
and add the estimates \eqref{beta_ndiff2}, \eqref{f_ndiff}, \eqref{drb_ndiff}
to establish the estimate
\begin{align}\label{d_nest}
d_n(t) \leq C \int_0^t d_{n-1}(s)\, ds
\end{align}
which holds for $n\geq 2$ and $t\in [0,\delta]$; note that
$\|f_{n}(s) - f_{n-1}(s)\|$ can be eliminated from the right hand side
using \eqref{f_ndiff}.

It is now a standard argument to conclude form \eqref{d_nest} that
\[
f_n \to f,\ \beta_n \to \beta,\ \partial_r \beta_n \to \partial_r \beta
\]
uniformly in $t\in [0,\delta]$ and uniformly in the other relevant variables.
This immediately implies the convergences
\[
\rho_n \to \rho,\ a_n \to a,\ \frac{1}{a_n} \to \frac{1}{a},\
\frac{1-A_n}{2r} \to \frac{1-A}{2r},
\]
where $A$ is given by \eqref{Adef-red},
the other source terms converge as well, and
\[
R_n(s,t,r)\to R(s,t,r),\ \partial_r R_n(s,t,r)\to \partial_r R(s,t,r)
\]
where $s\mapsto R(s,t,r)$ solves $\dot r = -\beta(s,r)$ with $R(t,t,r)=r$.
The fact that
\[
\partial_t \beta_n - \beta_{n-1}\partial_r \beta_n
= \frac{1-A_n}{2r} + 4 \pi r p_n
\]
implies that also $\partial_t \beta_n \to \partial_t \beta$, in particular,
$\beta\in C^1([0,\delta]\times [0,\infty[)$ satisfies
\eqref{ee-red3} and $\beta(0)=\mathring{\beta}$.
\subsection{Convergence of $\partial_r a_n$---regularity of $a$}\label{convader}
The convergence of $\partial_r a_n$ does not follow directly from the
Gronwall loop in the previous subsection, it requires a separate,
non-trivial argument. First we observe that it suffices to prove the
convergence of $\partial_r \lambda_n$. If we recall the definition of
$\lambda_n$ in \eqref{lambda_ndef} and take a radial derivative, the first
term on the right hand side converges by the results from the previous
subsection, and we are left with proving the convergence of
\[
\int_0^t \partial_r\left[(r j_n)(s,R_{n+1}(s,t,r))\right]\, ds .
\]
From \cite[Lemma 6.6, Lemma 6.8]{GB} we can conclude that
\[
\partial_r(r j_n) = r \int v\cdot \partial_x f_n dv - j_n
\]
so that we are left with proving the convergence of
\[
\int_0^t \left(r \int v\cdot \partial_x f_n dv\right)(s,R_{n+1}(s,t,r))\,
\partial_r R_{n+1}(s,t,r)\, ds .
\]
By the Vlasov equation for $f_n$,
\begin{align*}
  & \int v\cdot \partial_x f_n dv
  = - \int a_{n-1} \p0
  \biggl( \partial_t f_n - \beta_{n-1} \frac{x}{r}\cdot \partial_x f_n\\
  &\quad
  {}+\left[4\pi r a_{n-1}j_{n-1} \frac{x\cdot v}{r}  \frac{x}{r}
    + \left(\partial_r\beta_{n-1}
    -\frac{\beta_{n-1}}{r}\right)\frac{x\cdot v}{r}  \frac{x}{r}
    + \frac{\beta_{n-1}}{r} v\right]\cdot \partial_v f_n\biggl) dv.
\end{align*}
In the term with $[\dots]$ we integrate by parts with respect to $v$,
and the resulting terms converge by the previous subsection.
In order to analyze the remaining term we observe that
\[
\frac{d}{ds} f_n(s,X_{n+1}(s,t,x),v) =
\left(\partial_t f_n -\beta_{n} \frac{x}{r}\cdot
\partial_x f_n\right)(s,X_{n+1}(s,t,x),v),
\]
where we need to recall the meaning of the characteristic curve $X_{n+1}(s,t,x)$
introduced via \eqref{X_nbeta}; note that $|X_{n+1}(s,t,x)|= R_{n+1}(s,t,r)$
if $r=|x|$.
Paying attention to the mismatch in the subscript of $\beta$ we are left with
\begin{align*}
  & \int_0^t (r a_{n-1})(s,R_{n+1}(s,t,r))
  \int \p0 \frac{d}{ds} f_n(s,X_{n+1}(s,t,x),v)\,dv \partial_r R_{n+1}(s,t,r) ds\\
  &\qquad
       {}+\int_0^t [\ldots]
       \left(\beta_n(s,R_{n+1}(s,t,r) -\beta_{n-1}(s,R_{n+1}(s,t,r)\right) ds,
\end{align*}
where the $[\ldots]$ contains various terms which are bounded by the results
obtained in Section~\ref{loop1}
and $\partial_x f_n$ which is bounded by Section~\ref{loop2}.
Since $\beta_n$ converges uniformly by the results from
Section~\ref{loop1conv}, the second integral can be dropped,
and we must show the convergence of the first one. In that term we integrate by
parts with respect to $s$. The resulting boundary terms converge by the results
in Section~\ref{loop1conv}, and the same is true for the two integrals
containing $\partial_s R_{n+1}(s,t,r)$ and
$\partial_s \partial_r R_{n+1}(s,t,r)$, and the remaining
integral which needs to converge is
\[
\int_0^t R_{n+1}(s,t,r) \frac{d}{ds} a_{n-1}(s,R_{n+1}(s,t,r))
\rho_n (s,R_{n+1}(s,t,r)) \partial_r R_{n+1}(s,t,r)\, ds.
\]
We observe that
\begin{align*}
  &\frac{d}{ds} a_{n-1}(s,R_{n+1}(s,t,r))
  =
  \partial_t a_{n-1}(\ldots) + \partial_s R_{n+1}(s,t,r)\partial_r a_{n-1}(\ldots)\\
  &=\left(\partial_t a_{n-1} - \beta_{n-1}\partial_r a_{n-1}\right)(\ldots)
  +\left(\beta_{n-1}-\beta_n\right)(s,R_{n+1}(s,t,r)) \partial_r a_{n-1}(\ldots) \\
  &=-\left(4\pi r a_{n-1}^2 j_{n-1}\right)(\ldots)
  +\left(\beta_{n-1}-\beta_n\right)(s,R_{n+1}(s,t,r)) \partial_r a_{n-1}(\ldots),
\end{align*}
where in the last step we used \eqref{a_neqn}.
Since $\partial_r a_{n-1}$ is uniformly bounded
by the results obtained in Section~\ref{loop2} and $\beta_n$
converges uniformly according to
Section~\ref{loop1conv}  the second term vanishes,
while the first now contains only terms which
converge according to Section~\ref{loop1conv}.
Hence $\partial_r a_n$ converges, uniformly
in $t$ and $r$, by  \eqref{a_neqn} the same is true for $\partial_t a_n$,
and hence $a\in C^1([0,\delta]\times [0,\infty[)$ satisfies \eqref{ee-red4}
and its initial condition.
\subsection{Convergence---loop 2} \label{magic}
The convergence results and ensuing regularity of the limit objects
is not yet sufficient for $f$ to solve the Vlasov equation.
To remedy this situation we consider the derivatives
of the characteristics $Z_n(s,t,z)$ with respect to $z$ and form certain
combinations of these in such a way that the coefficients of
the differential equations satisfied by these combinations have
better convergence and regularity properties than
the original equations. This maneuver is sometimes referred to as
the ``Magic Lemma''; cf.~\cite[Lemma~2.3]{GB}, where its differential
geometric background is explained as well.
In the present situation it takes the following form,
which is motivated by the particular combination of second order derivatives of
$a$ and $\beta$ which is contained in \eqref{EE22}.

As above, we denote by
$Z_{n+1}(s,t,z)=(X_{n+1},V_{n+1})(s,t,x,v) = (X_{n+1},V_{n+1})(s)$
the solution to the characteristic system
\begin{align}
  &
  \dot x =
  \frac{v}{a_n(s,r) \p0}-\beta_n(s,r)\frac{x}{r}, \label{xdot}\\
  &
  \dot v =
  4\pi r \,(a_n\, j_n)(s,r)\, \frac{x\cdot v}{r} \frac{x}{r}
  +\left(\partial_r\beta_n
  -\frac{\beta_n}{r}\right)(s,r)\frac{x\cdot v}{r} \frac{x}{r}
    +\frac{\beta_n(s,r)}{r} v \label{vdot}
\end{align}
with $Z_{n+1}(t,t,z)=z$. For $j\in \{1,\ldots,6\}$ fixed we define
\begin{align}
  \xi &\coloneqq \partial_{z_j} X_{n+1},\label{xidef}\\
  \eta &\coloneqq
  \partial_{z_j} V_{n+1}\notag \\
  &\quad{}
  - \left(4\pi r \,a_n^2\, j_n
  + a_n\left(\partial_r\beta_n-\frac{\beta_n}{r}\right)\right) (s,X_{n+1})
  \langle V_{n+1}\rangle \frac{X_{n+1}}{|X_{n+1}|}
  \cdot \xi \frac{X_{n+1}}{|X_{n+1}|}.
  \label{etadef}
\end{align}
Then
\begin{align}
  \dot \xi
  &=
  c_{1,n}(s,z_{n+1}(s)) \xi +  c_{2,n}(s,z_{n+1}(s)) \eta,\label{xidot}\\
  \dot \eta
  &=
  c_{3,n}(s,z_{n+1}(s)) \xi +  c_{4,n}(s,z_{n+1}(s)) \eta
  + c_{n}^\ast (s,z_{n+1}(s)) \xi,
  \label{etadot}
\end{align}
where $c_{1,n},\ldots,c_{4,n}$ are continuous and such that for any $U>0$
there exists a constant $C>0$ such that
\begin{align}\label{c_nbounds}
  & |c_{i,n}(s,z)| \leq C,\ i=1,\ldots,4,\ n\in \N,\
  (s,z)\in [0,\delta]\times \R^3 \times B_U,\\
 & |\partial_z c_{i,n}(s)| \leq C,\ i=1,\ldots,4,\ n\in \N,\
  (s,z)\in [0,\delta]\times (\R^3\setminus \{0\}) \times B_U,
\end{align}
and
\[
c_{i,n} \to c_i\ \mbox{uniformly on}\ [0,\delta]\times \R^3 \times B_U,\
i=1,\ldots,4;
\]
$c_{1,n},\ldots,c_{4,n}$ contain only terms the convergence
of which is already established.
This is not true for the critical coefficient 
\begin{align}\label{c_n5}
  c_{n}^\ast(s,x,v)\xi
  =
  \left(4\pi r a_n^2 (\beta_n \partial_r j_n - \partial_t j_n)
  + a_n (\partial_r^2 \beta_n - \partial_r\partial_t \beta_n)\right)
  \p0 \frac{x}{r}\cdot \xi \frac{x}{r}.
\end{align}
The justification of this maneuver is a straight forward computation
for which we introduce the notation $A\sim B$ to mean that $A - B$
contains only terms the convergence of which
was established in Sections~\ref{loop1conv} and \ref{convader};
the critical terms, which we collect into $c_{n}^\ast$,
are those containing first order derivatives of $j_n$ or second order
ones of $\beta_n$. If we substitute $Z_{n+1}$ into \eqref{xdot} and \eqref{vdot}
and then differentiate with respect to $z_j$, the $\partial_{z_j} X_{n+1}$
component of the resulting ``variational system'' obeyed by
$\partial_{z_j} Z_{n+1}$ is of the form \eqref{xidot}.
In order to obtain \eqref{etadot} 
we differentiate \eqref{etadef} with respect to $s$ and take into account
the variational system for $\partial_{z_j} Z_{n+1}$:
\begin{align*}
  \dot \eta
  &\sim
  \left(4\pi r \,a_n\, \partial_r j_n + \partial_r^2 \beta_n\right)
  \, \frac{x\cdot v}{r} \frac{x}{r}\cdot\xi\frac{x}{r}
  -\left(4\pi r \,a_n^2\, \partial_t j_n + a_n \partial_t\partial_r\beta_n\right)
  \p0 \frac{x}{r} \cdot \xi \frac{x}{r}\\
  &\quad{}
  - \left(4\pi r \,a_n^2\, \partial_r j_n + a_n \partial_r^2\beta_n\right)
  \left(\frac{\frac{x\cdot v}{r}}{a_n \p0}-\beta_n\right)
  \p0 \frac{x}{r} \cdot \xi \frac{x}{r}\\
  &\sim
  \left(4\pi r \,a_n^2\, \partial_t j_n + a_n \partial_t\partial_r\beta_n\right)
  \p0 \frac{x}{r} \cdot \xi \frac{x}{r}
  + \left(4\pi r \,a_n^2\, \partial_r j_n + a_n \partial_r^2\beta_n\right) \beta_n
  \frac{x}{r}\cdot\xi\frac{x}{r}\\
  &= c_{n}^\ast \xi
\end{align*}
as desired.

It remains to investigate the critical coefficient $c_{n}^\ast$,
and we start with the term containing the derivatives of $\beta_n$.
The idea is to use the fact that the equation
\[
\partial_t\beta_n-\beta_{n-1} \partial_r\beta_n = \frac{1}{2r}(1-A_n)+4\pi r p_n,
\]
holds by construction, cf.~Section~\ref{beta_nsection},
and it can be differentiated with respect
to $r$. Comparing with the definition of $c_{n}^\ast$ there
is a mismatch in the subscripts,
but this can be remedied as follows:
\begin{align*}
  \partial_t\partial_r\beta_n-\beta_{n} \partial_r^2\beta_n
  &= \partial_t\partial_r\beta_n-\beta_{n-1} \partial_r^2\beta_n
  + (\beta_{n-1}-\beta_n)\partial_r^2\beta_n\\
  &= \partial_r\left(\partial_t\beta_n-\beta_{n-1} \partial_r\beta_n\right) +
  \partial_r\beta_{n-1}\partial_r\beta_n +
  (\beta_{n-1}-\beta_n)\partial_r^2\beta_n\\
  &\sim 4\pi r \partial_r p_n + (\beta_{n-1}-\beta_n)\partial_r^2\beta_n,
\end{align*}
where we used the fact that $\partial_r\left((1-A_n)/(2r)\right)$
does converge by \ref{loop1conv}, which is obvious from the relation
\eqref{Andef}. Since $\partial_r^2\beta_n$ is bounded by Section~\ref{loop2}
and $\beta_n$ converges by Section~\ref{loop1conv}
it follows that $(\beta_{n-1}-\beta_n)\partial_r^2\beta_n\to 0$,
uniformly in $t$ and $r$, and we are left with the expression
\[
4\pi r\left(a_n \beta_n \partial_r j_n - a_n \partial_t j_n
- \partial_r p_n\right)
\]
as the critical part of $c_n^\ast$. In order to proceed we express
$\partial_t j_n$ via Lemma~\ref{dtrhodtj} to find that
\begin{align*}
  & r\left(a_n \beta_n \partial_r j_n - a_n \partial_t j_n
  - \partial_r p_n\right)\\
  &=
  2 \left(p_n - p_{T,n}\right) 
  - 2 r a_n
  \left(\partial_r \beta_{n-1} -\frac{\beta_{n-1}}{r}
  + 4\pi r a_{n-1} j_{n-1}\right) j_n\\
  &\quad {}
  + r a_n (\beta_n - \beta_{n-1}) \partial_r j_n
  + r \left(\frac{a_n}{a_{n-1}}-1\right)
  \left(\partial_r p_n +\frac{2}{r} p_n - \frac{2}{r} p_{T,n}\right).
\end{align*}
The terms $\partial_r j_n$ and $\partial_r p_n$
are uniformly bounded by Section~\ref{loop2},
and the difference factors
in front of them converge to $0$ uniformly by Section~\ref{loop1conv}.
The remaining terms do not
contain any ``forbidden'' derivatives; they converge uniformly by
Section~\ref{loop1conv},
and their radial derivatives are uniformly bounded by Section~\ref{loop2}.

It is now straight
forward to conclude that $\xi=\xi_n$ and $\eta=\eta_n$ converge, and hence
$\partial_{z_j} X_n$ and $\partial_{z_j} V_n$ converge, uniformly on
$[0,\delta]\times [0,\delta]\times \R^3 \times B_U$, where $U>0$ was arbitrary.
This implies that $Z\in C^1([0,\delta]\times [0,\delta]\times \R^6)$
which in turn implies the missing regularity for $f$,
and the limiting objects $f, a, \beta$
do constitute a regular solution to the initial value problem as claimed in
Theorem~\ref{locex}.

Uniqueness of this solution follows if we apply the Gronwall-type estimates
established in Section~\ref{loop1conv} for the
difference of two consecutive iterates  to the difference of two
solutions with the same data. This unique solution can be extended to a maximal
existence interval $[0,T[$, and it remains to prove the continuation
criterion contained in Theorem~\ref{locex}.

We still need to check the boundary conditions for the metric at infinity.
As above, let $[0,T[\,\ni s\mapsto R(s,t,r)$ denote the unique solution
of the characteristic equation
\begin{align}\label{charactbeta}
  \dot r = - \beta (s,r)
\end{align}
with $R(t,t,r)=r$; these are the characteristics of both
\eqref{ee-red3} and \eqref{ee-red4}. We fix some
$T'<T$. Then $\beta$ is bounded on $[0,T']\times [0,\infty[$,
which for the characteristics implies that
\[
|R(s,t,r)-r|\leq C,\ r\geq 0,\ 0\leq s,t\leq T',
\]
and hence
\[
\lim_{r\to\infty} R(s,t,r) = \infty
\]
uniformly in $s,t\in [0,T']$. Since by Lemma~\ref{lemmabetaeq},
\[
a(t,r)=\mathring{a}(R(0,t,r))
- 4\pi \int_0^t \left(r a^2 j\right)(s,R(s,t,r))\, ds,
\]
the uniformly-in-$t$ compact support of $j$ and the boundary condition
$\mathring{a}(\infty)=1$ imply that $a(t,\infty)=1$ as required.
Next we note that
the compact support property and \eqref{Adef-red}
immediately imply that $A(t,\infty)=1$,
and by the relation \eqref{Adef}, $\beta(t,\infty)=0$ as desired.
\subsection{The continuation criterion}
Let us assume now that the existence interval $[0,T[$ of our local solution
is maximal with $T<\infty$,
and that the continuation condition \eqref{contcrit} holds.
Since $f$ is constant along characteristics this implies that there is a
constant $C^\ast >0$ such that for all $t\in[0,T[$,
\[
f(t,x,v) = 0 \ \mbox{if}\ |v| \geq C^\ast,\ \|f(t)\|, \|a(t)\|,
\|\partial_r \beta(t)\|
\leq C^\ast .
\]
The bound of $\partial_r \beta$ and the boundary condition
$\beta(t,0)=0$ imply that $|\beta(t,r)| \leq C^\ast r$.
Hence by the first component of the characteristic system,
\[
|\dot x(s)| \leq \|1/a(s)\| + C^\ast |x(s)|.
\]
For
\[
Q(t)\coloneqq \sup\{|X(s,0,z)| \mid z\in \supp \mathring{f},\ \ 0 \leq s\leq t\}
\]
this implies that
\[
Q(t) \leq e^{C^\ast T}\left(R_0 + \int_0^t \|1/a(s)\| ds\right).
\]
Next by Lemma~\ref{lemmabetaeq},
\begin{align}\label{1overa}
  (1/a)(t,r)=
  (1/\mathring{a})(R(0,t,r)) + 4\pi \int_0^t(r j)(s,R(s,t,r))\, ds,
\end{align}
where $R(s,t,r)$ are the characteristics defined above, cf.~\eqref{charactbeta}.
Hence
\[
\|1/a(t)\| \leq
\|1/\mathring{a}\| + C \int_0^t Q(s)\, ds, 
\]
and by these estimates for $Q$ and $1/a$ it follows that
\[
\|1/a(t)\|, Q(t) \leq C^\ast,\ t\in [0,T[,
\]
after suitably increasing $C^\ast$. But now \eqref{ee-red3} together
with \eqref{Adef-red} directly yield a Gronwall inequality for $\beta$:
\[
|\partial_t\beta (t,r)| \leq C^\ast |\beta(t,r)| + C + C \|\beta(t)\|,
\]
where $C$ is given in terms of $C^\ast$. Hence 
increasing $C^\ast$ again, also the bound
\[
\|\beta(t)\| \leq C^\ast,\ t\in [0,T[
\]
holds.

In order to extend the solution we also have to bound various derivatives.
First we notice that we can repeat the
arguments from Section~\ref{convader},
which lead to the convergence of $\partial_r a_n$,
on the level of our local solution. For the latter,
we do not get any of the ``error''
terms containing differences of two consecutive iterates,
and the remaining terms
which arise are bounded by the continuation assumption and
the bound which we already
derived. Thus $\partial_r a$ remains bounded on $[0,T[$.

Next we repeat the ``Magic Lemma'' arguments from Section~\ref{magic},
but again on the level
of the local solution. This means that in the equations corresponding to
\eqref{xidot} and \eqref{etadot} we drop the subscripts,
and in the analysis of the
corresponding coefficients we can drop all the ``error'' terms which contain
a difference of two consecutive iterates.
But these terms were indeed the only ones
which also contained derivatives which we have not yet bounded,
so that the actual
coefficients are bounded on $[0,T[$. This implies that $\|\partial_z f(t)\|$
and hence also $\|\partial_r^2 \beta(t)\|$ remain bounded.

Now we pick some $t_0 \in ]0,T[$ close to $T$ and consider the new initial value
problem, where we prescribe $f(t_0)$, $a(t_0)$, $\beta(t_0)$ as data at $t=t_0$;
these data are admissible in the sense of our local existence result,
which gives us a solution on some interval $[t_0,t_0+\delta]$.
If we now examine which parameters of the data determined the length $\delta$
we find that all these parameters are uniformly bounded on the interval $[0,T[$.
Thus we may pick $\delta>0$ independently of $t_0$, and if the latter is chosen
close enough to $T$, we have extended the solution beyond $T$,
in contradiction to     $[0,T[$ being the maximal existence interval.

The proof of Theorem~\ref{locex} is complete.
\subsection{Regularity of the metric}\label{regmetric}
On general principles and in particular in view of the discussion
of \eqref{EE22}
in Proposition~\ref{fullsystem} and the remark after its proof it
is desirable that
$a, \beta \in C^2([0,T[\times [0,\infty[)$. For $\beta$ this follows from
Definition~\ref{regular}~(e) and the fact that by \eqref{ee-red3} also
$\partial^2_t \beta$ exists and is continuous. We need to check the
regularity of $a$.
\begin{cor}\label{corareg}
  In addition to the general assumptions on the data
  $(\mathring{f}, \mathring{a},\mathring{\beta})$ specified
  in Section~\ref{data},
  let $\mathring{a}\in C^2([0,\infty[)$.
  Then $a\in C^2([0,T[\times [0,\infty[)$.
\end{cor}
\begin{proof}
  The argument basically follows that of Section~\ref{convader},
  but now on the level
  of a solution. We represent $\partial_r a$ by differentiating
  \eqref{1overa} and must show that
  \[
  \int_0^t \partial_r\left[(r j)(s,R(s,t,r))\right]\, ds
  \]
  is $C^1$. Using the relation
  \[
  \partial_r (r j) = r \int v\cdot \partial_x f dv - j,
  \]
  the Vlasov equation and integration by parts on the term
  containing $\partial_v f$
  the crucial term concerning regularity turns out to be
  \begin{align*}
    &
    \int_0^t\left(r a \int \p0
    \left(\partial_t f - \beta \frac{x}{r}\cdot\partial_x f\right)\,
    dv\right)(s,R(s,t,r)) \partial_r R(s,t,r)\, ds\\
    &\qquad =
    \int_0^t R(s,t,r) a(s, R(s,t,r)) \frac{d}{ds} \rho(s,R(s,t,r))
    \,\partial_r R(s,t,r)\, ds.
  \end{align*}
  We integrate by parts and have to check that the resulting terms are $C^1$.
  First we note that
  \[
  \frac{d}{ds} a (s,R(s,t,r))
  =\left(\partial_t a -\beta \partial_r a\right)(s,R(s,t,r))
  = -4\pi \left(r a^2 j\right)(s,R(s,t,r))
  \]
  is indeed $C^1$, and so are all the other terms which arise,
  where we recall the
  regularity discussion for $\beta_n$ in Section~\ref{beta_nsection}.
  Thus $\partial_r a$ is $C^1$, and the remaining derivative
  $\partial_{tt} a$ exists and is continuous due to \eqref{ee-red4}.
\end{proof}
\section{Homogeneous solutions} \label{sec_hom}
\setcounter{equation}{0}
\subsection{The dust case} \label{subsec_homdust}
In their classical paper \cite{OS} on continued gravitational collapse
and the formation of a black hole {\sc Oppenheimer} and {\sc Snyder}
used dust as matter model. Dust can be viewed as an ideal, compressible
fluid with pressure identically zero, or as a singular special case
of collisionless Vlasov matter which is $\delta$-distributed with respect to
the momenta. The latter viewpoint is particularly suitable in the present
context, since for any characteristic of the Vlasov equation
\eqref{Vl-cart}, $v(t)=0$ for all $t$ if $v(0)=0$. Thus
\[
f(t,x,v)=\rho(t,x)\, \delta(v)
\]
is a consistent ansatz for a solution of the Einstein-Vlasov system
in Painlev\'{e}-Gullstrand coordinates. It implies that
\[
p = p_T = j = 0
\]
and reduces the system to the equations
\[
\partial_t \rho - \beta\, \frac{x}{r} \cdot \partial_x \rho -
\left(\partial_r \beta + \frac{2}{r} \beta\right) \rho = 0,
\]
\[
1- A = 1 - \frac{1}{a^2} + \beta^2
= \frac{8\pi}{r} \int_0^r s^2 \rho(t,s)\, ds,
\]
\[
\partial_t a - \beta\, \partial_r a = 0.
\]
We restrict ourselves to the so-called marginally bound
case $a=1$, which is consistent with the
boundary conditions \eqref{bc} and rewrite the equation for $\rho$
in radial coordinates to obtain the equations
\begin{align}\label{rhoeq_dust}
\partial_t \rho - \beta\, \partial_r \rho -
\left(\partial_r \beta + \frac{2}{r} \beta\right) \rho = 0,
\end{align}
\begin{align}\label{betamrel_dust}
\beta^2(t,r)
= \frac{2 m(t,r)}{r} = \frac{8\pi}{r} \int_0^r s^2 \rho(t,s)\, ds,
\end{align}
which in turn are equivalent to a single equation for the mass
function $m$:
\begin{align}\label{meq_dust}
\partial_t m - \sqrt{\frac{2 m}{r}}\, \partial_r m = 0,
\end{align}
where we choose the positive root in \eqref{betamrel_dust}.
If $\mathring{m}$ denotes the initial data for the mass function, i.e.,
\[
m(0,r) = \mathring{m}(r) = 4\pi \int_0^r s^2 \mathring{\rho}(s)\, ds,
\]
then by standard maneuvers for first order PDEs and up to regularity
issues \eqref{meq_dust} is equivalent to the implicit equation
\begin{align}\label{impl_meq_dust}
  m(t,r)
  = \mathring{m}
  \left(\left(r^{3/2} + \frac{3}{2} \sqrt{2 m(t,r)} t\right)^{2/3}\right).
\end{align}
For the special case when $\mathring{\rho}$ is constant,
\[
\mathring{\rho}(r) = \alpha,\ \mbox{i.e.,}\
\mathring{m}(r) = \frac{4\pi}{3} \alpha r^3
\]
with some positive amplitude $\alpha >0$ we find that
\begin{align}\label{m_dust}
  m(t,r) = \frac{4\pi}{3} \alpha \frac{r^3}{(1-\sqrt{6\pi \alpha} t)^{2}},\
  r\geq 0,\ 0\leq t < \frac{1}{\sqrt{6\pi \alpha}}.
\end{align}
This implies that
\begin{align}\label{beta_dust}
  \beta(t,r) = \sqrt{\frac{8\pi}{3} \alpha}
  \frac{r}{1-\sqrt{6\pi \alpha} t},\
  r\geq 0,\ 0\leq t < \frac{1}{\sqrt{6\pi \alpha}}.
\end{align}
We introduce the abbreviation
\begin{align}\label{sigma_dust}
  \sigma(t) \coloneqq
  \sqrt{\frac{8\pi}{3} \alpha} \frac{1}{1-\sqrt{6\pi \alpha} t}
\end{align}
so that $\beta(t,r)=\sigma(t) r$. Moreover,
\begin{align}\label{gamma_dust}
  \exp\left(-\int_0^t \sigma (\tau)\, d\tau\right)
  = (1-\sqrt{6\pi \alpha} t)^{2/3} =: \gamma(t),\
  0\leq t < \frac{1}{\sqrt{6\pi \alpha}}.
\end{align}
Hence the trajectories of the dust particles, which are
the characteristics of the equations \eqref{rhoeq_dust} and \eqref{meq_dust},
are given as
\begin{align}\label{dustpart}
  [0,1/\sqrt{6\pi \alpha}[\, \ni s \mapsto \gamma(s) r,\
  r\geq 0.
\end{align}
In passing we notice that the scale function $\gamma$
is the unique solution of the initial value problem
\begin{align}\label{gammaeq_dust}
  \ddot \gamma = - \frac{4\pi}{3} \alpha \frac{1}{\gamma^2},\
  \gamma(0)=1,\ \dot \gamma(0) = - \sqrt{\frac{8\pi}{3} \alpha},
\end{align}
which will become important when we relate the homogeneous dust
solution to the corresponding solution with Vlasov matter.

The solution obtained to far is not asymptotically flat and does not
represent an isolated, collapsing matter distribution, but
we can cut it along any of the particle trajectories in \eqref{dustpart}
and extend it by vacuum. We choose to cut along the trajectory
starting at $r=1$ so that
\[
\mathring{\rho}(r) = \alpha 1_{[0,1]}(r),\ \mbox{i.e.,}\
\mathring{m}(r) = \frac{4\pi}{3} \alpha
\left\{ \begin{array}{ccl} r^3 &,& 0\leq r \leq 1,\\
  1&,& r>1,
  \end{array}\right.
\]
and
\begin{align}\label{OS_massf}
m(t,r) = \frac{4\pi}{3} \alpha
\left\{ \begin{array}{ccl}
  \displaystyle
  \frac{r^3}{(1-\sqrt{6\pi \alpha} t)^2} &,& 0\leq r\leq \gamma(t),\\
  1 &,& r> \gamma(t),
  \end{array}\right.
\ 0\leq t < \frac{1}{\sqrt{6\pi \alpha}},
\end{align}
\begin{align}\label{OSdustbeta}
\beta(t,r) = \sqrt{\frac{8\pi}{3} \alpha}
\left\{ \begin{array}{ccl}
  \displaystyle
  \frac{r}{1-\sqrt{6\pi \alpha} t}&,& 0\leq r\leq \gamma(t),\\
  \displaystyle\frac{1}{\sqrt{r}}&,& r> \gamma(t),
\end{array}\right.
\ 0\leq t < \frac{1}{\sqrt{6\pi \alpha}},
\end{align}
and in \eqref{dustpart} only dust particles with
$0\leq r \leq 1$ are present. In passing we note that
this cutting procedure works because we cut along a characteristic
curve of the relevant equations which is not crossed by other
characteristics; it is straight forward to check that the
truncated solution satisfies the Einstein-dust system in a
distributional sense. Alternatively, one can check that the
mass function \eqref{OS_massf} satisfies the implicit
relation \eqref{impl_meq_dust} to which we have reduced
the Einstein-dust system above.

From the equation for radial null geodesics we see that a surface
given by some pair $(t,r)$ is trapped if $\beta(t,r)>1$. This holds
if
\[
0<r<\frac{8\pi}{3} \alpha\ \mbox{and}\
t > \frac{1}{\sqrt{6\pi\alpha}} - \frac{2}{3} r. 
\]
The earliest (marginally) trapped surface forms at
\[
r = \frac{8\pi}{3}\alpha,\ 
t = \frac{1}{\sqrt{6\pi\alpha}}
\left(1-\left(\frac{8\pi}{3}\alpha\right)^{3/2}\right),
\]
and we shall choose $\alpha<\frac{3}{8\pi}$ so that this happens to the future
of the initial hypersurface.
\subsection{The Vlasov case} \label{subsec_homvlasov}
In order to show that Oppenheimer-Snyder type collapse also happens with
Vlasov matter we first need to establish solutions of the Einstein-Vlasov
system which are spatially homogeneous and resemble the homogeneous
dust solutions obtained in the previous section.
We demand that
\[
j=0,\ a=1,\ \beta(t,r) = \sigma(t) r.
\]
From the characteristic system of the Vlasov equation
\eqref{Vl-cart} we get that
\[
\dot v = \sigma(s) v,
\]
and hence
\begin{align*}
  &\frac{d}{ds}
  \left(\exp\left(-2\int_0^s \sigma(\tau)\, d\tau\right) |v|^2\right)\\
  &\qquad \qquad =
  \exp\left(-2\int_0^s \sigma(\tau)\, d\tau\right)
  \left(-2 \sigma(s) |v|^2 + 2 v\cdot \dot v \right) =0
\end{align*}
along characteristics. As before, let
\[
\gamma(t) = \exp\left(-\int_0^t \sigma(\tau)\, d\tau\right), 
\]
and define a spatially homogeneous particle distribution
\[
h(t,x,v) = H(\gamma^2(t) |v|^2),
\]
where $H\in C^1([0,\infty[)$ is non-negative and $H(\eta)=0$ for $\eta$
large. Then the Vlasov equation holds, and we have to determine
$\sigma$ respectively $\gamma$ such that
the field equations hold as well.
Now
\[
\rho(t) = 4\pi \gamma(t)^{-4} \int_0^\infty \sqrt{\gamma^2(t) + u^2}
H(u^2)\, u^2 du
\]
and
\[
p(t) = \frac{4\pi}{3} \gamma(t)^{-4} \int_0^\infty
\frac{u^2}{\sqrt{\gamma^2(t) + u^2}}
H(u^2)\, u^2 du .
\]
Since
\[
A(t,r) = 1-r^2 \sigma^2(t)
\]
\eqref{ee1} becomes
\begin{align} \label{ee1h}
\dot \sigma = 4\pi (\rho + p),
\end{align}
\eqref{ee2} becomes
\begin{align}\label{ee2h}
\sigma^2 = \frac{8\pi}{3} \rho,
\end{align}
and \eqref{ee3} becomes
\begin{align}\label{ee3h}
\dot\sigma = \frac{3}{2} \sigma^2 + 4\pi p;
\end{align}
\eqref{ee4} is satisfied identically.
If \eqref{ee2h} holds, then \eqref{ee1h} and \eqref{ee3h} become equal.
Hence we proceed as follows. We define $\sigma$ as the solution of
\eqref{ee1h} with initial condition
\[
\sigma(0) = \sqrt{\frac{8 \pi}{3} \rho(0)};
\]
note that $\gamma(0)=1$, so $\rho(0)$ is completely determined by $H$.
We need to check that now \eqref{ee2h} holds for all times for which
$\sigma$ exists. Indeed,
\begin{align*}
\frac{d}{dt}\Bigl(\sigma^2 -& \frac{8\pi}{3} \rho\Bigr)
  =
  2 \sigma \dot \sigma
  - \frac{8\pi}{3} \dot \gamma
  \left[-\frac{4\rho}{\gamma} + \frac{4\pi}{\gamma^4}
  \int_0^\infty\!\!\frac{\gamma}{\sqrt{\gamma^2 + u^2}} H(u^2)\, u^2 du
  \right]\\
  =&
  2 \sigma \dot \sigma
  - \frac{8\pi}{3} \dot \gamma
  \frac{4\pi}{\gamma^4}\int_0^\infty\!\!
  \left( \frac{\gamma}{\sqrt{\gamma^2 + u^2}}
  - 4\frac{\sqrt{\gamma^2 + u^2}}{\gamma}\right) H(u^2)\, u^2 du\\
  =&
  8 \pi \sigma (\rho + p)\\
  &
  {}- 8 \pi \sigma \frac{4\pi}{3}
  \frac{1}{\gamma^4}\int_0^\infty\!\!
  \left(3 \sqrt{\gamma^2 + u^2} +
  \frac{u^2}{\sqrt{\gamma^2 + u^2}}\right) H(u^2)\, u^2 du\\
  =&
  0.
\end{align*}
We can therefore use \eqref{ee1h} as our master equation which determines
$\sigma$, but we turn it into an equation for $\gamma$ as follows.
Since $\dot \gamma = -\sigma \gamma$ we find that
\[
\ddot \gamma = -\dot \sigma \gamma - \sigma \dot \gamma
= -4\pi (\rho + p) \gamma + \frac{(\dot \gamma)^2}{\gamma};
\]
note that $\rho$ and $p$ actually depend on $\gamma$. We can in
addition eliminate $\dot \gamma$ using the relation
$\dot \gamma = -\sigma \gamma$ and \eqref{ee2h}. This results in
\begin{align} \label{gammaeq_vl}
  \ddot \gamma = -\frac{4\pi}{3}
  \frac{\phi(\gamma) + 3 \psi(\gamma)}{\gamma^3},
\end{align}
where
\begin{align*}
\phi(\xi)
&=
4\pi \int_0^\infty \sqrt{\xi^2 + u^2} H(u^2)\, u^2 du,\\
\psi(\xi)
&=
\frac{4\pi}{3} \int_0^\infty \frac{u^2}{\sqrt{\xi^2 + u^2}} H(u^2)\, u^2 du;
\end{align*}
these functions belong to $C^1([0,\infty[)$.
Eqn.~\eqref{gammaeq_vl} is supplemented with the initial conditions
\[
\gamma(0)=1,\ \dot \gamma(0) = - \sqrt{\frac{8\pi}{3} \rho(0)}.
\]
\subsection{Relating homogeneous Vlasov solutions to homogeneous dust solutions}
\label{subsec_homvleps}
In what follows we need to relate homogeneous Vlasov solutions
to certain homogeneous dust solutions. To this end, let $\gamma_D$ be
the solution to
\begin{align}\label{gammaD_eq}
  \ddot \gamma = - \frac{2}{9}  \frac{1}{\gamma^2},\
  \gamma(0)=1,\ \dot \gamma(0) = - \frac{2}{3},
\end{align}
i.e., we choose $\alpha=1/(6\pi)$ in \eqref{gammaeq_dust}, and
\begin{align} \label{gammaD}
  \gamma_D(t) = (1-t)^{2/3},\ 0\leq t < 1.
\end{align}
Quantities related to the corresponding homogeneous dust solution
will be denoted by $\beta_D$, $m_D$ etc.
In addition, we define for $\epsilon>0$ small
\[
\alpha_\epsilon \coloneqq \frac{1}{6\pi} (1-\epsilon^{1/5}).
\]
Then
\begin{align} \label{gammaDe}
  \gamma_{D,\epsilon}(t) = (1-\sqrt{1-\epsilon^{1/5}} t)^{2/3},\ 0\leq t <
  \frac{1}{\sqrt{1-\epsilon^{1/5}}}
\end{align}
is the solution to \eqref{gammaeq_dust} with $\alpha_\epsilon$ substituted
for $\alpha$, and notations like $\beta_{D,\epsilon}$, $m_{D,\epsilon}$ etc
will again refer to the corresponding homogeneous dust solution. 

For the Vlasov case we redefine the homogeneous particle distribution as
\begin{align}\label{hepsdef}
  h(t,x,v) = h_{\epsilon} (t,x,v)
  = \alpha_\epsilon \epsilon^{-3/2} H(\epsilon^{-1}\gamma_{V,\epsilon}(t)^2 |v|^2),
\end{align}
for some fixed function
\begin{align}\label{Hdef}
  H\in C^1([0,\infty[),\ H\geq 0,\ \mbox{non-increasing},\ 
  \!\int\!\! H(|v|^2) dv = 1,\ H=0\ \mbox{on}\ [1,\infty[.
\end{align}
This leads to a
homogeneous solution of the Einstein-Vlasov system, provided
$\gamma_{V,\epsilon}$ is the unique solution to
the initial value problem
\begin{align} \label{gammaeq_vl_eps}
  \ddot \gamma = -\frac{4\pi}{3}
  \frac{\phi_{\epsilon}(\gamma) +
    3 \psi_{\epsilon}(\gamma)}{\gamma^3},\
  \gamma(0)=1,\ \dot \gamma(0) = - \sqrt{\frac{8\pi}{3} \mathring\rho_{\epsilon}},
\end{align}
where
\begin{align}
\phi_{\epsilon}(\xi)
&=
4\pi \alpha_\epsilon \int_0^\infty \sqrt{\xi^2 + \epsilon u^2} H(u^2)\, u^2 du,
\label{phi_eps}\\
\psi_{\epsilon}(\xi)
&=
\frac{4\pi}{3} \alpha_\epsilon \epsilon
\int_0^\infty \frac{u^2}{\sqrt{\xi^2 + \epsilon u^2}} H(u^2)\, u^2 du,
\label{psi_eps}
\end{align}
and
\[
\mathring\rho_{\epsilon} =
4\pi \alpha_\epsilon \int_0^\infty \sqrt{1 + \epsilon u^2} H(u^2)\, u^2 du.
\]
For future reference we note that the mass-energy density
induced by $h$ is given as
\begin{align}\label{rho_h}
  \rho_h(t)=\frac{1}{\gamma^4_{V,\epsilon}(t)}\phi_\epsilon(\gamma_{V,\epsilon}(t)).
\end{align}
The following comparison result will be useful.
\begin{lemma}\label{gamma_comp}
  Let $0<T<1$. Then for $\epsilon>0$ sufficiently small,
  $\gamma_{V,\epsilon}$ exists on $[0,T]$, and for $t\in ]0,T]$,
  \begin{itemize}
    \item[(a)]
      $\gamma_{V,\epsilon}(t) > \gamma_D(t)$,
    \item[(b)]
      $\sigma_{V,\epsilon}(t) < \sigma_D(t)$,
    \item[(c)]
      $|\gamma_{V,\epsilon}(t) - \gamma_{D,\epsilon}(t)| +
      |\sigma_{V,\epsilon}(t) - \sigma_{D,\epsilon}(t)|\leq C \epsilon$
      where $C>0$ is independent of $\epsilon$ or $t$.
  \end{itemize}
\end{lemma}
\begin{proof}
  By definition, $\gamma_D(0) = 1 = \gamma_\epsilon(0)$, and
  \begin{align}\label{gammadot0}
    \dot\gamma_{V,\epsilon}(0)
    &=
    - \frac{2}{3}\sqrt{(1 - \epsilon^{1/5})
    4\pi \int_0^\infty \sqrt{1 + \epsilon u^2} H(u^2)\, u^2 du}\notag\\
    &>
    - \frac{2}{3} \sqrt{(1 - \epsilon^{1/5}) (1 + \epsilon)}
    >
    - \frac{2}{3} = \dot\gamma_D(0),
  \end{align}
  provided  $\epsilon>0$ is sufficiently small.
  Hence there exists $t^\ast \in ]0,1]$ such that
  $\gamma_{V,\epsilon} > \gamma_D$ on $]0,t^\ast[$, in particular,
  $\gamma_{V,\epsilon}$ exists on this interval which we choose maximal.
  On the interval $]0,t^\ast[ \cap [0,T]$,
  \[
  \gamma_{V,\epsilon} (t) > \gamma_D(t) \geq \gamma_D(T),
  \]
  and for arguments $\xi \geq \gamma_D(T)$,
  \begin{align*}    
    &
    \frac{\phi_{\epsilon}(\xi) +
      3 \psi_{\epsilon}(\xi)}{\xi^3}\\
    &\quad =
    \frac{4\pi \alpha_\epsilon}{\xi^3}
    \left(\int_0^\infty \sqrt{\xi^2 + \epsilon u^2} H(u^2)\, u^2 du
    + \epsilon
    \int_0^\infty\frac{u^2}{\sqrt{\xi^2 + \epsilon u^2}} H(u^2)\, u^2 du
    \right) \\
    &\quad =
    \frac{4\pi \alpha_\epsilon}{\xi^2}
    \int_0^1\left(\sqrt{1 + \epsilon u^2/\xi^2}+\epsilon
    \frac{u^2}{\xi^2\sqrt{1 + \epsilon u^2/\xi^2}}\right) H(u^2)\, u^2 du\\
    &\quad <
    \frac{\alpha_\epsilon}{\xi^2}
    \left(1 + \frac{2 \epsilon}{\xi^2}\right)
    \leq
    \frac{\alpha_\epsilon}{\xi^2}
    \left(1 + \frac{2 \epsilon}{\gamma_D ^2(T)}\right)
    =
    \frac{1}{6\pi} \frac{1}{\xi^2}
    (1 - \epsilon^{1/5})\left(1 + \frac{2 \epsilon}{\gamma_D ^2(T)}\right)\\
    &\quad <
    \frac{1}{6\pi} \frac{1}{\xi^2},
  \end{align*}
  provided $\epsilon>0$ is sufficiently small.
  Hence on $]0,t^\ast[ \cap ]0,T]$,
  \begin{align}\label{ddotgamma}
  \ddot\gamma_{V,\epsilon}(t)
  = -\frac{4\pi}{3}\frac{\phi_{\epsilon}(\gamma_{V,\epsilon}(t)) +
    3 \psi_{\epsilon}(\gamma_{V,\epsilon}(t))}{\gamma^3_{V,\epsilon}(t)}
  \geq -\frac{2}{9} \frac{1}{\gamma^2_{V,\epsilon}(t)}
  > -\frac{2}{9} \frac{1}{\gamma^2_D(t)} = \ddot\gamma_D(t),
  \end{align}
  which implies that $t^\ast > T$ and proves part (a).

  As to part (b) we recall that
  \[
  \sigma_{V,\epsilon} = -\frac{\dot\gamma_{V,\epsilon}}{\gamma_{V,\epsilon}},\quad
  \sigma_D = -\frac{\dot\gamma_D}{\gamma_D}.
  \]
  Hence
  \begin{align}\label{sigmaD_eps}
  \sigma_D - \sigma_{V,\epsilon} =
  \frac{\dot\gamma_{V,\epsilon}\gamma_D - \dot\gamma_D \gamma_{V,\epsilon}}
       {\gamma_D\gamma_{V,\epsilon}}.
  \end{align}
  This quantity, in particular its numerator, is positive
  at $t=0$, since $\gamma_D(0)=1=\gamma_{V,\epsilon}(0)$ and
  $\dot\gamma_{V,\epsilon}(0)>\dot\gamma_D(0)$ by \eqref{gammadot0}.
  But by \eqref{ddotgamma} and since we already know that
  $\gamma_{V,\epsilon}>\gamma_D$ on $]0,T]$,
  \begin{align*}
    \frac{d}{dt}
    \left(\dot\gamma_{V,\epsilon}\gamma_D - \dot\gamma_D \gamma_{V,\epsilon}\right)
    & =
    \ddot\gamma_{V,\epsilon}\gamma_D - \ddot\gamma_D \gamma_{V,\epsilon}\\
    & >
    -\frac{2}{9} \frac{1}{\gamma^2_{V,\epsilon}} \gamma_D
    + \frac{2}{9} \frac{1}{\gamma^2_D} \gamma_{V,\epsilon}
    \geq 0
  \end{align*}
  which implies that the quantity in \eqref{sigmaD_eps} remains positive
  on $]0,T]$.

  Part (c) is simply a result on continuous dependence on parameters
  and initial data, but since we need the order in $\epsilon$ we carry
  out the argument. First we note that for $t\in [0,T]$,
  \[
  0< \gamma_D(T)\leq \gamma_{V,\epsilon}(t), \gamma_{D,\epsilon}(t)\leq 1 .
  \]
  For arguments $\xi \in [\gamma_D(T),1]$ it is straight forward to
  see that
  \[
  -\frac{4\pi}{3}\alpha_\epsilon \frac{1}{\xi^2} - C\epsilon
  \leq -\frac{4\pi}{3}\frac{\phi_{\epsilon}(\xi) + 3 \psi_{\epsilon}(\xi)}{\xi^3}
  \leq -\frac{4\pi}{3}\alpha_\epsilon \frac{1}{\xi^2},
  \]
  where $\epsilon>0$ is small and $C>0$ is independent of $\epsilon$ or $\xi$.
  In addition,
  \begin{align*}
    \dot\gamma_{D,\epsilon}(0)
    &=
    \sqrt{\frac{8\pi}{3}\alpha_\epsilon}
    \leq
    \sqrt{\frac{8\pi}{3}\alpha_\epsilon 4\pi\int_0^1\sqrt{1+\epsilon u^2} H(u^2)
      u^2 du} =  \dot\gamma_{V,\epsilon}(0)\\
    &\leq
    \sqrt{\frac{8\pi}{3}\alpha_\epsilon} + C \epsilon.
  \end{align*}
  Hence the differential equations for $\gamma_{D,\epsilon}$
  and $\gamma_{V,\epsilon}$ imply that
  \begin{align*}
    |\gamma_{V,\epsilon}(t) -\gamma_{D,\epsilon}(t)|
    &\leq
    \int_0^t\left(C \epsilon + \frac{4\pi}{3}\alpha_\epsilon
    \int_0^s \left|\frac{1}{\gamma^2_{V,\epsilon}(\tau)}-
    \frac{1}{\gamma^2_{D,\epsilon}(\tau)}\right|\,d\tau\right)\,ds\\
    &\leq
    C \epsilon + C \int_0^t |\gamma_{V,\epsilon}(s) -\gamma_{D,\epsilon}(s)|\,ds,
  \end{align*}
  and Gronwall's lemma shows that
  \[
  |\gamma_{V,\epsilon}(t) -\gamma_{D,\epsilon}(t)|
  + |\dot \gamma_{V,\epsilon}(t) -\dot\gamma_{D,\epsilon}(t)|\leq C \epsilon.
  \]
  Since $\sigma= - \dot\gamma/\gamma$ in both the Vlasov and the dust case
  the remaining estimate in part (c) follows.
\end{proof}
\section{Oppenheimer-Snyder type collapse with Vlasov matter}
\label{section_OScollapse}
\setcounter{equation}{0}
In what follows we consider solutions of the Einstein-Vlasov
system which are launched by initial data of the form
\begin{align} \label{fdata_ndust}
  \mathring{f} (x,v) = \epsilon^{-3/2} H(|v|^2/\epsilon)
  \rho_0 (x),
\end{align}
where $H$ is specified in \eqref{Hdef},
$\rho_0\in C^1_c (\R^3)$ is spherically symmetric,
non-negative, non-increasing as a function of $r=|x|$,
$\rho_0 (r)=0$ for $r\geq 1+\epsilon$, 
\[
\rho_0(r) = \alpha_\epsilon =\frac{1}{6\pi}(1-\epsilon^{1/5})
\ \mbox{for}\ 0\leq r \leq 1,
\]
and $\epsilon \in ]0,1]$ is a small parameter such that $\alpha_\epsilon>0$. 
In addition,
$\mathring{a}=1$, and $\mathring{\beta}$ is given by \eqref{beta0def}.
In passing we note that $\rho_0$ is not the mass-energy density
induced by $\mathring f$ which as before will be denoted by $\mathring \rho$:
By \eqref{Hdef},
\begin{align} \label{rho0s}
  \rho_0(r)
  & \leq
  \mathring\rho (r)
  =
  \rho_0(r)
  \epsilon^{-3/2} \int \sqrt{1+|v|^2} H(|v|^2/\epsilon)\,dv\notag\\
  & =
  \rho_0(r)
  \int \sqrt{1+\epsilon |v|^2} H(|v|^2)\,dv
  \leq
  (1+\epsilon)\rho_0 (r).
\end{align}
Notice also that dependence on $\epsilon$ will not always be marked
by the corresponding subscript, but only in cases where this is necessary
in order to avoid confusion.

We want to show that data as above lead to the formation of a black hole,
more specifically, we want to prove the following theorem
which is our main result and the precise version of
the ``intuitive'' Theorem~\ref{maini}, stated in the introduction.
\begin{theorem}\label{maine}
  For $\epsilon>0$ sufficiently small the solution to the Einstein-Vlasov system
  launched by the data specified above exists on the time interval
  $[0,\frac{8}{9}]$.
  There exist times $0<t_0 < t_1 < t_2 < \frac{8}{9}$ such that 
  there is no trapped surface for $t<t_0$ while the surfaces given by
  $(t,r)$ with $t\geq t_1$ and $r=\frac{1}{3}$ are trapped. For $t\geq t_2$
  all the matter is contained within the 
  radius $r=\frac{1}{3}$. Let $M=\int\mathring\rho$. Then for
  $t\geq t_2$  and $r>\frac{1}{3}$, $1-A(t,r) = \frac{2M}{r}$;
  a black hole with Schwarzschild radius
  $2 M = \frac{4}{9} + \mathrm{O}(\epsilon^{1/5})$ forms.
\end{theorem}
We recall the definition of $A$ in \eqref{Adef} and note that
a surface with coordinates $(t,r)$
is trapped iff $A(t,r)<0$.
It should be observed that
by the time $t=\frac{7}{9}$ the asymptotically flat
homogeneous dust solution has formed a
trapped surface at $r=\frac{1}{3}$ and for $t$ close to $\frac{8}{9}$
all its mass is contained within this trapped surface. This motivates
our specific choice of various parameters, but we {\em emphasize} that
they could also have been kept more general.

In order to speak about a black hole it is strictly speaking
not sufficient to consider the solution only up to time $t=\frac{8}{9}$.
However, the following additional assertions hold.
\begin{proposition}\label{bh}
  For $\epsilon$ sufficiently small
  the maximal Cauchy development of the given initial data
  contains the region
  \[
  O = \left\{(t,x) \mid t > t_2,\ |x| > \frac{1}{3}\right\}
  \]
  where the solution is vacuum and
  \[
  a=1,\ \beta(t,r) = \sqrt{\frac{2 M}{r}}.
  \]
  The line $\{t\geq t_2, r=2 M\}$ is part of the radially outgoing
  null geodesic which passes through the point $(t_2,2 M)$. This
  null geodesic is future complete and generates the event horizon
  of a black hole of mass $M$.
\end{proposition}
The proof of the theorem (and the proposition) is a fairly
involved bootstrap argument which aims to compare
the solution to the homogeneous ones introduced in Section~\ref{sec_hom}
and which we split in a sequence of lemmas. Before we start on this,
another remark seems in order to put these results into perspective.
\begin{remark}\label{contdep}
  In principle, Theorem~\ref{maine} can be viewed as a result on continuous
  dependence on initial data. However, in practice there exists no
  mathematical framework within which both the Einstein-Vlasov and the
  Einstein-dust system are well posed; the latter system must be viewed
  as a singular limit of the former system. But once a solution of
  the former system with a trapped surface is obtained, then continuous
  dependence on the Vlasov data does hold and further small perturbations
  of these data, which could for example make them slightly inhomogeneous on
  the interval $[0,1]$, would launch a solution to the Einstein-Vlasov
  system which still forms a trapped surface.
\end{remark}
\subsection{Getting the bootstrap argument started}
We consider the following bootstrap assumptions on the solution to the
Einstein-Vlasov system launched by the data specified above. The solution
is to exist on some time interval $[0,T^\ast[ \subset [0,\frac{8}{9}]$
and for $0\leq t < T^\ast$ and $r \geq 0$,
\[
\begin{array}{rc}
  \mathrm{(i)} &
  \displaystyle \partial_r \beta(t,r) < \partial_r \beta_D (t),\\
  \mathrm{(ii)} &
  \displaystyle \beta(t,r) >-1,\\
  \mathrm{(iii)} &
  \displaystyle \frac{1}{2} < a(t,r) < 2, \\
  \mathrm{(iv)} &
  \ \mbox{for all characteristics of the Vlasov equation
    starting in}\ \supp\mathring{f},\\
  & \displaystyle (1-t)^{2/3} |v(t)| \leq 2 \, |v(0)|;
\end{array}
\]
note that $\partial_r\beta_D$ depends only on $t$.
The key assumption is (i) which is quite non-trivial to recover in an
improved form. The assumptions (ii) and (iii) are needed to get the
bootstrap argument started, but they are easy to recover in an
improved form. The latter is also true for (iv) which will be used to
keep matter terms like $j$ and $p$, which vanish for dust, small.
A key step in recovering (i) will be a more sophisticated estimate
for the momentum support, which uses polar coordinates and yields
a refined estimate for $\rho$, cf.~Lemma~\ref{ref_rho_est} and
\eqref{ref_vbound} below.

To begin with, we have to show that the bootstrap assumptions hold at least
on some small time interval.

\begin{lemma}\label{ba_start}
  For $\epsilon>0$ sufficiently small the solution
  of the Einstein-Vlasov system which according to Theorem~\ref{locex}
  is launched by the data specified at the beginning of this section,
  cf.~\eqref{fdata_ndust}, exists on some time interval
  $[0,T^\ast[ \subset [0,\frac{8}{9}]$ and satisfies the bootstrap
  assumptions (i)--(iv) there. We choose $T^\ast=T^\ast_\epsilon >0$
  maximal with this property.    
\end{lemma}

\begin{proof}
  By assumption, $a(0) = \mathring a = 1$ which satisfies (iii).
  Next we recall that by \eqref{beta0def},
  \[
  \mathring{\beta} (r)
  = \sqrt{\frac{2 \mathring m(r)}{r}},
  \]
  where
  \[
  \mathring m(r) = 4\pi \int_0^r \mathring \rho(s)\, s^2 ds
  \geq 4\pi \mathring\rho (r) \int_0^r s^2 ds
  = \frac{4\pi}{3} r^3 \mathring\rho (r);
  \]
  note that $\mathring\rho$ inherits the monotonicity in $r$ from
  $\rho_0$. Thus
  \begin{align*}
    \partial_r\mathring\beta(r)
    & =
    \frac{1}{\sqrt{2}}\frac{1}{\sqrt{\mathring m(r)/r}}
    \left(4\pi r \mathring\rho(r) - \frac{\mathring m(r)}{r^2}\right)
    \leq
    \frac{1}{\sqrt{2}}\frac{1}{\sqrt{\mathring m(r)/r}}\frac{8\pi}{3}
    r \mathring\rho(r)\\
    & \leq
    \frac{1}{\sqrt{2}} \frac{8\pi}{3}
    \frac{1}{\sqrt{\frac{4\pi}{3}\mathring \rho(r)\,r^2}} r \mathring\rho(r)
    =
    \sqrt{\frac{8\pi}{3} \mathring\rho(r)}.
  \end{align*}
  By \eqref{rho0s},
  $\mathring\rho (r) \leq
    \frac{1}{6\pi} (1-\epsilon^{1/5}) (1+\epsilon)$, and hence
  for $\epsilon>0$ sufficiently small,
  \begin{align}\label{initbetadiff}
    \partial_r\mathring\beta(r)
    \leq \frac{2}{3} \sqrt{(1+\epsilon)(1-\epsilon^{1/5})}
    \leq \frac{2}{3} \left(1-\frac{1}{2}\epsilon^{1/5}\right)
    = \partial_r\mathring\beta_D(r) - \frac{1}{3}\epsilon^{1/5}.
  \end{align}
  Since obviously $\mathring\beta(r)\geq 0$ we see that (i)--(iii) hold
  initially and hence on some time interval $[0,\delta]$, provided
  $\delta>0$ is small enough. As to (iv) we note the identity
  \begin{align}\label{vsquaredot}
    \frac{d}{ds} |v|^2 =
    8 \pi r \,a\, j\, \left(\frac{x\cdot v}{r}\right)^2
    + 2 \left(\partial_r\beta-\frac{\beta}{r}\right)
    \left(\frac{x\cdot v}{r}\right)^2
    + 2 \frac{\beta}{r} |v|^2
  \end{align}
  which follows from the $v$-component of the characteristic system of the
  Vlasov equation \eqref{Vl-red}. The identity implies that on $[0,\delta]$,
  \[
  \frac{d}{ds} |v|^2 \leq C_\epsilon |v|^2,
  \]
  where $C_\epsilon>0$ depends on the solution of the Einstein-Vlasov
  system and hence on $\epsilon$, but not on $s\in [0,\delta]$. By making
  $\delta$ smaller if necessary, (iv) now holds as well.
\end{proof}
\subsection{Simple consequences of the bootstrap assumptions}
In what follows we always consider the solution of the Einstein-Vlasov
system obtained in Lemma~\ref{ba_start} which satisfies the assumptions
(i)--(iv) on $[0,T^\ast[$. We derive some consequences which in particular
will show that we recover (ii)--(iv) in an improved form.
Constants denoted by $C$ are positive,
may depend on $\gamma_D$  
but not on $T^\ast$ or $\epsilon$, and may change from line to line.
We start by establishing some bounds on the source terms and on $\beta$.
\begin{lemma}\label{jandbetaest}
  For all $0\leq t < T^\ast$,
  \begin{itemize}
  \item[(a)]
    $\|j(t)\|, \|p(t)\|, \|p_T(t)\| \leq C \epsilon^{1/2}$, 
    $\|\rho(t)\| \leq C$.   
  \item[(b)]
    $ -1 < \beta(t,r)\leq \beta_D(t,r)
    \leq C r$ for $r\geq 0$.
  \end{itemize}
\end{lemma}
\begin{proof}
  The form of the initial data in \eqref{fdata_ndust}
  implies that $|v|\leq \epsilon^{1/2}$
  for $(x,v)\in \supp \mathring{f}$, the bootstrap assumption (iv)
  implies a corresponding estimate for $(x,v)\in \supp f(t)$, and
  this implies  part (a).
  The boundary condition for $\beta$ together with the bootstrap
  assumptions (i), (ii) imply (b).
\end{proof}
Next we estimate characteristics of the Vlasov equation which
run in the support of $f$.
\begin{lemma}\label{charest}
  For any characteristic $t\mapsto (x(t),v(t))$ of the Vlasov equation
  which starts in the support of
  $\mathring{f}$ and $0\leq t < T^\ast$,
  \[
  |x(t)| \leq C
  \ \mbox{and}\ 
  (1-t)^{2/3} |v(t)| \leq
  (1+ C \epsilon^{1/4}) |v(0)| .
  \]
\end{lemma}
\begin{proof}
  The bootstrap assumption (iii) and Lemma~\ref{jandbetaest}~(b) imply that
  \[
  |\dot r(s)| \leq \frac{|x\cdot v/r|(s)}{a \p0} +
  |\beta(s,r(s))| \leq 3 + C |r(s)|
  \]
  which implies the bound on $|x(t)|$.
  If we apply this bound, the bootstrap assumption (i),
  and Lemma~\ref{jandbetaest}~(a) 
  to \eqref{vsquaredot} it follows that
  \begin{align*}
    \frac{d}{ds} |v|^2
    &\leq
    16\pi \|r j(s)\| \left(\frac{x\cdot v}{r}\right)^2
    + 2 \partial_r\beta \left(\frac{x\cdot v}{r}\right)^2
    + 2 \frac{\beta}{r}
    \left[|v|^2 - \left(\frac{x\cdot v}{r}\right)^2\right] \\
    &\leq
    C \epsilon^{1/2} \left(\frac{x\cdot v}{r}\right)^2
    + 2 \partial_r\beta_D(s) \left(\frac{x\cdot v}{r}\right)^2
    + 2 \partial_r\beta_D(s)
    \left[|v|^2 - \left(\frac{x\cdot v}{r}\right)^2\right]\\
    &\leq
    \left(\frac{4}{3} \frac{1}{1-s} + C \epsilon^{1/2}\right) |v|^2,
  \end{align*}
  which implies that
  \[
  |v(t)| \leq
  \frac{\exp(C \epsilon^{1/2})}{(1- t)^{2/3}} |v(0)|
  \leq
  \frac{1 + C \epsilon^{1/2}}{(1-t)^{2/3}} |v(0)|,
  \]
  for $\epsilon>0$ sufficiently small, and the proof is complete.
\end{proof}
Lemma~\ref{charest} shows that we get
the bootstrap assumption (iv) back in an improved form
if we choose $\epsilon$ such that
$1 + C \epsilon^{1/2} < 2$.
Next we examine the metric quantities $a$ and $\beta$ and in particular
show that we can get back the bootstrap assumption (iii) in an improved form. 
\begin{lemma} \label{simpleabeta}
  It holds that
  \begin{align*}
    a(t,r) &= 1 + \mathrm{O}(\epsilon^{1/2}),\\
    \beta^2(t,r) &= \frac{2 m(t,r)}{r} + \mathrm{O}(\epsilon^{1/2}),
  \end{align*}
  where 
  \begin{align}\label{mdef}
    m(t,r) = 4\pi \int_0^r \rho(t,s)\, s^2 ds.
  \end{align}
\end{lemma}
\begin{proof}
  We recall \eqref{1overa} which in the present situation
  takes the form
  \begin{align}\label{1overa1}
    \frac{1}{a(t,r)}=
    1 + 4\pi \int_0^t(r j)(s,R(s,t,r))\, ds,
  \end{align}
  where $R(s,t,r)$ is the characteristic defined via \eqref{charactbeta}.
  If we combine \eqref{Adef} and \eqref{Adef-red},
  \begin{align}\label{betasquare}
    \beta^2(t,r)= \frac{2 m(t,r)}{r} + \frac{1}{a^2(t,r)} - 1
    - \frac{8\pi}{r}\int_0^r (a \beta j)(t,s)\, s^2 ds .
  \end{align}
  The assertion for $a$ now follows by combining \eqref{1overa1} with the
  estimates for $r$ and $j$ from the previous lemmas. Using the same
  information together with the result for $a$ in \eqref{betasquare}
  gives the assertion for $\beta^2$.
\end{proof}

Lemma~\ref{simpleabeta} shows that we get the bootstrap assumption (iii)
back in an improved form.
\subsection{The homogeneous core}
For what follows we need to show that in a certain core region
and for $\epsilon$ sufficiently small
the solution $f$ of the Einstein-Vlasov system coincides with
the homogeneous solution $h= h_\epsilon$ introduced in
Section~\ref{subsec_homvleps}, cf.~\eqref{hepsdef}. To this end, let
$r^\ast \colon [0,1[\to ]0,\infty[$ denote the maximal solution to
the initial value problem 
\[
\dot r = - \epsilon^{1/4} - \beta_D(s,r),\ r(0)=1.
\]
This function can be computed to be
\begin{align}\label{rstar}
  r^\ast(t) = (1-3\epsilon^{1/4})(1-t)^{2/3} + 3\epsilon^{1/4} (1-t),\ t\in [0,1[;
\end{align}
we choose $\epsilon>0$ sufficiently small so that
$1-3\epsilon^{1/4} > \frac{1}{2}$.
Obviously,
\begin{align}\label{rstarest}
(1-3\epsilon^{1/4}) \gamma_D(t) \leq  r^\ast(t) \leq \gamma_D(t),\ t\in [0,1[,
\end{align}
which is interesting in view of the fact that $\gamma_D$ defines the
boundary of the matter support of the Oppenheimer-Snyder dust solution
given by \eqref{OS_massf}. In particular,
\begin{align}\label{lowr}
  r^\ast(t) \geq \frac{1}{2} \gamma_D(8/9) =: r_\ast >0 \ \mbox{for}\
  t\in [0,8/9].
\end{align}
\begin{lemma}\label{homcore}
  Let
  \[
  I \coloneqq \{(t,x)\in [0,T^\ast[ \times \R^3 \mid 0\leq |x|\leq r^\ast(t)\}.
  \]
  Then for $\epsilon>0$ sufficiently small,
  $f=h$ on $I\times \R^3$ with corresponding identities
  for the metric coefficients; note that both solutions exist on
  $[0,T^\ast[$.
\end{lemma}

\begin{proof}
  Let
  $s\mapsto (x(s),v(s))$ denote a characteristic of $f$
  such that
  \[
  |x(t)| = r^\ast(t)\ \mbox{and}\ |x(s)| < r^\ast(s)\ \mbox{for}\ s>t,
  \]
  so this characteristic is running in the
  interior of the region $I$ and hits its outer boundary
  at time $t$ when followed backwards in time.
  We claim that both $f(t,x(t),v(t))=0$ and
  $h(t,x(t),v(t))=0$, provided $\epsilon>0$ is sufficiently small.
  To see the assertion for $f$, we note that
  \[
  \frac{x(t)\cdot v(t)}{r(t)\, a(t,r(t))\, \langle v(t)\rangle} - \beta(t,r(t))
  = \dot r(t) \leq \dot r^\ast(t) = - \epsilon^{1/4} - \beta_D(t,r^\ast(t)).
  \]
  Since $r(t)=r^\ast(t)$ and $\beta\leq \beta_D$ by
  Lemma~\ref{jandbetaest}~(b), this implies that
  \begin{align*}
    \epsilon^{1/4}
    &\leq
    \epsilon^{1/4} + \beta_D(t,r^\ast(t)) - \beta(t,r^\ast(t))\\
    &\leq
    - \frac{x(t)\cdot v(t)}{r^\ast(t)\, a(t,r^\ast(t))\, \langle v(t)\rangle}
    \leq 2 |v(t)|
    \leq \frac{4}{(1-8/9)^{2/3}} |v(0)|,
  \end{align*}
  where we used the bootstrap assumptions (iii) and (iv). Thus
  \[
  |v(0)| \geq \frac{1}{20} \epsilon^{1/4} > \epsilon^{1/2}
  \]
  for $\epsilon>0$ sufficiently small. Together with the fact that $f$
  is constant along characteristics and \eqref{fdata_ndust} this implies that
  $f(t,x(t),v(t))=\mathring{f}(x(0),v(0)) =0$.
  But by Lemma~\ref{gamma_comp}~(b),
  \[
  \beta_h(t,r)=\sigma_{V,\epsilon}(t) r < \sigma_{D}(t) r =\beta_D(t,r)
  \]
  for $r>0$. Hence for a characteristic of the homogeneous solution
  $h$ the above estimate for $|v(t)|$ turns into
  \[
  |v(t)| \geq  \epsilon^{1/4} >
  (1-8/9)^{-2/3} \epsilon^{1/2} \geq \gamma_D(t)^{-1}  \epsilon^{1/2}
  \geq \gamma_{V,\epsilon}(t)^{-1} \epsilon^{1/2},
  \]
  where we used Lemma~\ref{gamma_comp}~(a).
  Hence $h(t,x(t),v(t))=0$, cf.~\eqref{hepsdef}.

  In addition, any solution to
  \begin{align}\label{charbeta_fh}
  \dot r = - \beta(s,r) \ \mbox{or}\ \dot r = - \beta_h(s,r)
  \end{align}
  which at some time $s$ satisfies $r(s) < r^\ast (s)$ must satisfy
  $r(t) < r^\ast (t)$ for all $t < s$ which again follows by comparing the
  relevant $\beta$ functions. Hence no characteristic
  of the field equations \eqref{ee3} or \eqref{ee4} for either $f$ or $h$
  can leave the region $I$ when followed backwards in time.

  The above behavior of the various characteristics with respect to
  the region $I$ together with the fact that for $t=0$
  the solutions $f$ and $h$ coincide on $I\times \R^3$ imply
  that $f$ and $h$ coincide on $I$ and as long as they exist.
  To see this we proceed as follows.

  Metric quantities or macroscopic densities corresponding to the
  homogeneous solution $h$ are always denoted by the corresponding
  subscript. In particular, we note that $a_h=1$, $j_h=0$, and
  $\partial_x h=0$. Taking the difference of the Vlasov equation
  \eqref{Vl-red} for $f$ and the one for $h$ implies that
  \begin{align}\label{Vl-diff}
    \partial_t (f-h)
    &+
    \left(\frac{v}{a \p0}-\beta\frac{x}{r}\right)\cdot \partial_x (f-h)\notag \\
    & {}+ \left[4\pi r \,a\, j\, \frac{x\cdot v}{r} \frac{x}{r}
      +\left(\partial_r\beta-\frac{\beta}{r}\right)\frac{x\cdot v}{r}
      \frac{x}{r} +\frac{\beta}{r} v\right]
    \cdot\partial_v (f-h)\notag\\
    &=
    D \cdot\partial_v h ,
  \end{align}
  where
  \begin{align*}
    D &=\left(\partial_r\beta_h-\frac{\beta_h}{r}\right)\frac{x\cdot v}{r}
      \frac{x}{r}
      +\frac{\beta_h}{r} v \\
      &\quad{}- 4\pi r \,a\, j\, \frac{x\cdot v}{r} \frac{x}{r}
      -\left(\partial_r\beta-\frac{\beta}{r}\right)\frac{x\cdot v}{r}
      \frac{x}{r} - \frac{\beta}{r} v .
  \end{align*}
  The characteristic system for \eqref{Vl-diff} is the same one as for
  \eqref{Vl-red}, and as before we denote by $s\mapsto (X,V)(s,t,x,v)$
  its solution with  $(X,V)(t,t,x,v)=(x,v)$. Let $|x| < r^\ast(t)$
  and choose $t^\ast$ minimal and such that
  $|X(s,t,x,v)|< r^\ast(s)$ for $s\in ]t^\ast, t]$. Then it follows that
  \begin{align}\label{diff_fh}
    (f-h)(t,x,v)
    &=
    (f-h)(t^\ast,X(t^\ast,t,x,v),V(t^\ast,t,x,v))\notag \\
    &\quad {}+
    \int_{t^\ast}^t (D\cdot \partial_v h)(s,X(s,t,x,v),V(s,t,x,v))\, ds\notag \\
    &=
    \int_{t^\ast}^t (D\cdot \partial_v h)(s,X(s,t,x,v),V(s,t,x,v))\, ds.
  \end{align}
  The difference $f-h$ at time $t^\ast$ vanishes, because either
  $t^\ast>0$ in which case $|X(t^\ast,t,x,v)|=r^\ast(t^\ast)$
  and both $f$ and $h$ vanish separately at the corresponding
  point in phase space, or $t^\ast=0$ in which case
  $|X(0,t,x,v)|\leq 1$ so the characteristic starts in the
  region where the initial data for $f$ and $h$ coincide.

  Based on \eqref{diff_fh} we now aim for a Gronwall argument to show that
  $f-h=0$ on $I$. For this purpose we denote by $\|\cdot\|_t$ the sup norm
  of a function of $x$ (or of $x$ and $v$) where the sup extends only
  over $|x|=r \leq r^\ast(t)$. We consider $f-h$ only for times
  $t\leq T^{\ast\ast}<T^\ast$; on the compact interval $[0,T^{\ast\ast}]$
  both $f$ and $h$ exist and are bounded together with their metric
  components etc., and in what follows, constants denoted by
  $c$ may depend on these bounds. By adding and subtracting suitable
  terms in the right hand side of \eqref{diff_fh},
  \begin{align}\label{fh_diff}
    &\|f(t)-h(t)\|_t \leq \notag \\
    &\quad
    \int_0^t\left(\|\partial_r \beta(s)-\partial_r\beta_h(s)\|_s
    + \|a(s)-a_h(s)\|_s + \|f(s)-h(s)\|_s\right)\, ds;
  \end{align}
  the last term arises from $j=j-j_h$, and $\beta/r-\beta_h/r$ 
  was estimated by $\partial_r\beta - \partial_r\beta_h$.
  The differences of the various
  metric terms can now be handled exactly as we handled them for two
  consecutive iterates in Section~\ref{loop1conv}.
  As before, let $s\mapsto R(s,t,r)$ denote the unique solution
  of \eqref{charbeta_fh}
  with $R(t,t,r)=r$ and analogously for $R_h(s,t,r)$;
  these are the characteristics of both
  \eqref{ee-red3} and \eqref{ee-red4} for either $f$ or $h$.
  As we saw above, $R(s,t,r), R_h(s,t,r) \leq r^\ast(s)$ for $s<t$
  if $r\leq r^\ast(t)$.
  As in \eqref{R_ndiff} a direct Gronwall argument implies that
  \[
  |R(s,t,r) - R_h(s,t,r)|
  \leq c \int_0^t
  \|\beta (\tau)-\beta_h (\tau)\|_\tau d\tau,\ 0\leq s\leq t,\ r\leq r^\ast(t).
  \]
  Proceeding as for \eqref{beta_ndiff}, integration
  of \eqref{ee-red3} along characteristics, cf.~Lemma~\ref{lemmabetaeq},
  implies that
  \begin{align*}
    \|\beta(t)-\beta_h(t)\|_t
    &\leq
    c \int_0^t\Bigl(\|f(s)-h(s)\|_s + \|a(s)-a_h(s)\|_s \notag\\
    &\qquad\qquad {}+ \|\beta(s)-\beta_h(s)\|_s\Bigr)\, ds;
  \end{align*}
  note that $\mathring{\beta}(r) = \mathring{\beta}_h(r)$ for
  $r\leq 1$, and that when going backwards in time
  the characteristics $R$ and $R_h$ remain in $I$ if they
  start there at time $t$.
  Since by Lemma~\ref{lemmabetaeq},
  \[
  a(t,r)=\mathring{a}(R(0,t,r))
  - 4\pi \int_0^t \left(r a^2 j\right)(s,R(s,t,r))\, ds
  \]
  and $\mathring{a}(R(0,t,r))=1=a_h$ for $r\leq r^\ast(t)$ it follows that
  \[
  \|a(t) - a_h(t)\|_t \leq c \int_0^t \|f(s)-h(s)\|_s ds.
  \]
  In order to complete the Gronwall loop we need to also estimate
  the difference $\partial_r\beta - \partial_r\beta_h$, cf.~\eqref{fh_diff}.
  To achieve this we can proceed exactly as for the derivation
  of \eqref{drb_ndiff} and obtain the estimate
  \begin{align*}
    &\|\partial_r \beta(t) - \partial_r \beta_h(t)\|_t \\
    &\quad \leq
    c \int_0^t\left(\|f(s)-h(s)\|_s
    + \|\beta(s)-\beta_h(s)\|_s +
    \|\partial_r\beta(s)-\partial_r\beta_h(s)\|_s\right) ds.
  \end{align*}
  If we combine these estimates, Gronwall's lemma implies that
  $f=h$ on $I\times \R^3$ and for $t\leq T^{\ast\ast}$, and hence
  for $t < T^\ast$ as claimed.
\end{proof}

As a first application of the above homogeneous core result we
show that we recover (ii) in improved form.

\begin{lemma}\label{betapos}
  For $\epsilon$ sufficiently small,
  $\beta(t,r) >0$ for all $r>0$.
\end{lemma}
\begin{proof}
  For $0<r\leq r_\ast$ the assertion follows by Lemma~\ref{homcore},
  since $\beta_h(t,r) = \sigma_{V,\epsilon}(t) r > 0$. Next we combine
  \eqref{1overa1} and \eqref{betasquare} to find that
  \begin{align}\label{rbetasq}
    r \beta(t,r)^2
    & =
    2 m(t,r)
    + 4 \pi r \left(\frac{1}{a}+1\right) \int_0^t(r j)(s,R(s,t,r))\, ds\notag\\
    &
    \quad {}- 8\pi \int_0^r (a \beta j)(t,s)\, s^2 ds\notag\\
    & \geq
    2 m(t,r) - 12 \pi r \int_0^t |r j|(s,R(s,t,r))\, ds - C \epsilon^{1/2},
  \end{align}
  where we used the bounds from Lemma~\ref{jandbetaest} and
  Lemma~\ref{charest}. The estimate
  \[
  R(s,t,r) = r + \int_s^t \beta(\tau,R(\tau,t,r))\, d\tau
  \geq r -(t-s) \geq r-1,\ 0\leq s\leq t, 
  \]
  together with part (b) of the former lemma imply in addition
  that $R(s,t,r)$ is outside the support of $j$ if $r$ is too large.
  Hence we can, for $r \geq  r_\ast$ and by using Lemma~\ref{simpleabeta}, 
  continue the estimate~\eqref{rbetasq} 
  to find that
  \begin{align*}
    r \beta(t,r)^2
    &\geq
    2 m(t,r) - C \epsilon^{1/2} \geq 2 m(t,r_\ast) - C \epsilon^{1/2}\\
    & = 2 m_h(t,r_\ast) - C \epsilon^{1/2} \geq C > 0,
  \end{align*}
  provided $\epsilon$ is sufficiently small;
  $m_h(t,r_\ast)$ is bounded from below
  by a positive constant independent of $\epsilon$ and $t$
  due to \eqref{rho_h} and the fact that
  $\gamma_{V,\epsilon} \leq 1$.
  The estimate above shows that $\beta$ cannot become zero
  for $r \geq  r_\ast$, and so it stays strictly positive. 
\end{proof}

We have now recovered the bootstrap assumptions (ii)-(iv) in improved form. For
(i) this requires more work.
\subsection{A formula for controlling $\partial_r\beta$}
In order to continue the argument and in particular in order to
recover (i) in improved form we will use the following relation
for $\partial_r\beta$.
\begin{lemma}\label{betarepchar}
  For $t\in [0,T^\ast[$ and $r\geq 0$,
  \begin{align*}
    \frac{d}{dt} \partial_r \beta(t,R(t,0,r))
    & = (\partial_r\beta)^2(t,R(t,0,r))
    - \frac{d}{dt}(4\pi r a j)(t,R(t,0,r)) \\
    & \quad
    {}+\left( 4 \pi (\rho - p + 2 p_T) + (4 \pi r a j)^2\right)(t,R(t,0,r))\\
    & \quad
    {} - \left(\frac{8\pi}{r^3}
    \int_0^r (\rho - a \beta j)(t,s)\, s^2 ds\right)(t,R(t,0,r))\\
    & \quad
    {}+\left(8\pi r a j \partial_r\beta\right)(t,R(t,0,r)) .
  \end{align*}
\end{lemma}
\begin{proof}
  The field equations \eqref{ee-red3} and \eqref{ee-red4}
  together with \eqref{Adef-red} imply that
  \[
  \partial_t\beta-\beta\,\partial_r\beta =
  \frac{4\pi}{r^2}
  \int_0^r\left(\rho - a \beta j\right)(t,s)\,s^2 ds + 4\pi r p,
  \]
  and
  \[
  \partial_t a -\beta\, \partial_r a =
  - 4\pi r a^2 j .
  \]
  If we now observe that $\frac{d}{dt} R = -\beta(t,R)$ it follows that 
  \begin{align*}
    &
    \frac{d}{dt} \partial_r \beta (t,R(t,0,r))
    =\left(\partial_t\partial_r \beta - \beta \partial_r^2\beta\right)
    (t,R(t,0,r))\\
    & \qquad
    =\left(\partial_r\left(\partial_t \beta - \beta \partial_r\beta\right) +
    (\partial_r\beta)^2 \right) (t,R(t,0,r))\\
    & \qquad
    = (\partial_r\beta)^2 (t,R(t,0,r))
    - \left(\frac{8\pi}{r^3}\int_0^r (\rho - a \beta j)(t,s)\,
    s^2 ds\right)(t,R(t,0,r))\\
    & \qquad\qquad
    {} + 4\pi \left(\rho - a \beta j + p + r \partial_r p\right)(t,R(t,0,r)), 
  \end{align*}
  and
  \begin{align*}
    &
    \frac{d}{dt}\left(4\pi r a j\right) (t,R(t,0,r))\\
    & =
    \left( -4\pi \beta a j - (4\pi r a j)^2
    + 4 \pi r a (\partial_t j - \beta \partial_r j)
    \right)(t,R(t,0,r))\\
    & =
     \left(- 4\pi \beta a j - (4\pi r a j)^2\right)(t,R(t,0,r))\\
    &\
    + \left(4 \pi r a \left(\frac{2\beta}{r} j
    - \frac{1}{a}(\partial_r p +\frac{2p-2p_T}{r})
    +2 (\partial_r\beta + 4\pi r a j) j\right) \right)(t,R(t,0,r));
  \end{align*}
  the last equality is obtained by inserting the second identity from
  Lemma~\ref{dtrhodtj} with $a$, $\beta$, and $j$
  instead of $\tilde a$, $\tilde b$,
  and $\tilde\jmath$. When we combine the previous
  two identities several terms cancel 
  and the assertion follows.
\end{proof}

We use this result to prove that $\partial_r\beta$ remains bounded;
notice that this is necessary in view of the continuation criterion
in our local existence result, Theorem~\ref{locex},
and that the bootstrap assumption (i) only provides a bound from above.

\begin{lemma}\label{betaprbound}
  There exists a constant $C>0$ such that for $\epsilon>0$ sufficiently small,
  $\|\partial_r \beta(t)\| \leq C$ for $t\in [0,T^\ast[$.
\end{lemma}
\begin{proof}
  Lemma~\ref{betarepchar} and the bounds established in
  Lemma~\ref{jandbetaest}~(a)
  imply that
  \[
  \frac{d}{dt} \partial_r \beta(t,R(t,0,r)) \geq
  - \frac{d}{dt}(4\pi r a j)(t,R(t,0,r)) - C -
  |\partial_r\beta(t,R(t,0,r))|;
  \]
  notice that the term which is quadratic in $\partial_r\beta$
  has the right sign.
  Integrating this estimate implies that
  \[
  - \partial_r \beta(t,R(t,0,r)) \leq C +
  \int_0^t |\partial_r\beta(s,R(s,0,r))|\, ds
  \]
  which together with the bootstrap assumption (i) and the boundedness of
  $\partial_r\beta_D$ on $[0,T^\ast[$ implies that
  \[
  |\partial_r \beta(t,R(t,0,r))|
  \leq C + \int_0^t |\partial_r\beta(s,R(s,0,r))|\, ds.
  \]
  By Gronwall's lemma,
  $|\partial_r \beta(t,r)|=|\partial_r \beta(t,R(t,0,R(0,t,r)))| \leq C$.
\end{proof}
  
We can use the lemma above to establish a sharper bound on the spatial support
of the solution. This is not strictly necessary,
but it will yield the nice feature
that all the mass will be contained in the black
hole for $t$ sufficiently large.

\begin{lemma}\label{rsuppest}
  There exists some constant $C>0$ such that
  at time $t\in [0,T^\ast[$ the spatial
  support of the solution is contained in the interval
  $[0,r^\ast(t) + C\epsilon^{1/5}]$; cf.~\eqref{rstar} for the definition of $r^\ast(t)$.
\end{lemma}
\begin{proof}
  Let $s \mapsto (x(s),v(s))$ denote a characteristic of
  the Vlasov equation which
  runs in $\supp f$, and let $r(s)=|x(s)|$. We aim to show that
  $r(s) \leq r^\ast(s) + C\epsilon^{1/5}$ for all $s$.
  If $r(t) \leq r^\ast(t)$ at some
  time $t$ the assertion holds trivially for $s=t$.
  If $r(t) > r^\ast(t)$ at some time $t$, then we define $t^\ast\in[0,t[$
  minimal with the property that $r(s) > r^\ast(s)$ on $]t^\ast,t]$.
  Since
  \[
  \dot r(s) \leq C \epsilon^{1/2} - \beta(s,r(s)),
  \]
  we need an estimate for $\beta(s,r)$ from below for $r\geq r^\ast(s)$. 
  By Lemma~\ref{betaprbound} and the mean value theorem,
  \[
  \beta(s,r) \geq \beta(s, r^\ast(s)) - C (r- r^\ast(s))
  = \beta_h(s, r^\ast(s)) - C (r- r^\ast(s))
  \]
  for  $r\geq r^\ast(s)$. Hence on $]t^\ast,t]$,
  \begin{align*}
    \frac{d}{ds} (r-r^\ast)(s)
    &
    \leq C \epsilon^{1/4} - \beta_h(s,r^\ast(s))
    + C (r- r^\ast)(s) + \beta_D(s,r^\ast(s))\\
    &
    \leq C \epsilon^{1/5} + C (r- r^\ast)(s);
  \end{align*}
  in the last estimate we used the fact that by Lemma~\ref{gamma_comp},
  \begin{align*}
    |\beta_D(s,r^\ast(s))-\beta_h(s, r^\ast(s))|
    &
    \leq |\sigma_D(s)-\sigma_{V,\epsilon}(s)|\\
    &
    \leq |\sigma_D(s)-\sigma_{D,\epsilon}(s)| + C \epsilon \leq C \epsilon^{1/5},
  \end{align*}
  where we recall the explicit form of $\sigma$ in the dust case,
  cf.~\eqref{sigma_dust},
  and choose $\epsilon$ sufficiently small. This implies that on $]t^\ast,t]$,
  \[
  (r-r^\ast)(s) \leq (r-r^\ast)(t^\ast)
  + C \epsilon^{1/5} + C \int_{t^\ast}^s (r- r^\ast)(\tau)\, d\tau.
  \]
  If $t^\ast=0$ then $(r-r^\ast)(t^\ast)\leq\epsilon$, cf.\ the assumptions
  on the data introduced in \eqref{fdata_ndust}, if
  $t^\ast>0$ then $(r-r^\ast)(t^\ast)=0$, and in both cases Gronwall's lemma
  completes the proof.
\end{proof}

Eventually we need to recover the bootstrap assumption (i) in an improved form.
To this end we now use Lemma~\ref{betarepchar} to analyze the difference
$\partial_r \beta - \partial_r\beta_D$.

\begin{lemma}\label{betaprdifference}
  There exists a constant $C>0$ such that for all $\epsilon$ sufficiently small,
  $t\in [0,T^\ast[$, and $r\geq 0$,
  \begin{align*}
    &\partial_r\beta(t,R(t,0,r)))
    \leq
    \partial_r\beta_D(t,R(t,0,r))) - C \epsilon^{1/5} \\
    &\qquad
    +
    \int_0^t
    \exp\left(\int_s^t(\partial_r\beta + \partial_r \beta_D)(\tau,R(\tau,0,r))
    \,d\tau\right)\,
    d(s,R(s,0,r))\,ds,
  \end{align*}
  where
  \[
  d(s,r)
  = \left(\frac{2 m_D}{r^3}-\frac{2 m}{r^3}
  - 4\pi \rho_D + 4\pi \rho\right)(s,r).
  \]
\end{lemma}

\begin{proof}
  We derive a formula analogous to Lemma~\ref{betarepchar}
  for $\partial_r\beta_D$.
  The derivation is much simpler, since in the dust case, $j=p=p_T=0$. 
  It is important to note that we again use the characteristic $R(t,0,r)$
  belonging
  to $\beta$ instead of the corresponding dust characteristic. But since
  in addition $\partial_r^2\beta_D =0$ this causes no problem:
  \begin{align*}
    &
    \frac{d}{dt} \partial_r \beta_D (t,R(t,0,r))
    =\left(\partial_t\partial_r \beta_D - \beta \partial_r^2\beta_D\right)
    (t,R(t,0,r))\\
    & \qquad
    =\left(\partial_r\left(\partial_t \beta_D
    - \beta_D \partial_r\beta_D\right) +
    (\partial_r\beta_D)^2 \right) (t,R(t,0,r))\\
    & \qquad
    = (\partial_r\beta_D)^2 (t,R(t,0,r))
    - \left(\frac{8\pi}{r^3}\int_0^r \rho_D (t,s)\, s^2 ds
    - 4\pi \rho_D\right)(t,R(t,0,r)). 
  \end{align*}
  If we combine this with the formula from Lemma~\ref{betarepchar}
  it follows that
   \begin{align}\label{diffeqbetadiff}
    &
     \frac{d}{dt} (\partial_r \beta-\partial_r \beta_D +4\pi r a j)
     (t,R(t,0,r)) =
      d(t,R(t,0,r)) + e(t)\notag \\
     &\qquad
     {} + (\partial_r\beta + \partial_r \beta_D)
     (\partial_r\beta - \partial_r \beta_D +4\pi r a j)(t,R(t,0,r)),
   \end{align}
   where
   \begin{align*}
     e(t)
     & =   
     \left( 4 \pi (- p + 2 p_T) + (4 \pi r a j)^2 +
     4\pi r a j (\partial_r\beta-\partial_r\beta_D)\right)(t,R(t,0,r))\\
     & \quad
     {}+\left(\frac{8\pi}{r^3}\int_0^r (a \beta j)(t,s)\, s^2 ds\right)
     (t,R(t,0,r)),
   \end{align*}
   in particular, $|e(t)| \leq C \epsilon^{1/2}$,
   cf.\ Lemma~\ref{jandbetaest}~(a)
   and Lemma~\ref{betaprbound}.
   We solve \eqref{diffeqbetadiff} using variation of
   constants to find that
   \begin{align*}
    &
     (\partial_r \beta-\partial_r \beta_D +4\pi r a j) (t,R(t,0,r)) \\
     &
     = \exp\left(\int_0^t(\partial_r\beta + \partial_r \beta_D)(s,R(s,0,r))
     \, ds\right)
     (\partial_r\mathring\beta  - \partial_r \mathring\beta_D)(r)\\
     &\quad
     {}+ \int_0^t 
     \exp\left(\int_s^t(\partial_r\beta + \partial_r \beta_D)(\tau,R(\tau,0,r))
     \,d\tau\right)
     \left(d(s,R(s,0,r)) + e(s)\right)\,ds;
   \end{align*}
   note that $j$ vanishes initially.
   By Lemma~\ref{betaprbound} the argument of the
   exponential function is bounded,
   $e(s)$ and the term with $j$ on the left hand side are bounded by
   $C \epsilon^{1/2}$, but the initial data term is smaller that
   $-\frac{1}{3}\epsilon^{1/5}$,
   cf.~\eqref{initbetadiff}.
   Choosing $\epsilon$ small enough completes the proof.
\end{proof}

Notice that by replacing $r$ by $R(0,t,\tilde r)$ so that
$R(t,0,r)=\tilde r$ and
$R(\tau,0,r)=R(\tau,t,\tilde r)$ we obtain, after dropping the
tildes, an analogous formula
with $r$ instead of $R(t,0,r)$ on the left hand side and the
corresponding substitution
on the right hand side.
We will see below that the term $d$ in Lemma~\ref{betaprdifference} is negative
so that the bootstrap assumption can be recovered in improved form, but in order
to analyze $d$ we need information on the evolution of $\rho$ and $m$.
\subsection{The evolution of $\rho$ and $m$}
It turns out that in order to obtain the required information on $\rho$ and $m$
we must start with the latter quantity.
\begin{lemma}\label{massalongchar}
  For $t\in [0,T^\ast[$ and $r\geq 0$,
  \[
  m(t,R(t,0,r)) = \mathring{m}(r) + \mathrm{O}(\epsilon^{1/2}).
  \]
\end{lemma}
\begin{proof}
  Using Lemma~\ref{dtrhodtj}, integration by parts, Lemma~\ref{jandbetaest},
  and the fact that by Lemma~\ref{betaprbound}
  $\partial_r\beta$ and $\beta/r$
  are bounded we find that
  \begin{align}\label{pdem}
    &
    \frac{1}{4\pi} \left(\partial_t -\beta\partial_r\right) m = - r^2 \beta \rho
    \nonumber\\
    &
    \quad {}+
    \int_0^r\left[\beta \partial_r\rho
      -\frac{1}{a}\partial_r j - \frac{2}{s a} j
      + \frac{2\beta}{s}(\rho + p_T)
      + (\partial_r \beta + 4\pi s a j)(\rho + p)\right]
    \, s^2 ds
    \nonumber\\
    &
    =
    - \int_0^r\frac{\partial_r a}{a^2} j\, s^2 ds
    + \int_0^r \left[\frac{2\beta}{s} p_T + \partial_r \beta \, p
      + 4\pi s a j(\rho + p)\right]\, s^2 ds
    - \frac{r^2}{a} j
    \nonumber\\
    &
    =
    - \int_0^r\frac{\partial_r a}{a^2} j\, s^2 ds + \mathrm{O}(\epsilon^{1/2}).
  \end{align}
  The latter integral is $\mathrm{O}(\epsilon^{1/2})$ as well, provided
  $\partial_r a/a^2 = \mathrm{O}(1)$.
  Differentiating \eqref{1overa1} yields
  \begin{align}\label{dra}
    \frac{\partial_r a}{a^2}
    = -4\pi \int_0^t \partial_r(r j)(s,R(s,t,r)) \partial_r R(s,t,r)\, ds.
  \end{align}
  The argument which follows is similar to the one in
  Section~\ref{beta_nsection},
  following \eqref{pr}. By \cite[Lemma~6.6]{GB},
  \begin{align} \label{drrj}
    \partial_r(r j)
    &=
    j + r \int v\cdot \partial_x f dv -
    r \int \left[\frac{|v|^2}{r}\frac{x}{r}
    - \frac{x\cdot v}{r}\frac{v}{r}\right]\cdot
    \partial_v f dv\nonumber \\
    &=
    r  \int v\cdot \partial_x f dv - j.
  \end{align}
  The Cartesian version $X(s,t,x)$ of the $\beta$-characteristic $R(s,t,r)$
  is the solution to
  \begin{align}\label{X_beta}
    \dot x = - \beta(s,r)\frac{x}{r}
  \end{align}
  with $X(t,t,x) = x$;
  the right hand side of this differential equation
  is in $C^1([0,T^\ast[\times \R^3)$ and vanishes at the center $x=0$. In
  particular, the center is a characteristic which no other characteristic can cross.
  Clearly, $|X(s,t,x)|= R(s,t,r)$ where $r=|x|$.
  Moreover,
  \[
  \frac{X(s,t,x)}{|X(s,t,x)|} = \frac{x}{r}
  \]
  for $x\neq 0$; $\beta$-characteristics are strictly radial.
  Let
  \[
  D_t f(t,x,v) \coloneqq \partial_t f(t,x,v) -
  \beta(t,r) \frac{x}{r}\cdot \partial_x f (t,x,v).
  \]
  In what follows we often abbreviate
  $X(s) = X(s,t,x)$ and $R(s)=R(s,t,r)$.
  The chain rule and the Vlasov equation \eqref{Vl-red} imply that
  \begin{align*}
    \frac{d}{ds} f(s,X(s),v)
    &=
    (D_t f)(s,X(s),v) =
    -\left(\frac{v}{a \p0}\cdot \partial_x f\right)(s,X(s),v)\\
    &\quad {} + \left(F_{2} \cdot \partial_v f\right)(s,X(s),v),
  \end{align*}
  where we recall the definition of $F_2$ in \eqref{F2-def}. Hence
  \[
  \left(v \cdot \partial_x f\right)(s,X(s),v) =
  - \left(a \p0 D_t f\right)(s,X(s),v)
  + \left(a \p0 F_{2} \cdot \partial_v f\right)(s,X(s),v).
  \]
  Integration with respect to $v$ and integration by parts imply that
  \begin{align*}
    &
    \left(\int v \cdot \partial_x f dv\right) (s,R(s)) =
    - a(s,R(s)) \frac{d}{ds} \rho (s,R(s))\\
    &
    {}- 
    \left[4\pi \,a^2\, j\, (p+\rho)
      + a \partial_r\beta (p+\rho)
      + a \frac{\beta}{r} (2\rho-p)
      + a\frac{\beta}{r}\int \frac{|v|^2}{\p0}f dv \right] (s,R(s))\\
    &
    =- a(s,R(s)) \frac{d}{ds} \rho (s,R(s)) + \mathrm{O}(1).
  \end{align*}
  In view of \eqref{dra} and \eqref{drrj} we need to show that the following
  integrals are $\mathrm{O}(1)$:
  \[
  \int_0^t j(s,R(s))\,\partial_r R(s)\, ds,
  \]
  and
  \begin{align*}
    &
    \int_0^t R(s) a(s,R(s)) \frac{d}{ds} \rho (s,R(s))\, \partial_r R(s)\, ds\\
    &\qquad\quad =
    r a(t,r) \rho(t,r) - R(0,t,r)
    \mathring{\rho}(R(0,t,r))\,\partial_r R(0,t,r)\\
    &
    \qquad\quad\quad{}
    - \int_0^t \biggl[\dot R(s) a(s,R(s)) \rho (s,R(s))\, \partial_r R(s)\\
    &
    \qquad\qquad\qquad\quad 
    + R(s) \frac{d}{ds} a(s,R(s)) \rho (s,R(s))\, \partial_r R(s)\\
    &
    \qquad\qquad\qquad\quad 
    + R(s) a(s,R(s)) \rho (s,R(s))\, \partial_r \dot R(s)\biggr]\, ds .
  \end{align*}
  By \eqref{ee4} and the analogue of \eqref{drRrep}, 
  \begin{align*}
    \frac{d}{ds} a(s,R(s))
    &=
    -4\pi (r a^2 j)(s,R(s)),\\
    \partial_r R(s,t,r)
    &=
    \exp\left(\int_s^t \partial_r \beta(\tau,R(\tau,t,r))\,d\tau\right),\\
    \partial_r \dot R(s,t,r)
    &=
    - \partial_r \beta(s,R(s,t,r))
    \exp\left(\int_s^t \partial_r \beta(\tau,R(\tau,t,r))\,d\tau\right),
  \end{align*}
  all these terms are $\mathrm{O}(1)$,
  and hence the same is true for $\partial_r a/a^2$.
  Thus \eqref{pdem} turns into the equation
  \[
  \partial_t m -\beta\,\partial_r m = \mathrm{O}(\epsilon^{1/2}),
  \]
  which upon integration along characteristics gives the assertion.
\end{proof}

In order to obtain sufficiently sharp information on $\rho$ it turns out that
we need to investigate the characteristic system in coordinates which are
adapted to spherical symmetry. For $x,v\in\R^3$ with $x\neq 0$ we define
\[
r=|x|,\ w=\frac{x\cdot v}{r},\ L=|x\times v|^2.
\]
In these coordinates the characteristic system takes the form
\begin{align}
  \dot r
  &=
  \frac{w}{a(s,r) \sqrt{1+w^2 +L/r^2}} - \beta(s,r),\label{dotr}\\
  \dot w
  &=
  4\pi (r a j)(s,r)\, w + \partial_r \beta(s,r)\, w +
  \frac{L}{r^3 a(s,r) \sqrt{1+w^2 +L/r^2}},\label{dotw}\\
  \dot L
  &=
  0.\label{dotL}
\end{align}
The variable $w$ is a radial momentum variable
while $L$ is the modulus of the particle angular momentum squared
and is conserved due to spherical symmetry. In addition,
\[
|v|^2 = w^2 + \frac{L}{r^2} .
\]
The source terms can also be rewritten in these variables, in particular,
\[
\rho(t,r) = \frac{\pi}{r^2}\int_{-\infty}^\infty\int_0^\infty \sqrt{1+w^2 +\frac{L}{r^2}}
f(t,r,w,L)\, dL\,dw.
\]
We obtain the following refined estimate.
\begin{lemma}\label{ref_rho_est}
  For $\epsilon$ sufficiently small, $t\in [0,T^\ast[$ and $r\geq 0$,
  \[
  \rho(t,r)\leq \frac{\alpha_\epsilon}{\gamma_D(t) \gamma^2_{V,\epsilon}(t)}
  + C \epsilon^{1/2}.
  \]
\end{lemma}
\begin{proof}
  On the homogeneous core established in Lemma~\ref{homcore},
  \begin{align}\label{rho_h_est}
  \rho(t,r) =
  \rho_h(t) = \frac{1}{\gamma^4_{V,\epsilon}(t)}\phi_\epsilon(\gamma_{V,\epsilon}(t))
  \leq
  \frac{\alpha_\epsilon}{\gamma^3_{V,\epsilon}(t)} + C \epsilon,
  \end{align}
  where we used \eqref{rho_h} and an obvious estimate for the function
  $\phi_\epsilon$ defined in \eqref{phi_eps}.
  By Lemma~\ref{gamma_comp}~(a) the assertion follows on the homogeneous core.
  
  Hence it remains to estimate $\rho(t,r)$ for $r\geq r^\ast(t) \geq r_\ast$,
  cf.~\eqref{lowr}. Consider a characteristic $s\mapsto (x(s),v(s))$ of $f$
  in $\supp f$ with $r(t)\geq r_\ast$. Since $\beta$ is positive if follows
  from \eqref{dotr} and the bootstrap assumption (iv) that
  $\dot r \leq C \epsilon^{1/2}$ and hence $r(s)\geq \frac{1}{2}r_\ast$ on
  $[0,t]$ for $\epsilon$ sufficiently small. We will prove that on  $[0,t]$,
  \begin{align}\label{ref_vbound}
    \gamma_D^2(s) w^2(s) + \gamma^2_{V,\epsilon}(s) \frac{L}{r^2(s)}
    \leq |v(0)|^2 + \epsilon^{3/2}.
  \end{align}
  Let us for the moment assume that \eqref{ref_vbound} is already established.
  Then the estimate for $\rho(t,r)$ for the remaining case $r\geq r^\ast(t)$
  works as follows.
  
  Since $f$ is constant along characteristics and
  $\rho_0\leq \alpha_\epsilon$,
  \begin{align*}
    f(t,r,w,L)
    &= \mathring f (R(0,t,r,w,L),W(0,t,r,w,L),L)\\
    &\leq
    \epsilon^{-3/2} \alpha_\epsilon
    H\left(\frac{1}{\epsilon}\left(W^2(0,t,r,w,L)+\frac{L}{R^2(0,t,r,w,L)}
    \right)\right).
  \end{align*}
  By \eqref{ref_vbound},
  \[
  W^2(0,t,r,w,L)+\frac{L}{R^2(0,t,r,w,L)} \geq
  \gamma_D^2(t) w^2 + \gamma^2_{V,\epsilon}(t) \frac{L}{r^2}-\epsilon^{3/2},
  \]
  and since $H$ is decreasing by assumption,
  \[
  f(t,r,w,L)  \leq 
  \epsilon^{-3/2} \alpha_\epsilon
  H\left(\frac{1}{\epsilon}
  \left(\gamma_D^2(t) w^2 + \gamma^2_{V,\epsilon}(t) \frac{L}{r^2}
  -\epsilon^{3/2}\right)_+\right),
  \]
  where $(\cdot)_+$ denotes the positive part of the argument. Hence
  \begin{align*}
    \rho(t,r)
    &
    \leq 
    \frac{2 \pi \alpha_\epsilon}{r^2}\epsilon^{-3/2}
    \int_0^\infty\int_0^\infty \sqrt{1+w^2 +\frac{L}{r^2}}\\
    &
    \qquad\qquad\qquad\qquad\qquad
    H\left(\frac{1}{\epsilon}
    \left(\gamma_D^2(t) w^2 + \gamma^2_{V,\epsilon}(t) \frac{L}{r^2}
    -\epsilon^{3/2}\right)_+\right)
    \,dL\,dw\\
    &
    = 
    \frac{2 \pi \alpha_\epsilon}{\gamma^2_{V,\epsilon}(t)}\epsilon^{-3/2}
    \int_0^\infty\int_0^\infty \sqrt{1+w^2 +\frac{L}{\gamma^2_{V,\epsilon}(t)}}\\
    &
    \qquad\qquad\qquad\qquad\qquad
    H\left(\frac{1}{\epsilon}
    \left(\gamma_D^2(t) w^2 +  L -\epsilon^{3/2}\right)_+\right)
    \,dL\,dw\\
    &
    = 
    \frac{\pi \alpha_\epsilon}{\gamma^2_{V,\epsilon}(t) \gamma_D(t)}\epsilon^{-3/2}
    \int_0^\infty\int_0^\eta
    \sqrt{1+\frac{1}{\gamma_D^2(t)} (\eta - L) +\frac{L}{\gamma^2_{V,\epsilon}(t)}}
    \frac{dL}{\sqrt{\eta-L}}\\
    &
    \qquad\qquad\qquad\qquad\qquad\qquad\qquad
    H\left(\frac{1}{\epsilon}
    \left(\eta -\epsilon^{3/2}\right)_+ \right)\, d\eta;
  \end{align*}
  in the last step we changed variables via
  $(\eta,L)=(\gamma_D^2(t) w^2 +  L,L)$.
  By Lemma~\ref{gamma_comp},
  \begin{align*}
    \int_0^\eta
    \sqrt{1+\frac{1}{\gamma_D^2(t)} (\eta - L) +\frac{L}{\gamma^2_{V,\epsilon}(t)}}
    \frac{dL}{\sqrt{\eta-L}}
    &\leq
    \sqrt{1+\frac{1}{\gamma_D^2(t)} \eta} \int_0^\eta
    \frac{dL}{\sqrt{\eta-L}}\\
    &=
    2  \sqrt{\eta} \sqrt{1+\frac{1}{\gamma_D^2(t)} \eta},
  \end{align*}
  and hence
  \begin{align*}
    \rho(t,r)
    &
    \leq \frac{2 \pi \alpha_\epsilon }{\gamma^2_{V,\epsilon}(t) \gamma_D(t)}
    \epsilon^{-3/2}
    \int_0^\infty \sqrt{\eta} \sqrt{1+\frac{1}{\gamma_D^2(t)} \eta}
    H\left(\frac{1}{\epsilon} \left(\eta -\epsilon^{3/2}\right)_+ \right)
    \, d\eta\\
    &
    = \frac{2 \pi \alpha_\epsilon}{\gamma^2_{V,\epsilon}(t) \gamma_D(t)}
    \int_{-\epsilon^{1/2}}^1 \sqrt{\eta+\epsilon^{1/2}}
    \sqrt{1+\frac{1}{\gamma_D^2(t)} (\epsilon \eta+\epsilon^{3/2})}
    H\left(\eta_+ \right)\, d\eta\\
    &
    = \frac{2 \pi \alpha_\epsilon}{\gamma^2_{V,\epsilon}(t) \gamma_D(t)}
    \int_{-\epsilon^{1/2}}^0 \sqrt{\eta+\epsilon^{1/2}}
    \sqrt{1+\frac{1}{\gamma_D^2(t)} (\epsilon \eta+\epsilon^{3/2})}
    H (0)\, d\eta\\
    &
    \quad {}+ \frac{2 \pi \alpha_\epsilon}{\gamma^2_{V,\epsilon}(t) \gamma_D(t)}
    \int_0^1 \sqrt{\eta+\epsilon^{1/2}}
    \sqrt{1+\frac{1}{\gamma_D^2(t)} (\epsilon \eta+\epsilon^{3/2})}
    H(\eta)\, d\eta\\
    &
    \leq
    \frac{2 \pi \alpha_\epsilon}{\gamma^2_{V,\epsilon}(t) \gamma_D(t)}
    \int_0^1 \sqrt{\eta} H(\eta)\, d\eta + C \epsilon^{1/2}\\
    &
    = \frac{\alpha_\epsilon}{\gamma^2_{V,\epsilon}(t) \gamma_D(t)} + C \epsilon^{1/2}
  \end{align*}
  as claimed.

  It remains to prove \eqref{ref_vbound}.
  By the characteristic system \eqref{dotr}--\eqref{dotL},
  \begin{align*}
    &
    \frac{1}{2}\frac{d}{ds}\left(\gamma_D^2 w^2
    + \gamma^2_{V,\epsilon} \frac{L}{r^2}\right)\\
    &
    \quad = \gamma_D \dot\gamma_D w^2 + \gamma_D^2 w \dot w
    + \gamma_{V,\epsilon} \dot\gamma_{V,\epsilon} \frac{L}{r^2}
    - \gamma^2_{V,\epsilon} \frac{L}{r^3}\dot r \\
    &
    \quad = \gamma_D \dot\gamma_D w^2 +
    \gamma_D^2 (4\pi r a j + \partial_r\beta)\,w^2
    + \gamma_D^2 \frac{w L}{r^3 a \sqrt{1+w^2 +L/r^2}}\\
    &\qquad
    {}+ \gamma_{V,\epsilon} \dot\gamma_{V,\epsilon} \frac{L}{r^2}
    - \gamma^2_{V,\epsilon} \frac{L}{r^3}
    \left(\frac{w}{a \sqrt{1+w^2 +L/r^2}} - \beta\right)  \\
    &
    \quad \leq \gamma_D^2 4\pi r a j \,w^2  + \gamma^2_{V,\epsilon}\frac{L}{r^2}
    \left(\frac{\beta}{r}+\frac{\dot\gamma_{V,\epsilon}}{\gamma_{V,\epsilon}}\right)
    +\frac{L w}{r^3 a \sqrt{1+w^2 +L/r^2}}(\gamma_D^2- \gamma^2_{V,\epsilon}). 
  \end{align*}
  In the estimate we used the bootstrap assumption (i) and
  the explicit form of $\beta_D$ and $\gamma_D$,
  cf.~\eqref{beta_dust} and \eqref{gamma_dust}, which imply
  \[
  \gamma_D^2\partial_r\beta < \gamma_D^2\partial_r\beta_D =
  - \gamma_D \dot\gamma_D .
  \]
  Since we are looking at characteristics in $\supp f$,
  \[
  w^2, \frac{L}{r^2}\leq |v|^2 \leq C \epsilon
  \]
  and $|4\pi r a j|\leq C \epsilon^{1/2}$. In addition, we may use the lower
  bound $r=r(s)\geq \frac{1}{2}r_\ast$ as explained above.
  Thus
  \begin{align}\label{wLest}
  \frac{d}{ds}\left(\gamma_D^2 w^2 + \gamma^2_{V,\epsilon} \frac{L}{r^2}\right)
  \leq C \epsilon^{3/2} + 2\gamma^2_{V,\epsilon}\frac{L}{r^2}
  \left(\frac{\beta}{r}+\frac{\dot\gamma_{V,\epsilon}}{\gamma_{V,\epsilon}}\right).
  \end{align}
  Hence the assertion follows, provided we can suitably estimate the term
  in parenthesis on the right hand side.

  In order to do the latter we first observe that
  $\dot\gamma_{V,\epsilon}(s)/\gamma_{V,\epsilon}(s)= - \beta_h(s,r)/r$.
  In addition to the characteristic $s\mapsto (r(s),w(s),L)$ of $f$,
  for which we may assume that $r(s)\geq \frac{1}{2}r_\ast$ on
  $[0,t]$, we consider the characteristic $s\mapsto R(s,0,r) =R(s)$ of $\beta$
  where $r=|x(0)|$.
  By the bound on $\partial_r\beta$ from Lemma~\ref{betaprbound}, the mean value
  theorem, and the bootstrap assumption~(iv),
  \[
  |\dot R(s) - \dot r(s)| \leq C \epsilon^{1/2} + C |R(s) - r(s)|
  \]
  so that by Gronwall's lemma
  \begin{align}\label{rminR}
    |R(s) - r(s)| \leq C \epsilon^{1/2}.
  \end{align}
  Using Lemma~\ref{betaprbound}, Lemma~\ref{simpleabeta},
  and Lemma~\ref{massalongchar},
  \begin{align} \label{betasqest}
    \beta^2(s,r(s))
    &
    \leq \beta^2(s,R(s)) + C \epsilon^{1/2}
    \leq \frac{2 \mathring m (r)}{R(s)} + C \epsilon^{1/2} \notag \\
    &
    \leq \frac{2 \mathring m_h (r)}{R(s)} + C \epsilon^{1/2}
    = \beta_h^2 (s,R_h(s)) \frac{R_h(s)}{R(s)} + C \epsilon^{1/2}.
  \end{align}
  On the other hand,
  \begin{align*}
    \frac{d}{ds} R^{3/2}(s)
    &\geq
    - \frac{3}{2} \sqrt{2m(s,R(s))} - C \epsilon^{1/2}
    \geq
    - \frac{3}{2} \sqrt{2\mathring m(r)} - C \epsilon^{1/2}\\
    &\geq
    - \frac{3}{2} \sqrt{2\mathring m_h(r)} - C \epsilon^{1/2}
     \geq
    \frac{d}{ds} R_h^{3/2}(s) - C \epsilon^{1/2},
  \end{align*}
  and since all this happens strictly away from zero,
  \[
  R(s,0,r) \geq R_h (s,0,r)- C \epsilon^{1/2}.
  \]
  Together with \eqref{rminR} and \eqref{betasqest},
  \[
  \frac{\beta(s,r(s))}{r(s)} \leq
  \frac{\beta_h(s,R_h(s))}{r(s)} + C \epsilon^{1/2}
  \leq 
  \frac{\beta_h(s,r(s))}{r(s)} + C \epsilon^{1/2}.
  \]
  If we substitute this into the estimate \eqref{wLest} and recall the fact that
  $L\leq C \epsilon$ the assertion~\eqref{ref_vbound} follows.  
\end{proof}

We are now ready to recover the bootstrap assumption (i) in improved form.

\begin{lemma}\label{iback}
  There exists a constant $C>0$ such that for $\epsilon>0$ sufficiently small
  and for all $t\in [0,T^\ast[$ and $r\geq 0$,
  \[
  \partial_r\beta(t,r)\leq\partial_r\beta_D(t,r) - C \epsilon^{1/5}.
  \]
\end{lemma}
\begin{proof}
  In view of Lemma~\ref{betaprdifference}
  it is sufficient to show that the quantity
  $d(s,R(s,0,r))$ is negative for all relevant arguments. By \eqref{m_dust} and
  the corresponding formula for $\rho_D$,
  \[
  \frac{2 m_D(s,r)}{r^3} - 4\pi \rho_D(s,r) = \frac{4}{9}\frac{1}{(1-s)^2} -
  \frac{2}{3}\frac{1}{(1-s)^2} = -  \frac{2}{9}\frac{1}{(1-s)^2}.
  \]
  Thus if $\rho(s,R(s,0,r))=0$, then $d(s,R(s,0,r)) < 0$ as desired.

  Next we observe that if $R(s,0,r)$ lies in the homogeneous core introduced
  in Lemma~\ref{homcore}, then
  \begin{align*}
    d(s,R(s,0,r))
    &=
    \left(\frac{2 m_D}{r^3} - 4\pi \rho_D
    - \frac{2 m_h}{r^3} + 4\pi \rho_h\right)
    (s,R(s,0,r))\\
    &= \frac{4\pi}{3}
    \left(\rho_h(s) - \rho_D(s) \right)\\
    &\leq
    \frac{4\pi}{3}
    \left(\frac{\alpha_\epsilon}{\gamma^3_{V,\epsilon}(t)}
    - \frac{\alpha}{\gamma^3_{D}(t)}\right)+ C \epsilon
    <-C \epsilon^{1/5} < 0
  \end{align*}
  as desired, where we observe \eqref{rho_h} and Lemma~\ref{gamma_comp}~(a).

  It remains to consider the case $\rho(s,R(s,0,r))>0$ with $R(s,0,r)$
  outside the homogeneous core. We want to use the estimate
  \begin{align}\label{mest1}
  m(s,R(s,0,r)) \geq \mathring m(r) - C\epsilon^{1/2}
  \end{align}
  from Lemma~\ref{massalongchar} and compare the resulting term in $d$ with the
  dust contribution. However, this only works if have control on where $r$
  ranges.
  Consider a characteristic
  $s\mapsto (X,V)(s,0,\tilde x,\tilde v)$
  of the Vlasov equation with the property that
  $|X(s,0,\tilde x,\tilde v)| = R(s,0,r)$ and
  $(X,V)(s,0,\tilde x,\tilde v)\in \supp f(s)$; such a characteristic
  exists since $\rho(s,R(s,0,r))>0$. It then follows that
  $(\tilde x,\tilde v)\in \supp\mathring f$,
  in particular, $\tilde r = |\tilde x|< 1+\epsilon$.
  On the other hand we recall that
  $R(s,0,r)=|X(s,0,x,0)|$ and $V(s,0,x,0)=0$. The radial component of
  the characteristic system together with the bound for $\partial_r \beta$
  from Lemma~\ref{betaprbound} and the bootstrap assumption (iv) imply that
  for $0\leq t\leq s$,
  \[
  |R(t,0,r) - R(t,0,\tilde x,\tilde v)| \leq C \epsilon^{1/2}
  + C \int_t^s |R(\tau,0,r) - R(\tau,0,\tilde x,\tilde v)|\, d\tau,
  \]
  so that by Gronwall's lemma in particular
  $|r-\tilde r|\leq  C \epsilon^{1/2}$ which
  in view of the estimate for $\tilde r$ above implies that
  $r \leq 1 + C \epsilon^{1/2}$. The initial data are such that
  $\mathring\rho\geq\alpha_\epsilon$
  on the interval $[0,1]$ which
  implies that
  \[
  \mathring m(r) \geq \mathring m_{D,\epsilon}(r) - C\epsilon^{1/2}
  =\frac{2}{9} (1-\epsilon^{1/5}) r^3 - C\epsilon^{1/2},\
  0\leq r\leq 1+C \epsilon^{1/2}
  \]
  provided $\epsilon>0$ is sufficiently small. We can now continue
  the estimate \eqref{mest1} to find that
  \begin{align}\label{mest2}
    m(s,R(s,0,r)) \geq \frac{2}{9} (1-\epsilon^{1/5}) r^3 - C\epsilon^{1/2} .   
  \end{align}
  If we combine \eqref{mest2} with
  Lemma~\ref{ref_rho_est} it follows that for the relevant arguments,
  \begin{align*}
    d(s,R(s,0,r))
    &
    \leq
    -  \frac{2}{9}\frac{1}{(1-s)^2}
    - \frac{4}{9} (1-\epsilon^{1/5}) \frac{r^3}{R(s,0,r)^3} \\
    &
    {}+ \frac{2}{3}(1-\epsilon^{1/5})\frac{1}{\gamma_D(s) \gamma^2_{V,\epsilon}(s)}
    + C \epsilon^{1/2}.
  \end{align*}
  We need to be able to compare the various terms on the right hand side.
  First we recall that $\gamma_D(s) = (1-s)^{2/3}$, cf.~\eqref{gamma_dust}
  with $\alpha=1/6\pi$.
  If we combine Lemma~\ref{gamma_comp}~(c) with \eqref{gamma_dust}
  and $\alpha_\epsilon$ instead of $\alpha$, it follows that
  \[
  \gamma_{V,\epsilon} (s) \geq (1-\sqrt{1-\epsilon^{1/5}}s)^{2/3} - C \epsilon.
  \]
  Let us for the moment assume that for $\epsilon>0$ sufficiently small,
  \begin{align}\label{Rest}
    R(s,0,r) \leq (1-\sqrt{1-\epsilon^{1/5} - C \epsilon^{1/2}}s)^{2/3}\, r.
  \end{align}
  Then altogether we get the estimate
   \begin{align*}
    d(s,R(s,0,r))
    &
    \leq
    -  \frac{2}{9}\frac{1}{(1-s)^2}
    - \frac{4}{9} 
    \frac{1-\epsilon^{1/5}}{(1-\sqrt{1-\epsilon^{1/5} - C \epsilon^{1/2}}s)^2} \\
    & \quad 
    {}+ \frac{2}{3} \frac{1-\epsilon^{1/5}}{(1-s)^{2/3}
    (1-\sqrt{1-\epsilon^{1/5}}s)^{4/3}}
    + C \epsilon^{1/2}.
  \end{align*}
   We claim that the right hand side is negative for all $s\in [0,\frac{8}{9}]$
   and $\epsilon>0$ sufficiently small which can be seen as follows.
   For $\epsilon=0$
   the right hand side vanishes. Its derivative with respect to
   $\tilde\epsilon=\epsilon^{1/5}$ at $\epsilon=0$ equals
   \[
   \frac{4}{9}\frac{1}{(1-s)^2} + \frac{4}{9}\frac{s}{(1-s)^3}
   - \frac{2}{3}\frac{1}{(1-s)^2} - \frac{4}{9}\frac{s}{(1-s)^3}
   < - \frac{2}{9}
   \]
   which proves the assertion.

   Hence it remains to prove \eqref{Rest},
   where we recall that this is needed only
   when $R(s,0,r)$ is outside the homogeneous core and $\rho(s,R(s,0,r))>0$.
   We abbreviate $R(\tau)=R(\tau,0,r)$ for $\tau \in [0,s]$.
   Then by Lemma~\ref{simpleabeta} and Lemma~\ref{massalongchar},
   \begin{align*}
     \frac{d}{d\tau} R^{3/2}
     &=
     - \frac{3}{2} R^{1/2} \beta(\tau,R) \leq
     - \frac{3}{2}\left(2 m(\tau,R) - C \epsilon^{1/2} R\right)^{1/2}\\
     &\leq
     - \frac{3}{2}\left(2 \mathring m(r) - C \epsilon^{1/2} r^3\right)^{1/2}\\
     &\leq
     - \frac{3}{2}\left(\frac{4}{9} (1-\epsilon^{1/5}) r^3
     - C \epsilon^{1/2} r^3\right)^{1/2}\\
     &=
     -\left(1-\epsilon^{1/5} - C \epsilon^{1/2} \right)^{1/2} r^{3/2}
   \end{align*}
   and hence
   \[
   R(s)^{3/2} \leq (1-\sqrt{1-\epsilon^{1/5} - C \epsilon^{1/2}} s) r^{3/2}
   \]
   as desired; note that in these estimates
   $R(\tau)$ is bounded and bounded away from zero.
\end{proof}
\subsection{Proof of Theorem~\ref{maine} and Proposition~\ref{bh}}
\begin{proof}[Proof of Theorem~\ref{maine}]
  In order to finally prove Theorem~\ref{maine} we
  choose some $\epsilon>0$ sufficiently small so that all the above
  lemmas apply. Lemma~\ref{ba_start} tells us that the solution exists
  on some time interval $[0,T^\ast[\subset[0,\frac{8}{9}]$ and satisfies
  the bootstrap assumptions (i)--(iv) there; $T^\ast$ is chosen maximal
  with these properties.
  If we now apply Lemma~\ref{charest}, Lemma~\ref{simpleabeta},
  Lemma~\ref{betapos}, and finally Lemma~\ref{iback} we see that
  $T^\ast < \frac{8}{9}$ is only possible, if the solution blows
  up at $t=T^\ast$. But if we recall Theorem~\ref{locex} and
  Lemma~\ref{betaprbound} we see that this is not possible,
  and indeed the solution exists on $[0,\frac{8}{9}]$ and satisfies
  all the estimates there.
  
  Since $\mathring \beta(r)<\mathring \beta_D(r) < 1$ for $r< 3/2$ and
  $\mathring \beta(r)
  \leq \frac{2}{3} (1+\epsilon)^2 <1$
  for $r>1$ and $\epsilon$ small, cf.~\eqref{rho0s},
  the initial data contain no trapped surface and
  the existence of $t_0$ follows.

  Next we note that $\beta_D(t,\frac{1}{3}) >1$ for all $t\in]\frac{7}{9},1[$,
  and $\gamma_D(\frac{7}{9}) > \frac{1}{3}$ so that the asymptotically
  flat dust solution given by \eqref{OSdustbeta} with $\alpha=\frac{1}{6\pi}$
  has a trapped surface at such $t$ and $r=\frac{1}{3}$.
  Using Lemma~\ref{homcore} together
  with \eqref{rstarest} and Lemma~\ref{gamma_comp} we see that
  for some fixed $t_1\in ]\frac{7}{9},\frac{8}{9}[$ and $\epsilon$
  sufficiently small
  the surface corresponding to $t=t_1$ and $r=\frac{1}{3}$ is trapped
  for the solution of the Einstein-Vlasov system as well.
  We show that the surface with this radius stays trapped for $t>t_1$.
  To see this we cannot assume that $a=1$ and observe that in the general
  case a surface with coordinates $(t_T,r_T)$ is trapped if
  \[
  \beta(t_T,r_T) > \frac{1}{a(t_T,r_T)}.
  \]
  As we showed above, this condition holds for $t_T=t_1$ and $r_T=\frac{1}{3}$.
  It can be rewritten as
  \[
  A(t_T,r_T) =\left(\frac{1}{a}-\beta\right)
  \left(\frac{1}{a}+\beta\right)(t_T,r_T) < 0.
  \]
  But for $r=r_T=\frac{1}{3}$,
  \begin{align*}
    \partial_t (r(1-A))
    &=
    8 \pi r^2
    \beta \left( \rho + p - \left(a \beta + \frac{1}{a\beta}\right) j\right)\\
    &\geq
    8 \pi r^2 \beta \left(\rho - C \epsilon^{1/2} \rho\right) \geq 0
  \end{align*}
  by the bootstrap assumptions (iii) and (iv), the fact that
  $\beta(t,r_T)+1/\beta(t,r_T) \leq C$ via (i) and Lemma~\ref{massalongchar},
  and for $\epsilon$ sufficiently small.
  But this implies that $\partial_t A \leq 0$,
  and the trapped surface condition remains
  valid for $t>t_T$ and $r=r_T$.
  The assertion on the matter support for large times
  follows from Lemma~\ref{rsuppest} together with \eqref{rstarest}
  and by, if necessary, making $\epsilon$ again smaller one last time.
\end{proof}

The proof of Theorem~\ref{maine} being complete, we turn to the proof
of Proposition~\ref{bh}.

\begin{proof}[Proof of Proposition~\ref{bh}]
  First we consider any characteristic
  $s\mapsto (x(s),v(s))$ with $(x(s),v(s)) \in O$ for some $s$.
  Since for $r=\frac{1}{3}$ and $s>t_2$,
  \[
  \dot r \leq \frac{|w|}{a(s,r) \sqrt{1+w^2+L/r^2}} - \beta(s,r)
  < \frac{1}{a(s,r)} - \beta(s,r) < 0,
  \]
  such a characteristic cannot hit the line $r = \frac{1}{3}$ before it hits
  $t=t_2$. Since  $f=0$ for $t=t_2$ and $r \geq \frac{1}{3}$
  and $f$ is constant along characteristics, $f=0$ on the region $O$.
  
  Next we recall that by Lemma~\ref{rsuppest}
  we have vacuum for $t\in [0,\frac{8}{9}]$
  and $r\geq r^\ast (t) + C \epsilon^{1/5}$. If we recall \eqref{rstarest}
  and the explicit form of $\gamma_D$, it follows that we have
  vacuum for $t\in [0,\frac{8}{9}]$ and
  \[
  r^{3/2} \geq 1-t + C_1 \epsilon^{1/5}
  \]
  with some constant $C_1>0$. Denote by
  $R^\ast$ the solution to the initial value problem
  \[
  \dot r = -\sqrt{\frac{2M}{r}},\ r(0)=\mathring r,
  \]
  which can be computed explicitly:
  \[
  R^\ast(t)^{3/2} = \mathring r^{3/2} - \sqrt{\frac{9}{2} M} t,
  \]
  where the parameters must be such that the right hand side is positive.
  We recall that $M=\frac{2}{9} + \mathrm{O}(\epsilon^{1/5})$ so that
  there exists a constant $C_2>0$ such that for $\epsilon$ sufficiently small,
  \[
  R^\ast(t)^{3/2} \geq \mathring r^{3/2} - t - C_2 \epsilon^{1/5}.
  \]
  We now pick $\mathring r$ such that
  \[
  \mathring r^{3/2} = 1 + (C_1+C_2)\epsilon^{1/5}.
  \]
  Then
  \[
  R^\ast(t)^{3/2} \geq 1-t + C_1 \epsilon^{1/5},\ t\in[0,\frac{8}{9}],
  \]
  and this implies that the curve $R^\ast$ runs in the vacuum region.
  The field equations \eqref{ee1} and \eqref{ee2} imply that
  for $t\in[0,\frac{8}{9}]$ and $r\geq R^\ast(t)$
  the quantity $r(1-A)$ is constant,
  and hence by \eqref{Adef-red} and the fact that initially $j$
  vanishes,
  \[
  1-A(t,r) = \frac{2 M}{r}
  \]
  with $M=\int\mathring\rho$. In addition, the field equation \eqref{ee4}
  implies that along $\beta$-characteristics starting at the initial
  hypersurface $a$ remains constant and hence equal to $1$, provided these
  characteristics do not intersect the matter support. A simple bootstrap
  argument thus shows that the curve $R^\ast$ is indeed such a
  $\beta$-characteristic, and
  \begin{align}\label{vacmetric}
    a(t,r)= 1, \beta(t,r)= \sqrt{\frac{2 M}{r}} \ \mbox{for}\
    t\in [0,\frac{8}{9}], r\geq R^\ast(t).
  \end{align}
  Since
  \[
  R^\ast (t) \leq (1-t)^{2/3} + C \epsilon^{1/5}
  \]
  it follows that for $t_2$
  close to $\frac{8}{9}$ and $\epsilon$ sufficiently small,
  $R^\ast (t_2) < \frac{1}{3}$. In view of \eqref{vacmetric}
  this implies that the assertion on the metric remains valid
  on the region $O$ as claimed.
  
  The assertion on the generator of the event horizon follows from the
  explicit form of the metric coefficients on the region $O$. 
\end{proof}


\begin{thebibliography}{50}

\bibitem{AGS}
  {\sc R. O.~Acu\~{n}a-C\'{a}rdenas, C.~Gabarrete, O.~Sarbach}, 
  An introduction to the relativistic kinetic theory on curved spacetimes. 
  {\em General Relativity and Gravitation} {\bf 54}, 23 (2022).

\bibitem{AAR1}
  {\sc E.~Ames, H.~Andr\'{e}asson, O.~Rinne}, 
  Hoop and weak cosmic censorship conjectures for
  the axisymmetric Einstein-Vlasov system. 
  {\em Phys.\ Rev.\ D }{\bf 108}, 064054 (2023).

\bibitem{AAR2}
  {\sc E.~Ames, H.~Andr\'{e}asson, O.~Rinne},  
  Dynamics of gravitational collapse in the axisymmetric
  Einstein-Vlasov system. 
  {\em Class.\ Quantum.\ Grav.}\ {\bf 38}, 105003 (2021).

\bibitem{An0}
  {\sc H.~Andr\'{e}asson},
  On static shells and the Buchdahl inequality for
  the spherically symmetric Einstein-Vlasov system.
  {\em Commun.\ Math.\ Phys.}\ {\bf 274}, 409-425 (2007).

\bibitem{An1}
  {\sc H.~Andr\'{e}asson},
  The Einstein-Vlasov system/Kinetic theory. 
  {\em Liv.\ Rev.\ Relativity } 
  {\bf 14} (1), 4 (2011).

\bibitem{An2}
  {\sc H.~Andr\'{e}asson},
  Black hole formation from a complete regular past for collisionless matter. 
  {\em Ann.\ Henri Poincar\'{e}} {\bf 13}, 1511–-1536 (2012).

\bibitem{An3}
  {\sc H.~Andr\'{e}asson,} 
  Existence of steady states of the massless Einstein-Vlasov
  system surrounding a Schwarzschild black hole. 
  {\em Ann. Henri Poincar\'{e}} {\bf 22}, 4271--4297 (2021).

\bibitem{AFT}
  {\sc H.~Andr\'{e}asson, D.~Fajman, M.~Thaller},  
  Models for self-gravitating photon shells and geons. 
  {\em Ann. Henri Poincar\'{e}} {\bf 18}, 681--705 (2017).

\bibitem{AKR08}
  {\sc H.~Andr\'{e}asson, M.~Kunze, G.~Rein},
  Global existence for the spherically symmetric Einstein-Vlasov 
  system with outgoing matter. 
  {\em Commun.\ Partial Differential Eqns.}\ 
  {\bf 33}, 656--668 (2008).

\bibitem{AKR10}
  {\sc H.~Andr\'{e}asson, M.~Kunze, G.~Rein},
  Gravitational collapse and the formation of black holes 
  for the spherically symmetric Einstein-Vlasov system
  {\em Quarterly of Appl.\ Math.}\
  {\bf LXVIII}, 17--42 (2010).

\bibitem{AKR11}
  {\sc H.~Andr\'{e}asson, M.~Kunze, G.~Rein},
  The formation of black holes in spherically symmetric gravitational collapse.
  {\em Mathematische Annalen} {\bf 350}, 683--705 (2011).

\bibitem{AKR14}
  {\sc H.~Andr\'{e}asson, M.~Kunze, G.~Rein,}
  Rotating, stationary, axially symmetric spacetimes with collisionless matter.
  {\em Commun.\ Math.\ Phys.}\ {\bf 329}, 787--808 (2014).

\bibitem{AR06}
  {\sc H.~Andr\'{e}asson, G.~Rein}, 
  A numerical investigation of the stability of steady states 
  and critical phenomena for the spherically symmetric 
  Einstein-Vlasov system.
  {\em Class.\ Quantum Grav.} {\bf 23}, 3659--3677 (2006).

\bibitem{AR10}
  {\sc H.~Andr\'{e}asson, G.~Rein},
  The asymptotic behaviour in Schwarzschild time
  of Vlasov matter in spherically
  symmetric gravitational collapse.
  {\em Math.\ Proc.\ Camb.\ Phil.\ Soc.},
  {\bf 149}, 173--188 (2010).

\bibitem{AR10b}
  {\sc H.~Andr\'{e}asson, G.~Rein},
  Formation of trapped surfaces for the spherically symmetric
  Einstein-Vlasov system.
  {\em J.\ Hyperbolic Differential Eqns.}\ {\bf 7}, 707--731 (2010).

\bibitem{AR_naked}
  {\sc H.~Andr\'{e}asson, G.~Rein},
  Instability of naked singularities of the Einstein-dust system
  under perturbation of the matter model.
  {\em In preparation}.

\bibitem{BT}
  {\sc J.~Binney, S.~Tremaine},
  \textit{Galactic Dynamics} (second edition),
  Princeton Series in Astrophysics {\bf4},
  Princeton University Press 2008.
  
\bibitem{chrnd}
  {\sc D.~Christodoulou},
  Violation of cosmic censorship
  in the gravitational collapse of a dust cloud,
  {\em Comm.\ Math.\ Phys.}\ {\bf 93}, 171--195 (1984). 

\bibitem{Chr86}
  {\sc D.~Christodoulou},
  The problem of a self-gravitating scalar field,
  {\em Comm.\ Math.\ Phys.}\ {\bf 105}, 337--361 (1986).
  
\bibitem{Chr87}
  {\sc D.~Christodoulou},
  A mathematical theory of gravitational collapse,
  {\em Comm.\ Math.\ Phys.}\  {\bf 109}, 613--647 (1987).
  
\bibitem{Chr91}
  {\sc D.~Christodoulou},
  The formation of black holes and singularities in spherically symmetric
  gravitational collapse,
  {\em Comm.\ Pure Appl.\ Math.}\  {\bf 44}, 339--373  (1991).
  
\bibitem{Chr93}
  {\sc D.~Christodoulou},
  Bounded variation solutions of the spherically symmetric
  Einstein-scalar field equations,
  {\em Comm.\ Pure Appl.\ Math.}\  {\bf 46}, 1131--1220  (1993).

\bibitem{Chr94}
  {\sc D.~Christodoulou},
  Examples of naked singularity formation in the gravitational collapse of
  a scalar field, {\em Ann. of\ Math.} (2) {\bf 140}, 607-653 (1994).
  
\bibitem{Chr99a}
  {\sc D.~Christodoulou}, The instability of naked singularities in
  the gravitational collapse of a scalar field, {\em Ann. of\ Math.} (2)
  {\bf 149}, 183-217 (1999).
  
\bibitem{Chr99}
  {\sc D.~Christodoulou},
  On the global initial value problem and the issue of singularities,
  {\em Class. Quantum Gravity}\ {\bf 16}, A23--A35 (1999).

\bibitem{FaJoSm} 
  {\sc D.~Fajman, J.~Joudioux,  J.~Smulevici},
  The Stability of the Minkowski space for the Einstein-Vlasov system,
  {\em Preprint. Available on Arxiv at: https://arxiv.org/abs/1707.06141}.

\bibitem{GC}
  {\sc R.~Gautreau, M.~Cohen},
  Gravitational collapse in a single coordinate system,
  {\em Am.~J.~Phys.}\
  {\bf 63}, 991--999 (1995).
  
\bibitem{GueStrRe} 
  {\sc S.~G\"unther, C.~Straub, G.~Rein}, 
  A Birman-Schwinger principle in General Relativity: Linearly stable shells
  of collisionless matter surrounding a black hole,
  {\em arXiv:2204.10620v1}  (2022).

\bibitem{HaLinRe} 
  {\sc M.~Had\v zi\'c,  Z.~Lin, G.~Rein},
  Stability and instability of self-gravitating
  relativistic matter distributions,
  {\em Arch.\ Ration.\ Mech.\ Anal.}\ 
  {\bf 241}, 1--89 (2021).
  
\bibitem{HaRe2013} 
  {\sc M.~Had\v zi\'c, G.~Rein},
  Stability for the spherically symmetric Einstein-Vlasov
  system---a coercivity estimate,
  {\em Math.\ Proc.\ Cambridge Philos.\ Soc.}\ 
  {\bf 155}, 529--556 (2013).
  
\bibitem{HaRe2014} 
  {\sc M.~Had\v zi\'c, G.~Rein},
  On the small redshift limit of steady states of the spherically
  symmetric Einstein-Vlasov system and their stability,
  {\em Math.\ Proc.\ Cambridge Philos.\ Soc.}\ 
  {\bf 159}, 529--546 (2015).
  
\bibitem{LiTa19} 
  {\sc H.~Lindblad, M.~Taylor},
  Global stability of Minkowski space for the Einstein-Vlasov system
  in the harmonic gauge,
  {\em Arch.\ Ration.\ Mech.\ Anal.}\ 
  {\bf 235}, 517--633 (2020).
  
\bibitem{OChop}
  {\sc I.~Olabarrieta, M.~Choptuik},
  Critical phenomena at the threshold of black hole formation
  for collisionless matter in spherical symmetry,
  {\em Phys.\ Rev.\ D},
  {\bf 65} 024007 (2002).
  
\bibitem{OS}
  {\sc J.~R.~Oppenheimer, H.~Snyder},
  On continued gravitational contraction,
  {\em Phys.\ Rev.}\ {\bf 56}, 455--459 (1939).

\bibitem{pen}
  {\sc R.~Penrose},
  Gravitational collapse and space-time singularities,
  {\em Phys.\ Rev.\ Lett.}\ {\bf 14}, 57--59 (1965).

\bibitem{RaRe} 
  {\sc T.~Ramming, G.~Rein},
  Spherically symmetric equilibria for self--grav\-itating
  kinetic or fluid models in the non-relativistic and
  relativistic case---A simple proof for finite extension,
  {\em SIAM J.\ Math.\ Anal.}\ 
  {\bf 45}, 900--914 (2013).

\bibitem{GB}
  {\sc G.~Rein},
  {\em\,\,The Vlasov-Einstein System with Surface Symmetry},
  Habilitationsschrift, M\"unchen 1995.

\bibitem{Rein94} 
  {\sc G.~Rein},
  Static solutions of the spherically symmetric Vlasov-Einstein system,
  {\em Math.\ Proc.\ Cambridge Philos.\ Soc.}\ 
  {\bf 115}, 559--570 (1994).
  
\bibitem{rein07}
  {\sc G.~Rein},
  Collisionless Kinetic Equations from Astrophysics---The Vla\-sov-Poisson
  System.
  {\em Handbook of Differential Equations, Evolutionary Equations. Vol. 3}. 
  Eds.\ C.~M.~Dafermos and E.~Feireisl, Elsevier (2007).

\bibitem{RR}
  {\sc G.~Rein, A.~D.~Rendall}, 
  Global existence of solutions of the spherically symmetric
  Vlasov-Einstein system with small initial data.
  {\em Commun.\ Math.\ Phys.}\ 
  {\bf 150}, 561--583 (1992).

\bibitem{RR00} 
  {\sc G.~Rein, A.~Rendall},
  Compact support of spherically symmetric equilibria in
  non-relativistic and relativistic galactic dynamics,
  {\em Math.\ Proc.\ Cambridge Philos.\ Soc.}\ 
  {\bf 128}, 363--380 (2000).
  
\bibitem{RRS}
  {\sc G.~Rein, A.~D.~Rendall, J.~Schaeffer},
  Critical collapse of collisionless matter: A numerical investigation.
  {\em Phys.\ Rev.\ D} 
  {\bf 58}, 044007 (1998).
  
\bibitem{RT}
 {\sc G.~Rein, L.~Taegert},
 Gravitational collapse and the Vlasov-Poisson system.
 {\em Ann.\ de l'Inst.\ H.\ Poincar\'e},
 {\bf 17}, 1415--1427 (2016).
 
\bibitem{ST94}
  {\sc J.~Smoller, B.~Temple},
  Shock-wave solutions of the Einstein equations:
  The Oppenheimer-Snyder model of
  gravitational collapse extended to the case of
  non-zero pressure.
  {\em Arch. Rational Mech. Anal.}\
  {\bf 128}, 249--297 (1994).
  
\end{thebibliography}
\end{document}